\documentclass[11pt,a4paper,oneside]{article}

\usepackage[latin1]{inputenc}
\usepackage{amsmath,amsthm}
\usepackage[english]{babel}
\usepackage{latexsym}
\usepackage{amsfonts}
\usepackage{amssymb}
\usepackage{euscript}
\usepackage{color}
\usepackage{longtable}
\usepackage{tabularx}
\usepackage{pdfsync}
\usepackage{pstricks,pst-plot,pst-node,epsfig}
\usepackage{hyperref}
\usepackage{pgfplots}
\catcode`\@=11
\def\newremark#1{\@ifnextchar[{\@orem{#1}}{\@nrem{#1}}}
\def\@nrem#1#2{%
\@ifnextchar[{\@xnrem{#1}{#2}}{\@ynrem{#1}{#2}}}
\def\@xnrem#1#2[#3]{\expandafter\@ifdefinable\csname #1\endcsname
{\@definecounter{#1}\@addtoreset{#1}{#3}%
\expandafter\xdef\csname the#1\endcsname{\expandafter\noexpand
    \csname the#3\endcsname \@remcountersep \@remcounter{#1}}%
\global\@namedef{#1}{\@rem{#1}{#2}}\global\@namedef{end#1}{\@endremark}}}
\def\@ynrem#1#2{\expandafter\@ifdefinable\csname #1\endcsname
{\@definecounter{#1}%
\expandafter\xdef\csname the#1\endcsname{\@remcounter{#1}}%
\global\@namedef{#1}{\@rem{#1}{#2}}\global\@namedef{end#1}{\@endremark}}}
\def\@orem#1[#2]#3{\expandafter\@ifdefinable\csname #1\endcsname
    {\global\@namedef{the#1}{\@nameuse{the#2}}%
\global\@namedef{#1}{\@rem{#2}{#3}}%
\global\@namedef{end#1}{\@endremark}}}
\def\@rem#1#2{\refstepcounter
      {#1}\@ifnextchar[{\@yrem{#1}{#2}}{\@xrem{#1}{#2}}}
\def\@xrem#1#2{\@beginremark{#2}{\csname the#1\endcsname}\ignorespaces}
\def\@yrem#1#2[#3]{\@opargbeginremark{#2}{\csname
         the#1\endcsname}{#3}\ignorespaces}
\def\@remcounter#1{\noexpand\arabic{#1}}
\def\@remcountersep{.}
\def\@beginremark#1#2{\rm \trivlist \item[\hskip \labelsep{\bf #1\ #2}]}
\def\@opargbeginremark#1#2#3{\rm \trivlist
        \item[\hskip \labelsep{\bf #1\ #2\ (#3)}]}
\def\@endremark{\endtrivlist}
\catcode`\@=12
\newcommand{\biindice}[3]%
{

\begin{array}[t]{c}
#1\\
{\scriptstyle #2}\\
{\scriptstyle #3}
\end{array}

}
\def\a{\alpha}
\def\b{\beta}
\def\t{\theta}
\def\d{\delta}
\def\D{\Delta}

\def\g{\gamma}
\def\s{\sigma}
\def\t{\theta}
\def\l{\lambda}
\def\p{\partial}
\def\O{\Omega}
\def\e{\varepsilon}
\def\v{\varphi}
\def\G{\Gamma}
\def\k{\kappa}
\def\o{\omega}
\def\mc{\mathcal}
\def\mf{\mathfrak}

\newcommand{\R}{\mathbb R}
\newcommand{\N}{\mathbb N}
\newcommand{\Z}{\mathbb Z}

\def\div{\mbox{\rm div}}
\def\supp{\mbox{\rm supp\,}}

\newcommand\sap{\mathrel{\Bumpeq\!\!\!\!\!\!\!\longrightarrow}}

\numberwithin{equation}{section}

\theoremstyle{definition}
\newtheorem{definition}{Definition}[section]

\theoremstyle{plain}
\newtheorem{theorem}{Theorem}[section]
\newtheorem{proposition}{Proposition}[section]
\newtheorem{lemma}{Lemma}[section]

\newtheorem{remark}{Remark}[section]

\title{\vskip-2.5cm
{Subharmonic solutions for a class of predator-prey models with degenerate weights in periodic environments}
\thanks{This paper has been written under the auspices of the Ministry of Science, Technology and Universities of Spain, under Research Grant PID2021-123343NB-I00, and of the IMI of Complutense University. The second author, ORCID 0000-0003-1184-6231, has been also supported by contract CT42/18-CT43/18 of Complutense University of Madrid.}
}
\author{
\sc Juli\'an L\'opez-G\' omez
\\
\small Universidad Complutense de Madrid
\\
\small Instituto de Matem\'{a}tica Interdisciplinar (IMI)
\\
\small Departamento de An\'alisis Matem\'atico y Matem\'atica
Aplicada
\\
\small  Plaza de las Ciencias 3, 28040   Madrid, Spain
\\
\small E-mail: {\tt   julian@mat.ucm.es}
\medskip
\\
\sc Eduardo Mu\~{n}oz-Hern\'andez
\\
\small Universidad Complutense de Madrid
\\
\small Instituto de Matem\'{a}tica Interdisciplinar (IMI)
\\
\small Departamento de An\'alisis Matem\'atico y Matem\'atica
Aplicada
\\
\small  Plaza de las Ciencias 3, 28040   Madrid, Spain
\\
\small E-mail: {\tt eduardmu@ucm.es }
\medskip
\\
\sc Fabio Zanolin
\\
\small Università degli Studi di Udine
\\
\small Dipartimento di Scienze Matematiche, Informatiche e Fisiche
\\
\small  Via delle Scienze 2016, 33100 Udine, Italy
\\
\small E-mail: {\tt fabio.zanolin@uniud.it } }

\date{\today}

\numberwithin{equation}{section}

\begin{document}

\maketitle

{
\begin{abstract}
This paper deals with the existence, multiplicity, minimal complexity and global structure of
the subharmonic solutions to a class
of planar Hamiltonian systems with periodic coefficients, being the classical
predator-prey model of V. Volterra its most paradigmatic example. By means of a topological approach based on techniques from global bifurcation theory,   the first part of the paper ascertains their nature, multiplicity and minimal complexity, as well as their global minimal structure,
in terms of the configuration of the function coefficients in the setting of the model.
The second part of the paper introduces a dynamical system approach based on the theory
of topological horseshoes that permits to detect, besides subharmonic solutions, ``chaotic-type''
solutions. As a byproduct of our analysis, the simplest predator-prey prototype models in periodic environments can provoke chaotic dynamics. This cannot occur in cooperative and quasi-cooperative dynamics, as a consequence of the ordering imposed by the maximum principle.
    \end{abstract}
}


\section{Introduction}\label{sec1}
The analysis of subharmonic solutions to differential systems with periodic coefficients is a classical
research topic which has been widely investigated also with respect to its
relevant significance in several applications, including the study of differential equations models
arising from Celestial Mechanics and Engineering.
Generally speaking, given a first-order differential system
\begin{equation}
    \label{i.1}
    z'= F(t,z),
\end{equation}
for $z=(z_1,\dots,z_d)\in \Omega$, where $\O$ is  an open domain of ${\mathbb
    R}^d$ and $F: {\mathbb R}\times\Omega\to {\mathbb R}^d$ is a
sufficiently regular vector field which is $T$-periodic in the
$t$-variable, by a \textit{subharmonic solution of order} $n\geq
2$ to system \eqref{i.1} we mean a $nT$-periodic solution of
the system which is not $kT$-periodic for all integers
$k\in\{1,\dots,n-1\}.$ As pointed out by P. H. Rabinowitz in
\cite{Ra-1980}:
\par
\bigskip
\begin{small}
    ``This latter quest is complicated by the fact
    that any $T$-periodic solution is a fortiori $kT$-periodic. Thus
    an additional argument is required to show that any subharmonics
    are indeed distinct.''
\end{small}
\bigskip
\par
\noindent In particular, it will also be important to
check whether a subharmonic solution of order $n$ has indeed $nT$
as its minimal period. This is a difficult task that nevertheless
can be overcome in some circumstances thanks to the special structure of the
vector field $F$.  Some sufficient conditions have
been already proposed in the literature (see, e.g., Michalek and Tarantello \cite{MiTa-1988}).

\par

The general aim of this work is to investigate, as deeply as possible,  the subharmonics to a class of Lotka--Volterra systems under seasonal effects. Precisely, the model is a planar Hamiltonian system of the form
\begin{equation}\label{i.2}
    \begin{cases}
        x' = -\lambda\alpha(t)f(y),\\
        y' = \lambda\beta(t)f(x),
    \end{cases}
\end{equation}
where $f:{\mathbb R}\to {\mathbb R}$ is a locally Lipschitz continuous function with
$f(0)=0$ and  $f(s)s>0$ for $s\not=0$ such that $f$ is bounded on $(-\infty,0]$
and has a superlinear growth at $+\infty$. The assumptions on $f$ are motivated by the paradigmatic case
\begin{equation}\label{i.3}
    f(s)= e^s-1,
\end{equation}
coming from the original Volterra's equations.
\par
In 1926, Vito Volterra in \cite{Vo-1926} proposed a mathematical model for the predator-prey interactions
as an answer to the statistics on the fishing data in Northern Adriatic sea, provided by the biologist Umberto D'Ancona. His model, extraordinarily famous today, since it is discussed in most of textbooks on
Ordinary Differential Equations and Ecology, has  shown to be
a milestone for the development of more realistic predator-prey models in Environmental Sciences and
Population Dynamics (see, e.g., Begon,  Harper and Townsend \cite{BHT-1990}). It can be formulated as follows
\begin{equation}
    \label{i.4}
    \begin{cases}
        N_1' = N_1(a-bN_2),\\
        N_2' = N_2(-c+dN_1),
    \end{cases}
\end{equation}
where $N_1(t)>0$ and $N_2(t)>0$ represent, respectively, the density of the prey and the predator
populations at time $t$. In \eqref{i.4}, the coefficients, $a$, $b$, $c$ and $d$, are assumed to be positive (see Braun \cite{Br-1983}). The same system had been already introduced few years before
by Alfred J. Lotka from a hypothetical chemical reaction exhibiting periodic behavior in the
chemical concentrations (see Murray \cite[\S 3.1]{Mu-2002}), which reveals how the same models can
mimic a variety of phenomenologies of a different nature. After these pioneering contributions, differential equations involving the interaction of two or more species are usually named as Lotka--Volterra systems and are represented in the general form
\begin{equation}\label{i.5}
    N'_i=N_i\bigl(a_i - \sum_{j=1}^N b_{ij}N_j\bigr), \quad i=1,\dots,n,
\end{equation}
though in this paper we will focus attention on those satisfying $b_{ii}=0$.
The choice of the sign of the coefficients $a_i$ and $b_{ij}$ allows describing
different kinds of interactions, as competition, cooperation, or parasitism.
\par
Although Volterra \cite{Vo-1931} and Lotka \cite{Lo-1998} considered the possibility
of some varying in time coefficients,  an extensive study of the Lotka--Volterra systems with periodic
coefficients has not been carried out until more recently (see, e.g.,
Butler and Freedman \cite{BuFe-1981}, Cushing \cite{Cu-1977,Cu-1980} and Rosenblat \cite{Ro-1980},  for some early works in this direction). However, in the past four decades, the researches in this area have
originated a great number of contributions, also in connection with the study of
periodically perturbed Hamiltonian systems and  Reaction-Diffusion systems arising in genetics and
population dynamics, beginning with the influential monograph of P. Hess \cite{Hess-1991} and the refinements of L\'{o}pez-G\'{o}mez \cite{LG-1992}, where a general class of spatially heterogeneous diffusive Lotka--Volterra systems with periodic coefficients was studied. To understand the role
played by the spatial heterogeneities in these models the reader is sent to the early works of L\'opez-G\'omez \cite{LG-1995} and Hutson et al. \cite{HLMV-1995}, as well as to the monographs \cite{LG-2001,LG-2015}. Further models covering more general spatial interactions, outside the Lotka--Volterra World, were introduced
by L\'opez-G\'omez and Molina-Meyer \cite{LGMJ-2006,LGMT-2006}.
\par
Although the assumption that the coefficients of the model are periodic with a common period might seem  restrictive, it is based on the natural assumption that the species are affected by identical seasonal effects, which is rather natural as they interact in the same environment. Precisely, the interest of this paper focuses into the following periodic counterpart of the autonomous system \eqref{i.4}
\begin{equation}\label{i.6}
    \begin{cases}
        N_1' = N_1(a(t)-b(t)N_2)\\
        N_2' = N_2(-c(t)+d(t)N_1)
    \end{cases}
\end{equation}
where, typically, $a$, $b$, $c$, $d:{\mathbb R}\to {\mathbb R}$ are
$T$-periodic functions such that
$$
b(t)\gneq 0,\qquad d(t)\gneq 0,\qquad \int_0^T a(t)\,dt>0,\qquad \int_0^T c(t)\,dt>0.
$$
In this case,  by \cite[Th. A.1]{LMZ-2020}, the system \eqref{i.6} has a component-wise positive $T$-periodic solution, $(\tilde{N}_1(t),\tilde{N}_2(t))$, and the change of variables
$$
N_1(t)=u(t)\tilde{N}_1(t), \quad N_2(t)=v(t)\tilde{N}_2(t),
$$
leads to the study of the equivalent system
\begin{equation}
    \label{i.7}
    \begin{cases}
        u' = \lambda\alpha(t)u(1-v)\\
        v' = \lambda\beta(t)v(-1+u)
    \end{cases}
\end{equation}
where $(\tilde{N}_1(t),\tilde{N}_2(t))$ becomes $(1,1)$ (see the Appendix of \cite{LMZ-2020} for any further details). In  \eqref{i.7}, $\lambda>0$ is regarded as a parameter, while $\alpha(t)\gneq 0$ and $\beta(t)\gneq 0$ are $T$-periodic function coefficients, where $T>0$ is their minimal period. The component-wise \textit{positive} periodic solutions of \eqref{i.7} are called (periodic) \textit{coexistence states} and, obviously, are relevant in population dynamics, as they represent states
where none of the interacting species is driven to extinction by the other.
\par
Among the periodic coexistence states, the subharmonics are of particular interest. Focussing attention
into the simplest prototype model with constant coefficients
\eqref{i.4}, it is folklore since Volterra \cite{Vo-1931} that all its
positive solutions are periodic and oscillate around the equilibrium $P_0\equiv (c/d,a/b)$ in the counterclockwise sense, lying on the ``energy levels''
$$
E(N_1,N_2):=dN_1 -c \log N_1 + bN_2 - a \log
N_2 = \text{constant} = k,
$$
for $k\in (k_0,+\infty)$, where
$$
k_0:= E(P_0) = c\left(1-\log\tfrac{c}{d}\right)  +
a\left(1-\log\tfrac{a}{b}\right).
$$
Therefore, $P_0$ is a global center in the open first quadrant. The fact that
the period  of the orbit at the level $k$, $\tau(k)$,  is an
increasing function of $k$ is a more recent finding of Rothe \cite{Ro-1985}, Shaaf \cite{Sh-1985} and
Waldvogel \cite{Wa-1986}, where it was also shown that
$$
\lim_{k\downarrow k_0}{\tau(k)} =
\frac{2\pi}{\sqrt{ac}}=:\tau_0 \quad\text{ and }\;
\lim_{k\uparrow \infty}{\tau(k)} = +\infty.
$$
Thus, for every  $T\in (0,\tau_0)$, the equilibrium point $P_{0}$
is the unique $T$-periodic coexistence state (harmonic solution),
though there are infinitely many nontrivial subharmonic
coexistence states corresponding to the energy levels $k$ for
which $\tau(k)= mT$ for sufficiently large  $m\geq 2$. This
classical example also illustrates how subharmonics can be
packaged in equivalence classes (two subharmonics are equivalent
when they are a time-shift of the other).  This also holds for
\textit{nonautonomous systems}. Indeed,  if $z(t)$ is a
$nT$-periodic solution of \eqref{i.1}, then $z(t+jT)$ is also a
$nT$-periodic solution for every $j=0,1,\dots,n-1.$ Michalek and
Tarantello \cite{MiTa-1988} referred to the set
$$
\Theta(z)=\{z(\cdot+jT):j=0,1,\dots,n-1\}
$$
as the \emph{$\Z_n$-orbit of $z$.}
Hence, searching subharmonics
in  \textit{nonautonomous} systems  is a task fraught with a
number of difficulties.
\par
The main goals of this paper are the following:
\begin{itemize}
    \item{} finding out $nT$-periodic solutions having $nT$ as \textit{minimal period},
    namely \textit{subharmonics of order $n$};
    \item{} among the subharmonics having the same order, ascertaining whether, or not, they belong to the same periodicity class.
\end{itemize}
Previous results on the existence and multiplicity of harmonic and
subharmonic solutions in periodically perturbed predator-prey
systems have been obtained by  Hausrath \cite{Ha-1982} and Liu \cite{Liu-1995}, as a
consequence of the Moser twist theorem, and by
Hausrath and Man\'asevich \cite{HM-1991}, Ding and Zanolin \cite{DZ-1993,DZ-1996} and Boscaggin  \cite{Bo-2011} from the Poincar\'{e}--Birkhoff fixed point theorem. The latest results
have been recently extended by Fonda and Toader \cite{FoTo-2019} to systems in
${\mathbb R}^{2n}$ by means of a (previous) higher-dimensional version of the
Poincar\'{e}--Birkhoff theorem due to Fonda and Ure\~{n}a
\cite{FoUr-2017}. More recently, Boscaggin and Mu\~{n}oz-Hern\'{a}ndez \cite{BoMH-2022} have also analyzed the subharmonic solutions in a class of planar Hamiltonian systems including \eqref{i.7}.
\par
L\'{o}pez-G\'{o}mez, Ortega and Tineo \cite{LOT-1996} and L\'{o}pez-G\'{o}mez \cite[\S 5]{LG-2000} carried out a thorough investigation of the positive coexistence states for the generalized Lotka--Volterra system
with periodic coefficients
\begin{equation}
    \label{i.8}
    \begin{cases}
        N_1' = N_1(a_1(t)-c_1(t)N_1-b_1(t)N_2),\\
        N_2' = N_2(a_2(t)+b_2(t)N_1-c_2(t)N_2),
    \end{cases}
\end{equation}
which includes the presence of logistic terms incorporating to the model setting some interspecific competition effects. This model has, in addition, semi-trivial coexistence states of the form
$(\hat{N}_1(t),0)$ and $(0,\hat{N}_2(t))$, where  $\hat{N}_i$ stands for the (unique) positive $T$-periodic solution of the logistic equation
$$
N_i' = N_i(a_i(t)-c_i(t)N_i),\qquad i=1, 2.
$$
As in the periodic-parabolic counterpart of \eqref{i.8}, already analyzed by Hess \cite{Hess-1991} and L\'{o}pez-G\'{o}mez \cite{LG-1992}, the local character of the semitrivial positive solutions plays a crucial role in determining the dynamics of \eqref{i.8}, being a challenging task to ascertain the stability, or instability, of the coexistence states.
\par
Ortega, L\'{o}pez-G\'{o}mez and Tineo \cite{LOT-1996} made the crucial observation that, for the special choice
\begin{equation}\label{i.9}
    \alpha(t)\equiv 0, \quad\text{for }\; t\in [T/2,T], \qquad
    \beta(t)\equiv 0, \quad\text{for }\; t\in [0,T/2],
\end{equation}
\eqref{i.7} has a (unique) linearly unstable coexistence state, though the system can admit two coexistence
states (see \cite[Rem. 7.5]{LOT-1996}). This is strong contrast with some previous one-dimensional uniqueness results available for the diffusive counterparts of these models (see the detailed discussion of \cite{LG-2000}). Actually, this prototype model has shown to be rather paradigmatic for analyzing the local character and the multiplicity of harmonic and subharmonic coexistence states. Indeed, setting
\begin{equation}\label{i.10}
    A:=\int_{0}^{T}\alpha(t)\,dt, \qquad B:=\int_{0}^{T}\beta(t)\,dt,
\end{equation}
it follows from \cite[Pr. 7.1]{LOT-1996} that $P_0\equiv (1,1)$ is linearly unstable
if
\begin{equation}
    \label{i.11}
    \lambda > \frac{2}{\sqrt{AB}}.
\end{equation}
Moreover, by \cite[Th. 5.3]{LG-2000}, the instability of $P_0$ guarantees the existence  of three
coexistence states for \eqref{i.8} within the appropriate ranges of values of its
function coefficients. It turns out that, besides the unique
(harmonic) coexistence state $P_0$,  there are, at least, two additional $2T$-periodic
coexistence states if \eqref{i.11} holds. After two decades, L\'{o}pez-G\'{o}mez and
Mu\~{n}oz-Hern\'{a}ndez \cite{LM-2020} were able to construct $nT$-periodic solutions for
every $n\geq 2$, providing simultaneously with a sharp estimate of their minimal cardinals. Figure \ref{fig-i} shows the (minimal) global bifurcation diagram of  subharmonics of \eqref{i.7} found in \cite{LM-2020} for an arbitrary choice of $\a(t)$ and $\b(t)$ satisfying the \emph{orthogonality condition} \eqref{i.9}.

\begin{figure}[h!]
    \centering
    \includegraphics[scale=0.5]{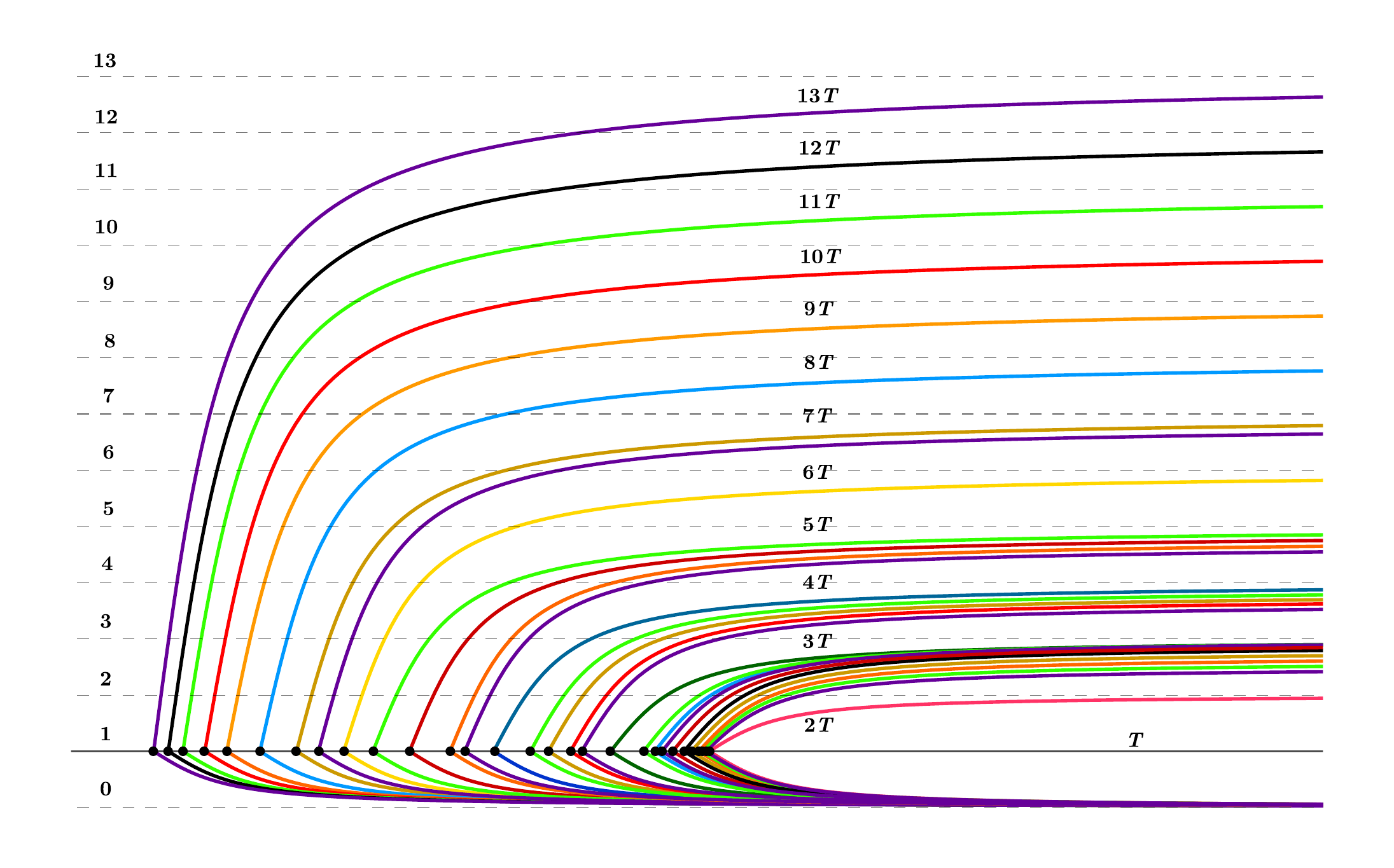}
    \caption{Admissible global bifurcation diagram under \eqref{i.9}, after \cite{LM-2020}.}
    \label{fig-i}
\end{figure}

These findings should be also true, at least,  when $\a(t)$ and $\b(t)$ are \emph{nearly orthogonal}, in the sense that the product  $\a\b$ is sufficiently small in $[0,T]$, and actually they might be also true for
very general classes of weight functions $\a(t)$ and $\b(t)$ far from satisfying \eqref{i.9}. Thus, it is rather natural to investigate the structure of the set of subharmonics of \eqref{i.7} for more general configurations of $\alpha(t)$ and $\beta(t)$ than those satisfying \eqref{i.9}. Besides its intrinsic interest, this analysis might reveal some new important features, based, e.g., on the shapes of $\alpha(t)$ and $\beta(t)$, that might be significative in population dynamics. Indeed, as already discussed by the authors in \cite{LMZ-2020}, when $\beta\equiv 0$, one is considering a seasonal time-interval where predator is still aggressive with respect to the prey but does not take advantage of the harvest, which may simulate the difference between the hunting and the gathering period. Analog interpretations may be given for the time-intervals when $\alpha\equiv 0$.
\par
From the mathematical analysis of  \eqref{i.7} it becomes apparent that its dynamics might vary according to the measure of the intersection of the supports of $\a$ and $\b$. To differentiate the two extreme cases, we will name as ``degenerate'' the case when $|\supp \alpha \cap \supp \beta|=0$, while the case when
$$
|\supp \alpha \cap \supp \beta|>0.
$$
will be referred to as the ``non-degenerate'' case. In \cite{LMZ-2020,LMZ-2021}, the authors have already shown the applicability of the Poincar\'{e}--Birkhoff theorem
in searching subharmonics in both cases (see also the recent refinements of Boscaggin and Muñoz-Hernández \cite{BoMH-2022}).
\par
As a sharp  analysis of \eqref{i.7} seems  imperative  to classify all the possible dynamics associated to a general predator-prey system with periodic coefficients, throughout this paper we will focus our attention into the simplest prototype model \eqref{i.7}, with special emphasis towards  the problem of analyzing its subharmonic solutions depending on whether, or not, the chosen configuration
of $\a(t)$ and $\b(t)$ is degenerate.
\par
Among the various existing approaches in finding out subharmonics in Hamiltonian systems, the most successful
ones are the following:
\begin{itemize}
    \item variational methods (critical point theory),
    \item Kolmogorov--Arnold--Moser (KAM) theory, Moser twist theorem, and Poincar\'{e}--Birkhoff
    theorem,
    \item symbolic dynamics associated to Smale's horseshoe-type structures, and
    \item reduction via symmetry and bifurcation theory.
\end{itemize}
As we are going to invoke the last three in this paper, we will shortly revisit them in the next few paragraphs.
\par
\bigskip

\noindent\textbf{Variational methods (critical point theory):} In this framework, a number of different arguments have been given to show that the periodic solutions that are  critical
points of the associated functional are indeed subharmonics of
large period. For instance, Rabinowitz  \cite{Ra-1980} achieves it by
analyzing the critical levels of the functional
(see also Fonda and Lazer \cite{FoLa-1992} and Serra, Tarallo and Terracini \cite{STT-2000}), whereas  Michalek and Tarantello \cite{MiTa-1988} apply a combination of estimates on critical levels and
$\mathbb{Z}_p$-index theory. Assuming an additional non-degeneracy condition on the solutions,
Conley and Zhender \cite{CoZe-1984} can get subharmonics through their Morse indices (see
also Abbondandolo \cite{Ab-2001} and the references therein).
\par
\bigskip

\noindent\textbf{KAM theory, Moser twist theorem, and Poincar\'{e}--Birkhoff
    theorem:} In this setting, a typical approach consists in constructing
an annular region enclosed by two curves, $\Gamma_{\rm int}$
and $\Gamma_{\rm out}$, which are invariant by the Poincar\'{e} map
associated to the given planar system. Recall that
the Poincar\'{e} map is the homeomorphism $\Phi=\Phi_{0}^{T}$ transforming
an initial point $z_0$ to $\zeta(T;0,z_0)$, where
$\zeta(\cdot;0,z_0)$ stands for the solution of \eqref{i.1} with $z(0)=z_0.$ The associated
flow is  area-preserving when $\div_{z}F(t,z)\equiv 0$,  a
condition which is always satisfied by Hamiltonian systems. Then, introducing  a
suitable \textit{rotation number}, ${\rm rot}([0,nT],z_0)$,
the Poincar\'{e}--Birkhoff fixed point theorem guarantees the existence of at least two fixed points for $\Phi^n$ in the interior of the annulus as soon as the following  \textit{twist
    condition}
\begin{equation}\label{i.12}
    {\rm rot}([0,nT],z_0) \begin{cases}
        < j \; [\text{resp} >j]\; \; \hbox{for all} \;\, z_0\in \Gamma_{\rm int}\\
        > j \; [\text{resp} <j]\; \; \hbox{for all}\;\, z_0\in \Gamma_{\rm out}
    \end{cases}
\end{equation}
holds for some integer $j$. In such situation, the fixed points of $\Phi^n$ have $j$ as an associated rotation number. This actually entails the existence of subharmonics of order $n$ if $n$ and $j$ are co-prime integers. Theorems 2.1 and 3.1 of Neumann  \cite{Ne-1977} give some sufficient conditions  so that these subharmonic solutions belong to different periodicity classes.
\par
A typical definition of rotation number can be  given by passing to polar coordinates and
counting the number of turns (counterclockwise, or clockwise) that the solution $\zeta(t;0,z_0)$
makes around the origin in a given time-interval. The choice of the origin as a ``pivotal point'' is merely conventional, but it fits well for systems like \eqref{i.2}, where $F(t,0)\equiv 0$. Adopting this methodology,  for any given interval $[t_0,t_1]$, we will set
\begin{equation}
    \label{i.13}
    {\rm rot}([t_0,t_1],z_0):=\delta\frac{1}{2\pi} \int_{t_0}^{t_1}
    \frac{F(t,\zeta(t;t_0,z_0))\wedge \zeta(t;t_0,z_0)}{||\zeta(t;t_0,z_0)||^2}\,dt
\end{equation}
where  $\delta=\pm1$. Boscaggin \cite{Bo-2011} and Boscaggin and Mu\~{n}oz-Hern\'{a}ndez \cite{BoMH-2022}
give some other equivalent notions.
\par
Ding \cite{Di-1982}, Fonda and Ure\~{n}a \cite{FoUr-2017}, Frank \cite{Fr-1988}, Qian and Torres
\cite{QiTo-2005} and Rebelo \cite{Re-1997} have found some refinements  of the
Poincar\'{e}--Birkhoff fixed point theorem adapted to more general settings where the boundary invariance
of the underlying annular region is lost. Dalbono and Rebelo  \cite{DaRe-2002} and more recently Fonda, Sabatini and Zanolin \cite{FSZ-2012} have reviewed a series of results on the applicability of the Poincar\'{e}--Birkhoff theorem in non-invariant annuli.
\par
According to some recent findings of Boscaggin and Mu\~{n}oz-Hern\'{a}ndez \cite{BoMH-2022},  there is a link between the variational approach discussed in the previous paragraph and the rotation number estimates necessary to apply the Poincar\'{e}--Birkhoff theorem based on the Conley--Zehnder--Maslov  index (see
Abbondandolo \cite{Ab-2001} and Long \cite{Lo-2002}).
\par
Nevertheless, establishing the existence of the invariant curves is a hard task that typically requires to invoke the KAM theory or some of its variants (see Laederich and Levi \cite{LaLe-1991}), like, e.g.,
the Moser twist theorem \cite{Mo-1962}, as in Dieckerhoff and Zehnder \cite{DiZh-1987} and Levi \cite{Le-1991} for the scalar second order differential equation
$$
x'' + V_x(x,t)=0,
$$
or as in Hausrath \cite{Ha-1982} and Liu \cite{Liu-1995} for the most sophisticated  predator-prey system.
\par
\bigskip

\noindent\textbf{Symbolic dynamics associated with horseshoe-type structures:}
In connection with the previous discussion, subharmonics can be also detected by
applying to the Poincar\'{e} map and its iterates various methods coming from the theory of dynamical systems, the most paradigmatic being the celebrated
Smale's horseshoe (see Smale \cite{Sm-1965,Sm-1967}, Moser \cite[Ch. III]{Mo-1973} and Wiggins \cite{Wi-2003}). Essentially, it consists of a toy-diffeomorphism stretching and bending, recursively,  a square onto a horseshoe-type domain.
\par
In this paper, we will work in a slightly weaker setting entering into the theory of \textit{topological horseshoes} as developed, adapting different perspectives, by Carbinatto, Kwapisz and Mischaikow \cite{CKM-2000}, Mischaikow and Mrozek \cite{MiMr-1995}, Srzednicki \cite{Sr-2000}, Srzednicki and W\'{o}jcik \cite{SrWo-1997}, W\'{o}jcik and Zgliczy\'{n}ski \cite{WoZg-2000}, Zgliczy\'{n}ski \cite{Zg-1996} and Zgliczy\'{n}ski and Gidea \cite{ZgGi-2004}, just to quote some of the most illustrative contributions among a vast literature on this topic. The theory of topological horseshoes, as discussed  by Burns and Weiss \cite{BuWe-1995} and Kennedy and Yorke \cite{KeYo-2001}, consists of a series of methods introduced to extend the classical geometry of the Smale horseshoe to more general dynamical situations involving topological crossings and, crucially,  avoiding any hyperbolicity assumptions on the diffeomorphism, as it might be a challenge to verify them in applications.
\par
In particular, in this paper we will benefit of a topological approach developed by Papini and Zanolin  \cite{PaZa-2004b,PaZa-2004a} and Pascoletti, Pireddu and Zanolin \cite{PPZ-2008} leading to the following notion of \textit{chaos in the coin-tossing sense}.

\begin{definition}\label{def1.1}
    Let $X$ be a metric space and $\Phi: {\mathcal Q}\to X$ be a homeomorphism. It is said that
    $\Phi$ has a topological horseshoe in the set ${\mathcal Q}$ if there are  $\ell\geq 2$
    non-empty pairwise disjoint compact sets ${\mathcal H}_{0},\dots,{\mathcal H}_{\ell-1}
    \subset \mc{Q}$ such that, for every  two-sided sequence of $\ell$ symbols,
    $$
    \textbf{s}=(s_{i})_{i\in {\mathbb Z}}\in \Sigma_{\ell}:=\{0,\dots,\ell-1\}^{\mathbb Z},
    $$
    there exists $z\in {\mathcal Q}$ such that
    $z_{i}:=\Phi^i(z)\in {\mathcal H}_{s_i}$ for all $i\in {\mathbb Z}$ and,
    whenever $(s_i)_{i}$ is a $n$-periodic sequence, then $z$ can be chosen so that
    the sequence of iterates $(z_i)_{i}$ is as well  $n$-periodic.
    In this case, it is also said  that
    \textit{$\Phi$ induces chaotic dynamics on $\ell$ symbols} in ${\mathcal Q}.$
\end{definition}

Definition \ref{def1.1} is inspired in the concept of chaotic
dynamics as a situation where a deterministic map can
reproduce, along its iterates, all the possible outcomes of a
coin-flipping experiment, as discussed by Smale \cite{Sm-1998} (see also
Kirchgraber and Stoffer \cite{KiSt-1989}). According to Medio, Pireddu and Zanolin  \cite{MPZ-2009}, any
map $\Phi$ inducing chaotic dynamics on $\ell$ symbols in ${\mathcal Q}$ is \textit{semi-conjugate} to the Bernoulli shift automorphism
$$
\sigma: \Sigma_{\ell}\to \Sigma_{\ell}, \qquad \sigma(s_{i})_i:=(s_{i+1})_i\quad \hbox{for all}\;\; i\in\Z,
$$
in the sense that there exist a compact subset, $\Lambda\subset
\bigcup_{i=0,\dots\ell-1}{\mathcal H}_i\subset{\mathcal Q}$, invariant  for
$\Phi$, whose set of periodic points, $\mathrm{Per}\,\Phi$, is dense
in $\Lambda$,  and a \textit{continuous and surjective map} $g:
\Lambda\to \Sigma_{\ell}$ such that:
\begin{itemize}
    \item[i)] $g\circ\Phi=\sigma\circ g$, and
    \item[ii)] for every periodic sequence $\textbf{s}\in\Sigma_\ell$, the set $g^{-1}(\textbf{s})$ contains a periodic point of $\Phi$ with the same period.
\end{itemize}
Property ii) corresponds to the one  introduced by Zgliczy\'{n}ski in \cite[Th. 4.1]{Zg-1996}.
Thus, adopting Definition \ref{def1.1}  we are entering into a genuine
classical definition of chaotic dynamics of Block--Coppel type,
as discussed by Aulbach and Kieninger \cite{AuKi-2001}.
Hence, according to Adler, Konheim and McAndrew \cite{AKMc-1965},
$\Phi$ has a positive topological entropy.
The semi-conjugation property provides a weaker form of chaos with respect to the original Smale's
horseshoe, where $g:\Lambda \to \Sigma_{\ell}$ is assumed to be a homeomorphism and, so,
$\Phi|_{\Lambda}$ is \textit{conjugate} to the Bernoulli shift. The conjugation property provides us with
a stronger type of chaotic dynamics as, in such case, $\Phi$ inherits on the invariant set $\Lambda$
all the properties of the automorphism $\sigma$. Therefore, all the existing notions of chaos, such as those introduced by Devaney \cite{De-1989}, Li and Yorke \cite{LiYo-1975}, Aulbach and Kieninger \cite{AuKi-2001}
and Kirchgraber and Stoffer \cite{KiSt-1989} hold simultaneously.
Although there are series of results ensuring the conjugation, as, e.g., the celebrated
Melnikov theorem,  which requires  the delicate task of
analyzing perturbations of homoclinic, or heteroclinic,  configurations, or the theory of linked twist maps
developed by Devaney \cite{De-1978} and Sturman, Ottino and Wiggins \cite{SOW-2006},
except in concrete special examples, it is a challenge to make sure that these sufficient conditions hold.
\par
Clearly, for any map $\Phi$ satisfying the requirements of Definition \ref{def1.1}, one can
infer  the existence of periodic points of arbitrary  minimal period, and hence subharmonics if
$\Phi$ is the Poincar\'{e} map of an ODE with periodic coefficients. For instance,
if $\ell=2,$ given any periodic sequence in $\{0,1\}$ of minimal period $n$, there is also a periodic point of $\Phi$ in ${\mathcal Q}$ with minimal period $n$. Similar notions of ``chaos'' can be given by means of a number of  topological methods, based on the Conley index, fixed point theories, or topological degree (see, e.g., Carbinatto, Kwapisz and Mischaikow \cite{CKM-2000}, Mischaikow and Mrozek \cite{MiMr-1995}, Srzednicki \cite{Sr-2000}, Srzednicki and W\'{o}jcik \cite{SrWo-1997}, W\'{o}jcik and Zgliczy\'{n}ski \cite{WoZg-2000}, Zgliczy\'{n}ski \cite{Zg-1996} and Zgliczy\'{n}ski and Gidea \cite{ZgGi-2004}). Typically, in any setting
entailing Definition \ref{def1.1}, it is possible to detect a larger number of subharmonics than merely
applying the Poincar\'{e}--Birkhoff theorem (see Feltrin \cite[Rem. 4.1]{Fe-2018}).
\par
Among the main advantages of using the theory of topological
horseshoes, instead of the methods discussed in the two previous
items, it is worth-mentioning that, besides the system is not
required to inherit any Hamiltonian structure, there are not constraints on the dimension.
\par
\bigskip
\noindent\textbf{Reduction via symmetry and bifurcation theory:}
Another successful approach for solving a huge variety of
nonlinear differential equations, both ODEs and PDEs,  relies on
the topological degree  through local and global bifurcation
theory. Essentially, in the context of bifurcation theory, the
continuation methods in parametric models do substitute the
Implicit Function Theorem in the presence of degenerate solutions,
provided that a change of degree occurs as the parameter, $\l$,
crosses some critical value, $\l_0$.   The relevance of
bifurcation theory in studying nonlinear differential equations
was first understood by  Krasnosel'skii \cite{Kras-1964}, who was
able to show that any eigenvalue of the linearization  at a given
state with an odd  (classical) algebraic multiplicity is a
nonlinear eigenvalue. By a nonlinear eigenvalue we mean a
bifurcation value from the given state, regardless the nature of
the nonlinear terms of the differential equation. In  other words,
nonlinear eigenvalues are those for which the fact that
bifurcation occurs is \emph{based on the linear part}, as
discussed by Chow and Hale \cite{ChHa-1982}. Some years later,
Rabinowitz \cite{Ra-1971,Ra-1973} established his celebrated
global alternative within the setting of the Krasnosel'skii's
theorem founding Global Bifurcation Theory.  According to the
Rabinowitz's global alternative, the global connected component,
$\mf{C}$,  bifurcating from the given state at an eigenvalue with
an odd (classical) algebraic multiplicity must be either
unbounded, or it bifurcates from the given state at, at least, two
different values of $\l$. The relevant fact that if $\mf{C}$ is
bounded, then the number of bifurcation points from the given
state with an odd algebraic multiplicity must be even was observed
by Nirenberg \cite{Ni-1974}.
\par
However,  the precise role played by the classical spectral theory
in the context of the emerging bifurcation theory remained a real
mystery for two decades. That was a mystery is confirmed by the
astonishing circumstance that the extremely popular transversality
condition  of Crandall and Rabinowitz \cite{CrRa-1971,CrRa-1973}
for bifurcation from simple eigenvalues was not known to entail a
change of the Leray--Schauder degree until Theorem 5.6.2 of
L\'{o}pez-G\'omez \cite{LG-2001} could be derived through the
generalized algebraic multiplicity, $\chi$, of Esquinas and
L\'{o}pez-G\'{o}mez \cite{ELG-1988,Es-1988,LG-2001}. The
multiplicity $\chi$ is far more general that the one introduced in
\cite{CrRa-1971,CrRa-1973} for algebraically simple eigenvalues
and it was used, e.g., by L\'opez-G\'omez and Mora Corral
\cite{LGMC-2007}, to characterize the existence of the Smith
canonical form. According to \cite[Ch.4]{LG-2001}, the oddity of
$\chi$ characterizes whether, or not, $\l_0$ is a nonlinear
eigenvalue of the problem, and this occurs  if, and only if, the
local degree changes as $\l$ crosses $\l_0$. And this regardless
if we are dealing with the Leray--Schauder degree, or with degree
for Fredholm operators of Fitzpatrick, Pejsachowicz and Rabier
\cite{FP-1991,FPR-1994}, or Benevieri and Furi
\cite{BF-1998,BF-2000}, which are almost equivalent.  Therefore,
by Corollary 2.5 of L\'{o}pez-G\'{o}mez and Mora-Corral
\cite{LM-2005}, the local theorem of Crandall and Rabinowitz
\cite{CrRa-1971} is actually global. This important feature was
later \emph{rediscovered} by Shi and Wang \cite{SW-2009} in a much
less general context.
\par
Some more specific important information in the context of dynamical bifurcation theory and singularity theory can be found in the textbooks of Guckenheimer and Holmes \cite{GuHo-1983} and Golubitsky and Shaeffer \cite{GS-1985}. Essentially, singularity theory tries to classify canonically all the possible local structures at the bifurcation values, while dynamical bifurcation theory focuses attention in bifurcation
phenomena not involving only equilibria.
\par
These abstract developments have tremendously facilitated the mathematical analysis of a huge variety
of nonlinear bvp's related to a huge variety of nonlinear differential equations and systems (see, e.g., the monographs of L\'opez-G\'omez \cite{LG-2001,LG-2015}, Cantrell and Cosner \cite{CC-2004} and
Ni \cite{Ni-2011}, as well as their abundant lists of references). However, the underlying mathematical analysis is more involved when dealing with periodic conditions, instead of mixed boundary conditions, especially when searching for branches of
subharmonic solutions bifurcating from a given state. Indeed, although the pioneering strategy
for constructing coexistence states as bifurcating from the semitrivial solution branches  in Reaction-Diffusion systems of Lotka--Volterra type was developed by Cushing \cite{Cu-1977,Cu-1980}
for their classical periodic counterparts, rather astonishingly, except for certain technicalities inherent
to Nonlinear PDEs, the level of difficulty in establishing the existence of the coexistence states in the diffusive prototype models inherits the same order of magnitude as getting them in their classical periodic counterparts. Not to talk about finding out infinitely many subharmonics of large order. Although some additional contributions in this direction were done by T\'aboas \cite{Ta-1987}, the analysis of the periodic-parabolic counterparts of these classical models, extraordinarily facilitated by the pioneering results of Cushing \cite{Cu-1977,Cu-1980} for the non-spatial models, was already ready to be developed by Hess \cite{Hess-1991} and L\'opez-G\'omez \cite{LG-1992}.
\par
Nevertheless, in spite of the huge amount of literature on bifurcation for Reaction-Diffusion Systems in population dynamics, almost no reference is available  about harmonic and subharmonic solutions for  predator-prey systems of Volterra type, except
for those already discussed in this section, beginning with \cite{LOT-1996} and  \cite{LG-2000}, and continuing, after two decades, with \cite{LM-2020}, where the weight functions $\alpha(t)$ and $\beta(t)$ were assumed to have non-overlapping supports so that the underlying Poincar\'{e} map
associated with \eqref{i.7} could take a special form to allow solving the periodic
problem via a one-dimensional reduction.
\par
\bigskip

\noindent  As already told above, the main goal of this paper  is to analyze the existence, multiplicity and structure of subharmonic solutions to planar systems, including the periodic  Volterra's predator-prey model \eqref{i.7}. Naturally, as the underlying Poincar\'{e} maps  play  a crucial role in this analysis, we will benefit of a number of methods and tools  among those already described in the previous paragraphs. As a consequence of our analysis, the richness of the dynamics of \eqref{i.7} will become apparent even  for the simplest configurations of $\a(t)$ and $\b(t)$. Some recent applications of the Poincar\'{e}--Birkhoff
fixed point theorem to equations directly related to \eqref{i.7} have been given by
Boscaggin \cite{Bo-2011}, Ding and Zanolin \cite{DZ-1993,DZ-1996}, Fonda and Toader \cite{FoTo-2019},
Hausrath and Man\'asevich \cite{HM-1991} and Rebelo \cite{Re-1997}. In the more recent papers \cite{LM-2020,LMZ-2020,LMZ-2021} the authors have studied in detail some simple prototype models, non-degenerate and degenerate. Essentially, this paper continues the research program initiated in \cite{LM-2020,LMZ-2020,LMZ-2021} trying to understand  how the relative position of the supports of the weight functions $\a(t)$ and $\b(t)$ might influence the dynamics of \eqref{i.7} and the global structure of the set its subharmonics.
\par
These goals will be achieved in Section \ref{sec2} for the degenerate case by means of the
bifurcation approach introduced in \cite{LM-2020}. Precisely, we will consider the general case of
weight functions having multiple non-overlapping humps as in Figure \ref{fig-ii}. Then, depending on the
mutual distributions of the supports of  $\alpha(t)$ and $\beta(t)$, some sharp estimates on the parameter $\lambda$ ensuring the existence of nontrivial subharmonics will be given. These objectives will
be accomplished in Theorems \ref{th2.1}-\ref{th2.5}, up to
deal with the most general configuration admissible for the validity of these results.
\par
Further, in Section \ref{sec3}, we will analyze some simple prototype models, not previous considered in the literature,  for which the associated Poincar\'{e} consists
of a superposition of a stretching and a twist producing a horseshoe-type geometry. The new findings have been collected in Theorems \ref{th3.1} and \ref{th3.2}. As this topic is more sophisticated technically and not well understood by most of experts in Reaction-Diffusion systems, we will begin the proof of Theorem \ref{th3.1} by giving a rather direct proof for stepwise-constant functions $\a(t)$ and $\b(t)$ before completing the proof in the general case. At a further step we will  discuss the problem of the semi-conjugation/conjugation of the Poincar\'{e} map to the Bernoulli shift, which, essentially,  depends on the shape of the weight functions. Finally, we will close Section \ref{sec3} describing in full detail the geometric horseshoe nature of the Poincar\'{e} map associated with the periodic Volterra predator-prey system. This provides us with a (new) simple mechanism to mimic the Smale's horseshoe from one of the most paradigmatic models in population dynamics.
\par
As a byproduct of our mathematical analysis it becomes apparent how the evolution in seasonal environments where predator-prey interaction plays a role might be random. To catch the attention of experts in Reaction-Diffusion  systems and population dynamics, note that, actually,  the harmonics and subharmonics
of \eqref{i.7} are the non-spatial periodic harmonic and subharmonic solutions of the following Reaction-Diffusion periodic-parabolic problem
\begin{equation}
    \label{i.14}
    \left\{ \begin{array}{ll} \frac{\p u}{\p t}-d_1\D u = \l \a(t)u(1-v)& \quad \hbox{in}\;\; \O\times (0,+\infty),\\[1ex] \frac{\p v}{\p t}-d_2\D v = \l \b(t)v(-1+u)& \quad \hbox{in}\;\; \O\times (0,+\infty),\\[1ex]
        \frac{\p u}{\p n_x}=\frac{\p v}{\p n_x} =0 &\quad \hbox{on}\;\; \p\O \times (0,+\infty),\end{array}\right.
\end{equation}
where $\O$ stands for a $\mc{C}^2$ bounded domain of $\R^N$, $N\geq 1$, $n_x$ stands for the outward normal vector-field to $\O$ on its boundary, $d_1$ and $d_2$ are two positive constants, and $\D$ stands for the
Laplace's operator in $\R^N$. Therefore, the findings of this paper seem extremely relevant from the point of view of population dynamics. As a byproduct of our analysis, non-cooperative systems in periodic environments can provoke chaotic dynamics. This cannot occur in cooperative and quasi-cooperative dynamics, as a consequence of the ordering imposed by the maximum principle. Therefore, it is just the lack of a maximum principle for the predator-prey models the main mechanism provoking chaos in these models, though this
sharper analysis will be accomplished in a forthcoming paper.
\par
To avoid unnecessary repetitions, throughout this paper we will assume that
$\alpha,\beta:{\mathbb R}\to {\mathbb R}^+:=[0,+\infty)$ are continuous and $T$-periodic functions, though
our results are easily extended to the Carath\'{e}odory setting with coefficients measurable and integrable in $L^1([0,T],{\mathbb R}^+).$ In particular, the bounded and piecewise-continuous $\alpha,\beta$
fall within our functional setting. Hence, piece-constant coefficients are admissible in Section 3.

\section{A bifurcation approach: Minimal complexity of subharmonics for a class of degenerate
predator-prey models}\label{sec2}

\noindent We consider the non-autonomous Volterra predator-prey model
\begin{equation}
\label{ii.1}
\left \{
\begin{array}{ll}
u'=\lambda\alpha(t)u(1-v)\\
v'=\lambda\beta(t)v(-1+u)
\end{array}
\right.
\end{equation}
where $\l>0$ is regarded as a real parameter, and the intersection of the supports of the
non-negative weight functions $\a$ and $\b$, denoted by
\begin{equation*}
Z:=\supp \a\cap \supp \b,
\end{equation*}
is assumed to have Lebesgue measure zero, $|Z|=0$. This is the reason why the model \eqref{ii.1} is said to be \emph{degenerate}. In \eqref{ii.1}, given a real number $T>0$, $\a$ and $\b$ are $T$-periodic continuous functions such that
\[
   A:=\int_{0}^{T}\alpha>0,\qquad B:=\int_{0}^{T}\beta>0.
\]
These kind of degenerate Volterra predator-prey models were introduced in
\cite{LOT-1996} and \cite{LG-2000} and, then, deeply analyzed in \cite{LM-2020}.
Actually, these  references dealt with the very special, but interesting, case when
\[
  \supp\a\equiv[0,T/2],\qquad \supp\b\equiv[T/2,T],
\]
where we could benefit of the degenerate character of the model to
ascertain the global structure of the set of $nT$-periodic
coexistence states of \eqref{ii.1}. In the non-degenerate case
when $|Z|\neq0$, the techniques developed in \cite{LM-2020} do not
work. In such case the existence of high order subharmonics can be
established through the celebrated Poincar\'{e}--Birkhoff twist
theorem, or appropriate variants of it (see, e.g.,
\cite{DoZa-2020,FoUr-2017,GMR-2019,GiMa-2019,MRZ-2002}, as well as
\cite{LMZ-2020} for a specific application to \eqref{ii.1}).
However, the twist theorem cannot provide us with the global
bifurcation diagram of subharmonics constructed in \cite{LM-2020}.
The main goal  of this section is sharpening and generalizing as
much as possible the main findings of \cite{LM-2020}.
\par
Precisely, we analyze the existence of $nT$-periodic coexistence
states of \eqref{ii.1} for any integer $n\geq 1$ under the
following structural constraints on the weight functions $\a$ and
$\b$. For some integers $k, \ell\geq1$, with $|k-\ell|\leq 1$, it
is assumed the existence  of $k+\ell$ continuous functions in the
interval $[0,T]$, $\a_i\gneq 0$, $1\leq i\leq k$, and $\b_j\gneq
0$, $1\leq j\leq \ell$, such that
$$
\alpha = \alpha_1+\alpha_2+\cdots+\alpha_k,\quad \beta = \beta_1+\beta_2+\cdots+\beta_\ell,
$$
with
\begin{equation}
\label{ii.2}
\supp\a_i\subseteq[t_0^i,t_1^i]\;\;\hbox{and}\;\; \supp\b_j\subseteq[t_2^j,t_3^j],
\end{equation}
for some partition of $[0,T]$
\begin{equation*}
0\leq t_0^1<t_1^1\leq t_2^1<t_3^1\leq t_0^2<t_1^2\leq t_2^2<t_3^2\leq\cdots\leq t_0^{k}<t_1^{k}\leq
t_2^{k}<t_3^{k} \leq T
\end{equation*}
if $k=\ell$, or
\begin{equation*}
0\leq t_0^1<t_1^1\leq t_2^1<t_3^1\leq t_0^2<t_1^2\leq t_2^2<t_3^2\leq\cdots\leq t_0^{k}<t_1^{k}\leq T
\end{equation*}
if $k=\ell+1$. Similarly, we also consider the case when, instead of \eqref{ii.2},
\begin{equation}
\label{ii.3}
   \supp\b_j\subseteq[t_0^j,t_1^j]\;\;\hbox{and}\;\; \supp\a_i\subseteq[t_2^i,t_3^i],
\end{equation}
for some partition of $[0,T]$
\begin{equation*}
0\leq t_0^1<t_1^1\leq t_2^1<t_3^1\leq t_0^2<t_1^2\leq t_2^2<t_3^2\leq\cdots\leq t_0^{\ell}<t_1^{\ell}\leq
t_2^{\ell}<t_3^{\ell} \leq T
\end{equation*}
if $\ell=k$, or
\begin{equation*}
0\leq t_0^1<t_1^1\leq t_2^1<t_3^1\leq t_0^2<t_1^2\leq t_2^2<t_3^2\leq\cdots\leq t_0^{\ell}<t_1^{\ell}\leq T
\end{equation*}
if $\ell=k+1$.
\par
Moreover, we refer to an $\a$-interval (resp. $\b$-interval) as a
maximal interval where $\b\equiv 0$  (resp. $\a\equiv 0$) and we set
\begin{equation}
\label{ii.5}
A_i:=\int_0^T \a_i,\qquad B_j:=\int_0^T \b_j.
\end{equation}
Figure \ref{fig-ii} shows a series of examples satisfying the previous requirements. Note that the support of
the $\a_i$'s and the $\b_j$'s on each of the intervals $[t_r^i,t_{r+1}^i]$, $1\leq i\leq k$, and
$[t_r^j,t_{r+1}^j]$, $1\leq j\leq \ell$, might not be connected.

\begin{figure}[h!]
    \centering
    \includegraphics[scale=0.71]{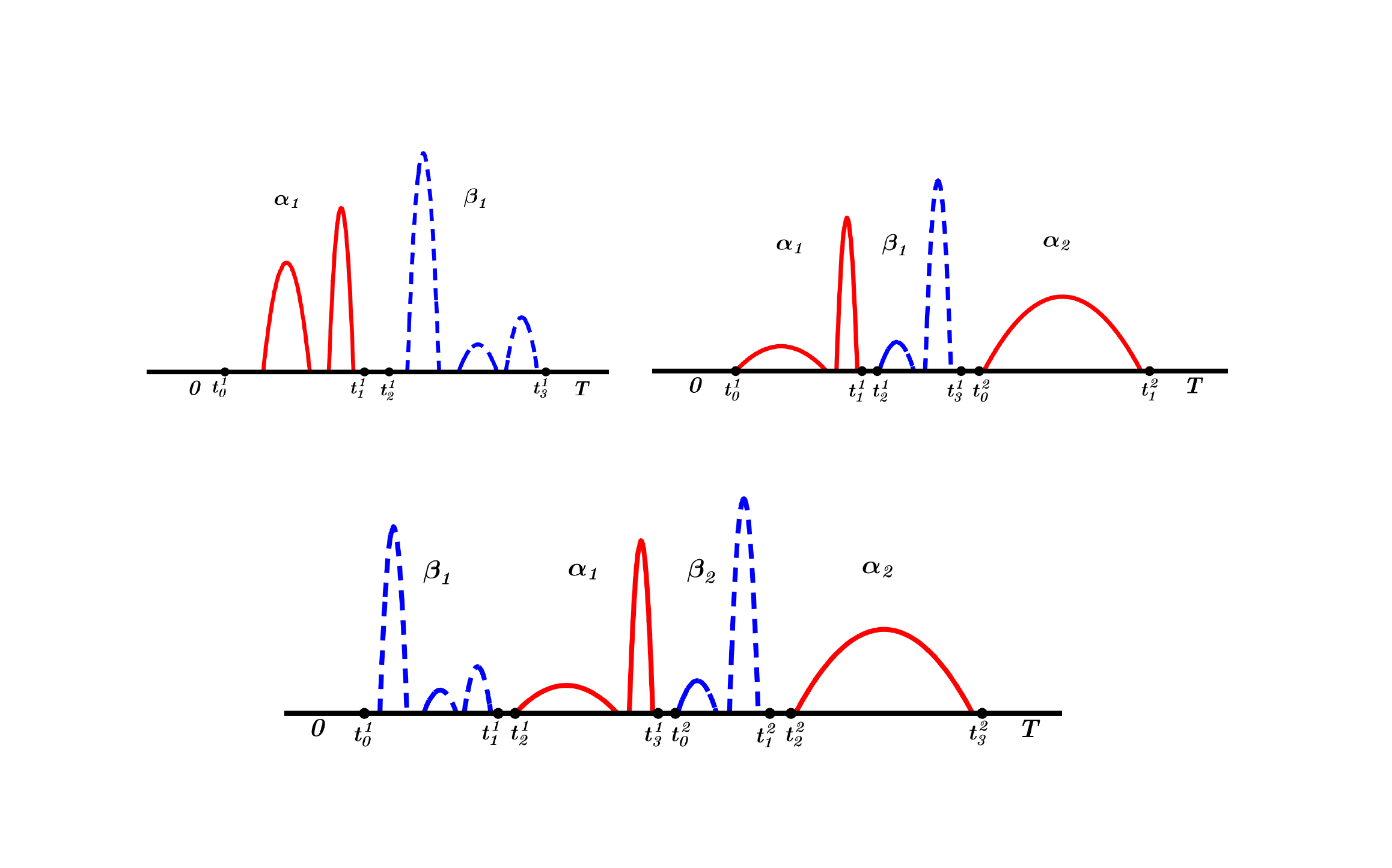}
    \caption{Some admissible examples of weight functions $\a$ and $\b$.}
    \label{fig-ii}
\end{figure}

There are two fundamental aims in this section. The first one is
to show that the  complexity of the global bifurcation diagram of
subharmonics of \eqref{ii.1} when $k=\ell$ depends on the size of
$k$, rather than on the particular structure of the $\a_i$'s and
the $\b_j$'s on each of the intervals of the partition of $[0,T]$.
In fact, all admissible global bifurcation diagrams when $k=\ell$
can be constructed systematically, through a certain algorithm,
from the one already found in \cite{LM-2020}, regardless the
particular locations of each of the points of the partitions,
$t_r^s$'s. The second aim of this section is to determine a lower
bound for the number of $nT$-periodic coexistence states of model
\eqref{ii.1}.
\par
It is elementary to show that, for any initial point
$z_0:=(u_0,v_0)$ (with $u_0,v_0>0$) and each initial time
$\tau_0$, there exists a unique solution
$(u(t;\tau_0,z_0),v(t;\tau_0,z_0))$ to system \eqref{ii.1} which
is globally defined in time. In the sequel, by convention, when
studying $nT$-periodic solutions (for any $n\geq 1$), we will be
looking for the fixed and periodic points of the Poincar\'{e} map
with $\tau_0=0,$ i.e.,
$$z_0=(u_0,v_0)\mapsto (u(t;0,z_0),v(t;0,z_0)).$$
This does not exclude the possibility of the existence of other
fixed points for the Poincar\'{e} maps defined with a different
initial point $\tau_0$. Typically, the corresponding solutions
will be equivalent through an appropriate time translation and
hence, they will be not considered in counting the multiplicity of
the solutions.

\subsection{The case when $k=\ell=1$ and ${\rm supp}\hspace{0.6mm}\a_1\subseteq[t_0^1,t_1^1]$}
\label{sec2.1}
\noindent Then,
\begin{equation}
\label{ii.6}
\supp\a_1\subseteq[t_0^1,t_1^1]\;\;\hbox{and}\;\;\supp\b_1\subseteq[t_2^1,t_3^1],
\end{equation}
where
$$
   0\leq t_0^1<t_1^1\leq t_2^1<t_3^1\leq T.
$$
Figure \ref{fig-iii} illustrates this case.

\begin{figure}[h!]
    \centering
    \includegraphics[scale=0.6]{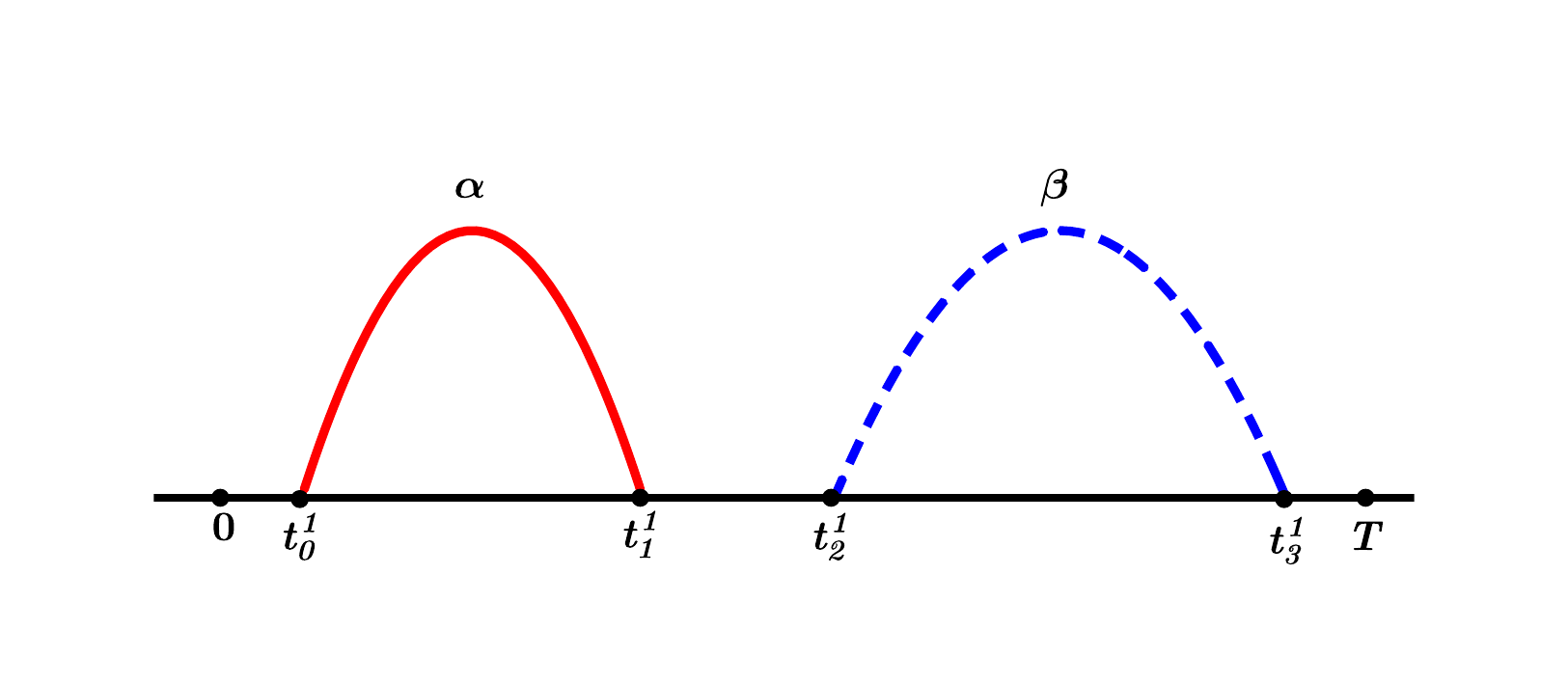}
    \caption{$\a$ and $\b$ satisfying \eqref{ii.6}}
    \label{fig-iii}
\end{figure}

Under condition
\eqref{ii.6}, the next result holds.

\begin{theorem}
\label{th2.1}
Assume \eqref{ii.6}. Then, the equilibrium $(1,1)$ is the unique $T$-periodic coexistence state of  \eqref{ii.1}. Thus, \eqref{ii.1} cannot admit non-trivial $T$-periodic coexistence states. Moreover,
\eqref{ii.1} possesses exactly two non-trivial $2T$-periodic coexistence states for every
\begin{equation*}
  \lambda>\frac{2}{\sqrt{A_1B_1}},
\end{equation*}
where $A_1$ and $B_1$ are those defined in \eqref{ii.5}.
Furthermore, in the special case when $A_1=B_1$ and $u(0)=v(0)$, for every $\lambda>\frac{2}{A_1}$ and $n\geq2$, \eqref{ii.1} has, at least, $n$ non-trivial $nT$-periodic coexistence states if $n$ is even, and $n-1$ if $n$ is odd.
\end{theorem}

\begin{proof}
Thanks to \eqref{ii.6}, \eqref{ii.1} can be integrated. Indeed,
for any given $(u_0,v_0)\in\mathbb{R}^2$, the (unique) solution of
\eqref{ii.1} such that $(u(0),v(0))=(u_0,v_0)$ is given by
\[
    u(t)=u_0e^{(1-v_0)\l\int_{0}^{t}\alpha(s)ds},\qquad v(t)=v_0e^{(-1+u(T))\l\int_{0}^{t}\beta(s)ds},\quad  t\in[0,T].
\]
Thus, the associated $T$-time and $2T$-time Poincar\'{e} maps are defined through
\begin{equation*}
    (u_1,v_1):=\mc{P}_1(u_0,v_0):=(u(T),v(T))=(u_0e^{(1-v_0)\l A_1},v_0e^{(-1+u_1)\l B_1})
\end{equation*}
    and
\begin{equation*}
    \begin{split}
    (u_2,v_2):= \mc{P}_2(u_0,v_0) & =\mc{P}_1(u_1,v_1)=(u_1e^{(1-v_1)\l A_1},v_1e^{(-1+u_2)\l B_1})\\&=(u_0e^{(2-v_0-v_1)\l A_1},v_0e^{(-2+u_1+u_2)\l B_1}).
    \end{split}
\end{equation*}
    It is apparent that the unique $T$-periodic coexistence state, i.e., the unique solution of \eqref{ii.1} such that $u_0,v_0>0$ and $\mc{P}_1(u_0,v_0)=(u_0,v_0)$, is the equilibrium (1,1). Similarly, the $2T$-periodic coexistence states are the solutions such that $u_0,v_0>0$ and $\mc{P}_2(u_0,v_0)=(u_0,v_0)$, i.e., those solutions satisfying
\begin{equation}
\label{ii.7}
    u_0,v_0>0,\qquad 2-u_0=u_1=u_0e^{(1-v_0)\l A_1},\qquad2-v_0=v_1=v_0e^{(1-u_0)\l B_1}.
\end{equation}
Thus, expressing $x\equiv v_0$ in terms of $u_0$, setting
$(A,B)\equiv (A_1,B_1)$, and adapting the corresponding argument
on \cite[p. 41]{LM-2020} it is easily seen that the $2T$-periodic
coexistence states are given by the zeroes of the function
\begin{equation*}
    \varphi(x)=x\left(e^{\frac{e^{(1-x)\l A}-1}{e^{(1-x)\l A}+1}\lambda B}+1 \right)-2.
\end{equation*}
This function has been already analyzed on \cite[Th. 2.1]{LM-2020}, where it was established that
it possesses  exactly two zeros different from the equilibrium $x=1$ for every $\lambda>\frac{2}{\sqrt{AB}}$.
This proves the first part of the theorem.
\par
By iterating $n$ times, it follows that the $nT$-time map is defined through
\begin{equation}
\begin{split}
\label{ii.8}
    (u_n,v_n)&:=\mc{P}_n(u_0,v_0)=\mc{P}_1^n(u_0,v_0)\\& =(u_0e^{(n-v_0-v_1-\cdots-v_{n-1})\l A},v_0e^{(u_1+u_2+\cdots+u_n-n)\l B}).
\end{split}
\end{equation}
Hence, due to \eqref{ii.8}, a solution $(u(t),v(t))$ of
\eqref{ii.1} provides us with a $nT$-periodic coexistence state
if, and only if, $u_0>0$, $v_0>0$ and
\begin{equation}
\label{ii.9}
    \left\{
    \begin{array}{ll}
    n&=v_0+v_1+\cdots+v_{n-1},\\
    n&=u_0+u_1+\cdots+u_{n-1}.
    \end{array}
    \right.
\end{equation}
Thus, assuming that $A=B$ and $u_0=v_0=x$, the system \eqref{ii.9} reduces to one equation (cf. \cite[Le. 3.1]{LM-2020}). Hence, setting
\begin{equation}
\label{ii.10}
E_n(\l,x):=\left\{
\begin{array}{ll}
\displaystyle \exp([\tfrac{n+1}{2}-x\sum_{\substack{i=0\\i\in 2\N}}^{n-1}E_i(\l,x)]\l A),&\quad n\in 2\N+1,\\[2ex]
\displaystyle\exp([x\sum_{\substack{i=1\\i\in 2\N+1}}^{n-1}E_i(\l,x)-\tfrac{n}{2}]\l A),&\quad n\in 2\N,
\end{array}
\right.
\end{equation}
where $E_0(\l,x)=1$, it follows from \cite[Th. 3.3]{LM-2020} that,
for every $n\geq1$,
\[
\varphi_n(x)=\varphi_{n-1}(x)-1+xE_{n-1}(\l,x)
\]
and $\varphi_0\equiv 0$, where these $\varphi_n$'s are the
functions whose zeroes provide us with the $nT$-periodic
coexistence states of \eqref{ii.1} that were constructed in
\cite{LM-2020}. Therefore, we are within the setting of
\cite{LM-2020}, where it was inferred from these features (see
Sections 2--6 of \cite{LM-2020}) the existence of, at least, $n$
non-trivial $nT$-periodic coexistence states if $n$ is even and
$n-1$ if $n$ is odd. This concludes the proof.
\end{proof}

Note that in \cite{LM-2020} we dealt with continuous non-negative
weight functions $\alpha$ and $\beta$ such that
\[
\supp\alpha\equiv[ 0,\tfrac{T}{2}]\;\;\hbox{and}\;\;\supp\beta\equiv[\tfrac{T}{2},T],
\]
whereas in Theorem \ref{th2.1} the weight functions  $\alpha$ and
$\beta$ are two arbitrary non-negative continuous functions with
disjoint supports. As the set of subharmonics obeys identical
equations as in \cite{LM-2020}, it is apparent that, much like in
\cite{LM-2020}, also in this more general case the global
bifurcation diagram of the positive subharmonics of \eqref{ii.1}
follows the general patterns sketched in Figure \ref{fig-i}, which
has been reproduced from \cite{LM-2020}.  Similarly, at the light
shared by the analysis of \cite{LM-2020}, Theorem \ref{th2.1}
establishes that the global topological structure of the
bifurcation diagram sketched in Figure \ref{fig-i} remains
invariant regardless the concrete values of $t_0^1$, $t_1^1$,
$t_2^1$, $t_3^1$ and the number and distribution of the components
of the supports of the weight functions $\a(t)$ and $\b(t)$ on
each of the intervals of the partition of $[0,T]$ induced by these
values. Naturally, much like in \cite{LM-2020}, Figure \ref{fig-i}
shows an ideal global bifurcation diagram, for as the local
behavior of most of the bifurcations from $(1,1)$ is unknown,
except for $n\in \{2,3,4\}$.

\subsection{The case when $k=\ell=1$ and ${\rm supp}\hspace{0.6mm}\b_1\subseteq[t_0^1,t_1^1]$}
\label{sec2.2}

Then,
\begin{equation}
\label{ii.11}
\supp\b_1\subseteq[t_0^1,t_1^1]\;\;\hbox{and}\;\;\supp\a_1\subseteq[t_2^1,t_3^1],
\end{equation}
where
$$
0\leq t_0^1<t_1^1\leq t_2^1<t_3^1\leq T.
$$
Figure \ref{fig-iv} shows  a simple example within this case.

\begin{figure}[h!]
    \centering
    \includegraphics[scale=0.6]{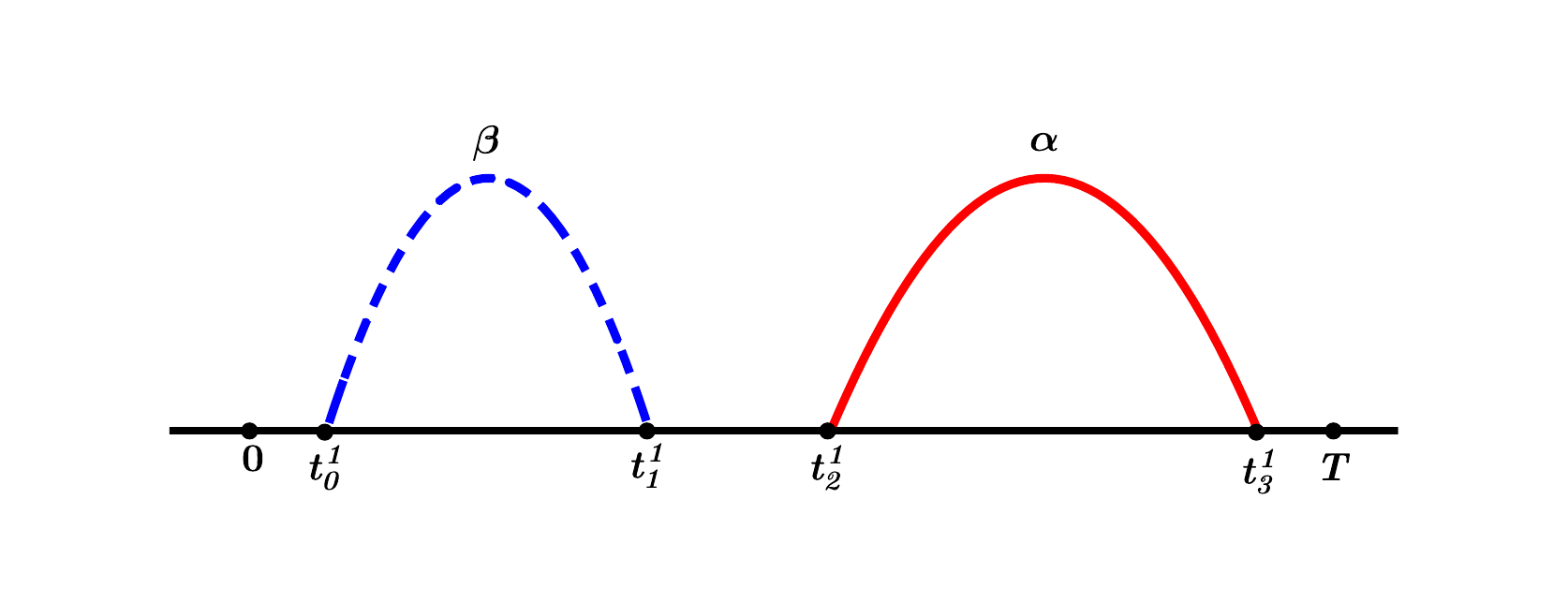}
    \caption{$\a$ and $\b$ satisfying \eqref{ii.11}.}
    \label{fig-iv}
\end{figure}

\begin{remark}\label{rem2.1}
{\rm Based on Theorem \ref{th2.1}, one can get, very easily,
solutions of \eqref{ii.1} satisfying \eqref{ii.11} in the interval
$[0,T]$. Indeed, if $(u_0,v_0)$ is the initial value to an
$nT$-periodic solution of \eqref{ii.1} for the weight distribution
\eqref{ii.6}, then, a solution with initial values
$(u_0e^{(1-v_0)A_1},v_0)$  provides us with an $nT$-periodic
solution of \eqref{ii.1} for the configuration \eqref{ii.11}. Next
section goes further by establishing that, for the distribution
\eqref{ii.11}, there are periodic solutions of \eqref{ii.6} with
initial data on the line $u=v$ by means of similar symmetry
reductions and techniques as in the proof of Theorem \ref{th2.1}.
Equivalently, fixing a time $\tau_0\in (t^1_1,t^1_2)$, and setting
$$
  \tilde{\alpha}(t):=\alpha(t+\tau_0),\qquad \tilde{\beta}(t):=\beta(t+\tau_0),
$$
the pair $(\tilde{\alpha},\tilde{\beta})$ lies within the
configuration of Figure \ref{fig-iii}. Thus, Theorem \ref{th2.1}
applies. Note that this is equivalent to consider the Poincar\'{e}
map with $\tau_0$ as initial time.}
\end{remark}

The next result focuses attention into the case when condition
\eqref{ii.11} holds. A genuine situation where this occurs is
represented in Figure \ref{fig-iv}. As this  case was left outside
the general scope of \cite{LM-2020}, it is a novelty here.

\begin{theorem}
\label{th2.2} Under condition \eqref{ii.11}, the equilibrium
$(1,1)$ is the unique $T$-periodic coexistence state of
\eqref{ii.1}. Moreover, \eqref{ii.1}  possesses exactly two
non-trivial $2T$-periodic coexistence states for every
\begin{equation}
\label{ii.12}
    \lambda>\frac{2}{\sqrt{A_1B_1}}.
\end{equation}
If, in addition, $A_1=B_1$, then the problem \eqref{ii.1} has, for every $\lambda>\frac{2}{A_1}$ and $n\geq 3$, at least $n-1$ non-trivial $nT$-periodic coexistence states with $u_0=v_0$ if $n$ is odd, whereas if $n$ is even, then \eqref{ii.1} possesses, at least, $n-2$ non-trivial $nT$-periodic coexistence states with $u_0=v_0$,
and exactly two with $u_0+v_0=2$.
\end{theorem}

The main difference between Theorems \ref{th2.1} and \ref{th2.2}
relies on the fact that all the solutions of \eqref{ii.1} when
$A_1=B_1$ and $n$ is even have been constructed to satisfy
$u_0=v_0$ under condition \eqref{ii.6}, while \eqref{ii.1} only
admits $n-2$ solutions with $u_0=v_0$ and $2$ solutions with
$u_0+v_0=2$ when \eqref{ii.11} holds.

\begin{proof}
Since, for  every  $(u_0,v_0)\in \R^2$, the unique solution of
\eqref{ii.1} with $(u(0),v(0))=(u_0,v_0)$ is given through
\[
    u(t)=u_0e^{(1-v(T))\l \int_{0}^{t}\alpha(s)ds},\qquad v(t)=v_0e^{(-1+u_0)\l \int_{0}^{t}\beta(s)ds},\quad t\in[0,T],
\]
the $T$-time and $2T$-time Poincar\'{e} maps of \eqref{ii.1} are
given  by
\begin{equation*}
    (u_1,v_1):= \mc{P}_1(u_0,v_0):=(u(T),v(T))=(u_0e^{(1-v_1)\l A_1},v_0e^{(-1+u_0)\l B_1})
\end{equation*}
    and
\begin{equation*}
    \begin{split}
    (u_2,v_2):= \mc{P}_2(u_0,v_0)=\mc{P}_1(u_1,v_1)&=(u_1e^{(1-v_2)\l A_1},v_1e^{(-1+u_1)\l B_1})\\&=(u_0e^{(2-v_1-v_2)\l A_1},v_0e^{(-2+u_0+u_1)\l B_1}).
    \end{split}
\end{equation*}
It is easily seen that $(u_0,v_0)=(1,1)$  is the unique fixed point
of $\mc{P}_1$. Moreover, a solution of \eqref{ii.1},
$(u(t),v(t))$, is a $2T$-periodic coexistence state if, and only
if, $u_0>0$, $v_0>0$ and
$$
   (u_2,v_2)=\mc{P}_2(u_0,v_0)=(u_0,v_0).
$$
In other words,
\begin{equation}
\label{ii.13}
    u_0,\;\; v_0>0,\qquad 2-u_0=u_1=u_0e^{(v_0-1)\l A_1},\qquad2-v_0=v_1=v_0e^{(u_0-1)\l B_1}.
\end{equation}
Setting $(A,B)\equiv(A_1,B_1)$ and arguing as in the proof of
Theorem \ref{th2.1}, it becomes apparent that the non-trivial
$2T$-periodic coexistence states of \eqref{ii.1} are given by the
zeroes of the map
\begin{equation}
\label{ii.14}
    \psi(x):=x\left(e^{\frac{1-e^{(x-1)\l A}}{1+e^{(x-1)\l A}}\lambda B}+1 \right)-2
\end{equation}
with $x=v_0\neq 1$. By definition, it is obvious that
\[
    \psi(x)<0\;\;\hbox{for all}\; x\leq0,\;\;\psi(1)=0,\;\;\psi(x)>0\;\;\hbox{for all}\;x\geq 2.
\]
Moreover, by differentiating with respect to $x$, after rearranging terms,  yields
\begin{align*}
  \psi'(x) & = e^{\tfrac{1-e^{(x-1)\l A}}{1+e^{(x-1)\l A}}\lambda B}\left[ 1-2\l^2
  AB x \tfrac{e^{(x-1)\l A}}{(1+e^{(x-1)\l A})^2} \right]+1, \\[1ex]
  \psi''(x) & =e^{{\frac{1-e^{(x-1)\l A}}{1+e^{(x-1)\l A}}\lambda B}}\!\left[x\!\left(\tfrac{2\lambda^2 ABe^{(x-1)\l A}}{(1+e^{(x-1)\l A})^2}\right)^2\!-\!\tfrac{4\lambda^2ABe^{(x-1)\l A}}{(1+e^{(x-1)\l A})^2}\!+\!\tfrac{2\lambda^3A^2Bxe^{(x-1)\l A}(e^{(x-1)2\l A}-1)}{(1+e^{(x-1)\l A})^4} \right],
\end{align*}
for all $x \geq 0$. In particular,
\[
    \psi'(1)=2-\lambda^2\frac{AB}{2}.
\]
Thus, owing to \eqref{ii.12}, $\psi'(1)<0$. Summarizing,
\eqref{ii.12} implies that
$$
  \psi(0)=-2<0, \quad \psi(1)=0,\quad \psi'(1)<0,\quad \psi(2)>0.
$$
Hence, the function $\psi$ possesses, at least, one zeroes in each
of the intervals $(0,1)$ and $(1,2)$. Therefore, \eqref{ii.1} has,
at least, two $2T$-periodic coexistence states. Moreover, adapting
the analysis carried out in \cite[Sec. 2]{LM-2020}, from the
previous value of $\psi''(x)$,  it is easily seen that any
critical point, $x_c$, of $\psi$ satisfies $\psi''(x_c)<0$ if
$x_c\in (0,1)$ and $\psi''(x_c)>0$ if $x_c\in (1,2)$.
Consequently, \eqref{ii.1} possesses exactly two $2T$-periodic
coexistence states under condition \eqref{ii.12}. This ends the
proof of the first assertion.
\par
Subsequently, we assume  that
\begin{equation}
\label{ii.15}
  A=B, \qquad u_0=v_0=x.
\end{equation}
In this case the $nT$-time map is defined as follows
\begin{equation*}
\begin{split}
(u_n,v_n)&:=\mc{P}_n(u_0,v_0)=\mc{P}_1^n(u_0,v_0)\\& =(u_0e^{(n-v_1-v_2-\cdots-v_{n})\l A},v_0e^{(u_0+u_2+\cdots+u_{n-1}-n)\l A}).
\end{split}
\end{equation*}
Thus, $(u_n,v_n)=(u_0,v_0)$, i.e., $(u_0,v_0)$ provides us with a subharmonic of order $n\geq 1$ of \eqref{ii.1}, if and only if
\begin{equation}
\label{ii.16}
\left\{
\begin{array}{ll}
n&=v_0+v_1+\cdots+v_{n-1},\\
n&=u_0+u_1+\cdots+u_{n-1}.
\end{array}
\right.
\end{equation}
Adapting the argument of the proof of \cite[Lem. 3.1]{LM-2020}, it
is easily seen that the two equations of \eqref{ii.16} coincide
under condition \eqref{ii.15}. Thus, the solutions of the system
\eqref{ii.16} are the zeroes of the function
\[
\psi_n(x):=x+u_1(x)+\ldots+u_{n-1}(x)-n, \qquad x >0.
\]
Consequently, the $nT$-periodic coexistence states of \eqref{ii.1} are given through $\psi^{-1}(0)$.
\par
Arguing as in the proof of \cite[Prop. 3.2]{LM-2020}, it becomes apparent that
\[
(u_n,v_n):=\mathcal{P}_n(x,x)=(xE_{2n}(-\l,x),xE_{2n-1}(-\l,x))
\]
for all $n\geq 1$. Hence,
\begin{equation}
\label{ii.17}
(1,1)=\mathcal{P}_n(1,1)=(E_{2n}(-\l,1),E_{2n-1}(-\l,1))
\end{equation}
for all $n\geq1$. Further, by the proof of \cite[Th. 3.3]{LM-2020}, we find that
\[
  \psi_n(\l,x)=x\sum_{j=0}^{n-1}E_j(-\l,x)-n=\psi_{n-1}(x)-1+xE_{n-1}(-\l,x)
\]
for all $x>0$. Note that $\psi_n$ also depends on the parameter $\l>0$.
To make explicit this dependence, we will subsequently write $\psi_n(\l,x)$,
instead of  $\psi_n(x)$, for all $n\geq 1$.
\par
By \eqref{ii.17}, differentiating with respect to $x$ and
particularizing at $x=1$ yields
\begin{equation}
\label{ii.18}
q_n(\l):=\frac{\partial\psi_n}{\partial x}(\l,1)=n+\sum_{j=1}^{n-1}E_j'(-\l,1).
\end{equation}
Adapting the induction argument of the proof of \cite[Lemma 4.1]{LM-2020},
it follows from the definition of the $E_j$'s that
the function $q_n(\l)$ is a polynomial for all $n\geq 1$.

\par
The next result relates the sequence
$\{q_n\}_{n\geq1}$ with the corresponding sequence
$\{p_n\}_{n\geq1}$ constructed in \cite{LM-2020} under condition
\eqref{ii.6}.

\begin{proposition}
\label{pr2.1}
For every $n\geq1$,
\begin{equation}
\label{ii.19}
q_{2n-1}(\l)=p_{2n-1}(\l), \qquad \frac{q_{2n}(\lambda)}{2+A\lambda}=\frac{p_{2n}(\lambda)}{2-A\lambda},
\end{equation}
and
\begin{equation}
\label{ii.20}
q_n(\l)=[2+(-1)^nA\l]q_{n-1}(\l)-q_{n-2}(\l).
\end{equation}
\end{proposition}
\begin{proof}
According to \cite[(4.6)]{LM-2020}, $p_n$  is defined as
\[
p_n(\l):=\frac{\partial\varphi_n}{\partial x}(\l,1)=n+\sum_{j=1}^{n-1}E_j'(\l,1).
\]
Thus, by \eqref{ii.19}, $q_n(\l)=p_n(-\l)$. By \cite[Cor. 4.7]{LM-2020},
$p_{2n-1}(\l)$ and $\frac{p_{2n}(\l)}{2-A\l}$ are even functions in $\l$. Hence,
\[
q_{2n-1}(\l)=p_{2n-1}(\l)\quad\hbox{and}\quad
q_{2n}(\l)=p_{2n}(-\l)=\frac{p_{2n}(-\l)}{2-A(-\l)}(2+A\l)=\frac{p_{2n}(\l)}{2-A\l}(2+A\l).
\]
This concludes the proof of \eqref{ii.19}.
\par
On the other hand, by \cite[Th. 4.6]{LM-2020},  we already know that
\begin{equation}
\label{ii.21}
p_{2n-1}(\l)=(2+A\l)p_{2n-2}(\l)-p_{2n-3}(\l)
\end{equation}
and
\begin{equation}
\label{ii.22}
\frac{p_{2n}(\l)}{2-A\l}=p_{2n-1}(\l)-\frac{p_{2n-2}(\l)}{2-A\l}.
\end{equation}
Therefore, owing to \eqref{ii.19}, \eqref{ii.21} and \eqref{ii.22},
we find that, for every $n\geq 1$,
\[
\frac{q_{2n-1}(\l)}{2-A\l}=\frac{p_{2n-1}(\l)}{2-A\l}=(2+A\l)\frac{p_{2n-2}(\l)}{2-A\l}-\frac{p_{2n-3}(\l)}{2-A\l}=q_{2n-2}(\l)-\frac{q_{2n-3}(\l)}{2-A\l}
\]
and
\[
\frac{q_{2n}(\l)}{2+A\l}=\frac{p_{2n}(\l)}{2-A\l}=p_{2n-1}(\l)-\frac{p_{2n-2}(\l)}{2-A\l}=
q_{2n-1}(\l)-\frac{q_{2n-2}(\l)}{2+A\l}.
\]
So, \eqref{ii.20} holds, and the proof is complete.
\end{proof}

As a direct consequence of \eqref{ii.19}, the corresponding sets
of bifurcation points from the curve $(\l,x)=(\l,1)$ coincide
under conditions \eqref{ii.6} and \eqref{ii.11}  as soon as
$u_0=v_0 (=x)$ and $A=B$, except for the bifurcation point
$(\l,x)=(2/A,1)$, because
$$
  2/A\in p_{2n}^{-1}(0)\setminus q_{2n}^{-1}(0)\quad \hbox{for all}\;\; n\geq 1.
$$
Moreover, also by \eqref{ii.19}, the mathematical analysis carried
out in  Sections 5 and 6 of \cite{LM-2020} applies \emph{mutatis
mutandis} to cover the case when \eqref{ii.11} holds, instead of
\eqref{ii.6}. As a byproduct, also in the case when \eqref{ii.11}
holds, all the zeroes of the polynomials $q_n(\l)$ are simple.
Thus, the Crandall--Rabinowitz theorem \cite{CrRa-1971} provides
us with a  local analytic curve of $nT$-periodic solutions.
Moreover, since the generalized algebraic multiplicity of Esquinas
and L\'{o}pez-G\'{o}mez \cite{ELG-1988} equals one, according to
\cite[Th. 6.2.1]{LG-2001} and the unilateral theorem \cite[Th.
6.4.3]{LG-2001}, these local curves of subharmonics can be
extended to maximal connected components of the set of
$nT$-subharmonics of \eqref{ii.1}. This proves the theorem when
$u_0=v_0=x$ and $A=B$, regardless the oddity of $n\geq 1$.
\par
Finally, assume that $2=u_0+v_0$ and $A=B$. As we already know that
a solution $(u(t),v(t))$ is $2T$-periodic if, and only if,
\[
  2=u_0+v_0\qquad 2=v_0+v_1=v_0+v_0e^{(u_0-1)\l A},
\]
it becomes apparent that the non-trivial  $2T$-periodic coexistence states of \eqref{ii.1} are the zeroes of the function
\[
  \v_2(x)=x[e^{(1-x)\l A}+1]-2, \qquad x \in (0,2)\setminus\{1\}.
\]
As, according to the proof of Theorem \ref{th2.1}, $\v_2$ possesses exactly two zeroes, the proof of Theorem \ref{th2.2} is completed.
\end{proof}

Since, according to \eqref{ii.19}, we already know that
\[
\{r\in q_n^{-1}(0)\;:\;r>0,\;n\geq1\}=\{r\in p_n^{-1}(0)\;:\;r>0,\;n\geq1\}\setminus\{2\},
\]
the global bifurcation diagram of subharmonics of \eqref{ii.1}
when \eqref{ii.11}, instead of \eqref{ii.6}, holds true, can be
obtained from the global bifurcation diagram plotted in Figure
\ref{fig-i} by removing the component of subharmonics of order
two. However, even the local behavior of the corresponding
components of subharmonics of order $n$ in each of the cases
\eqref{ii.6} and \eqref{ii.11} might be different, because, in
general, $\varphi_n\neq\psi_n$ for all $n\geq 2$ and, hence, the
identity  $\v_n^{-1}(0)= \psi_n^{-1}(0)$ cannot be guaranteed.
\par
By \eqref{ii.7}, $0<v_0<1$ if $0<u_0<1$. Similarly, $1<v_0<2$ if
$1<u_0<2$. Thus, the $2T$-periodic coexistence states of
\eqref{ii.1} under condition \eqref{ii.6} are localized in the
shadowed region of the left plot in Figure \ref{fig-v}. Moreover,
by \eqref{ii.13}, $1<v_0<2$ if $0<u_0<1$, and $0<v_0<1$ if
$1<u_0<2$. Thus, under condition \eqref{ii.11}, the $2T$-periodic
coexistence states of \eqref{ii.1} lye in the shadowed area of the
right plot of Figure \ref{fig-v}. This explains why \eqref{ii.1}
cannot admit a subharmonic of order two if $u_0=v_0$.

\begin{figure}[h!]
    \centering
    \includegraphics[scale=1]{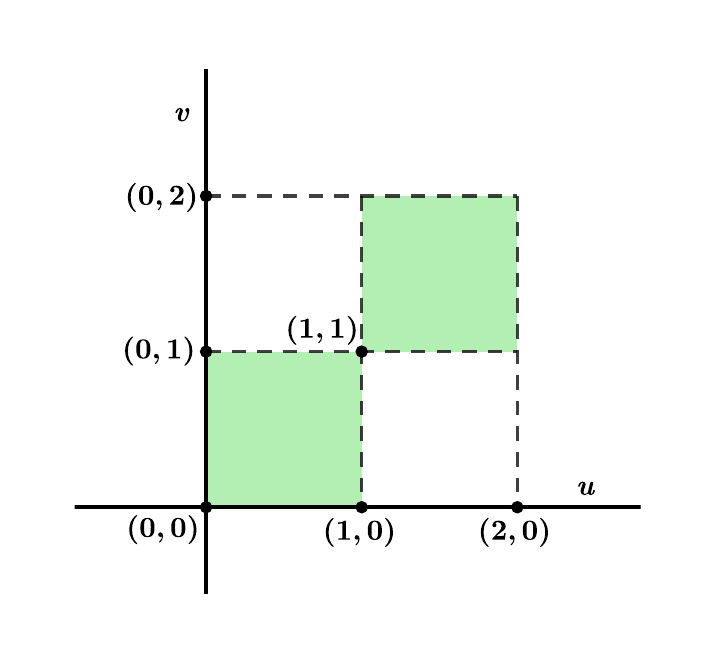}
    \includegraphics[scale=1]{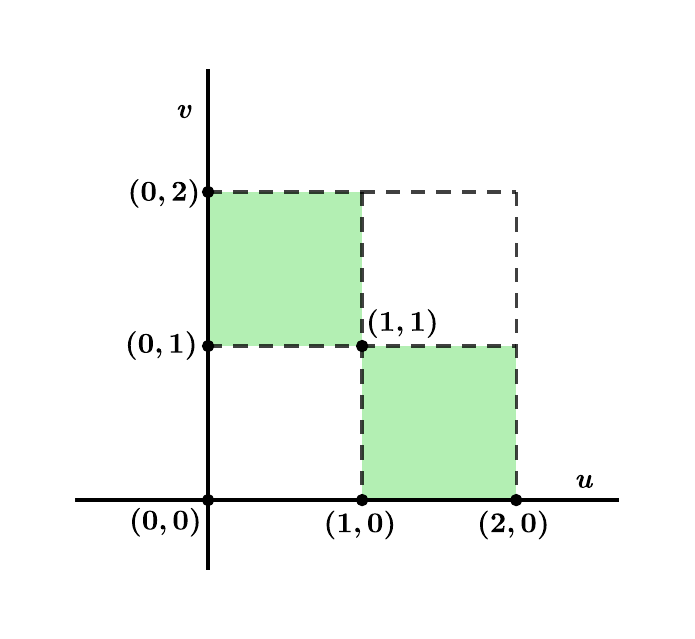}
    \caption{The shadow regions on each of these figures represent
    the quadrants of the phase-plane containing the initial data to
    $2T$-periodic coexistence states of \eqref{ii.1} under the
    conditions \eqref{ii.6} (left) and \eqref{ii.11} (right).}
    \label{fig-v}
\end{figure}

\subsection{The case when $k=\ell\geq 2$}
\label{sec2.3}
Then, either
\begin{equation}
\label{ii.23}
\supp\a_j\subseteq[t_0^j,t_1^j]\quad \hbox{and}\quad \supp\b_j\subseteq[t_2^j,t_3^j]\quad \hbox{for every}\;\; j\in\{1,2,\ldots,k\},
\end{equation}
or
\begin{equation}
\label{ii.24}
\supp\b_j\subseteq[t_0^j,t_1^j]\quad \hbox{and}\quad \supp\a_j\subseteq[t_2^j,t_3^j]\quad \hbox{for every}\;\; j\in\{1,2,\ldots,k\},
\end{equation}
for some
$$
   0\leq t_0^1<t_1^1\leq t_2^1<t_3^1\leq\cdots\leq t_0^k<t_1^k\leq t_2^k<t_3^k\leq T.
$$
As in Theorems \ref{th2.1} and \ref{th2.2}, to analyze the higher
order subharmonics of \eqref{ii.1} we need to impose the
constraints
\begin{equation}
\label{ii.25}
A_1:=\int_{0}^{T}\alpha_1=\cdots=\int_{0}^{T}\alpha_k =
\int_{0}^{T}\beta_1=\cdots=\int_{0}^{T}\beta_k, \qquad u_0=v_0 >0.
\end{equation}
Even in the  simplest case when $k=\ell=1$ it is far from evident
that the analysis of \cite{LM-2020} will be possible to refine  as to
sharpen Theorems \ref{th2.1} and \ref{th2.2} to cover the general
case when
$$
  \int_0^T \a \neq \int_0^T \b.
$$
Essentially, \eqref{ii.25} reduces the problem of finding out the
subharmonics of \eqref{ii.1} to the problem of getting the zeroes
of a sequence of real functions, instead of vectorial ones, much
like in the simplest case when $k=\ell=1$ already covered by
Sections \ref{sec2.1} and \ref{sec2.2}. The next lemma is pivotal
in the proof of the main theorem of this section. We observe that
it holds independently of \eqref{ii.25}.
\begin{lemma}
\label{le2.1}
Assume
\[
A_1:=\int_{0}^{T}\alpha_1=\cdots=\int_{0}^{T}\alpha_k\quad\hbox{and}\quad B_1:=\int_{0}^{T}\beta_1=\cdots=\int_{0}^{T}\beta_k.
\]
Then, under condition \eqref{ii.23} (resp. \eqref{ii.24}),
for any integers $n,m,q,r\geq1$ such that
$$nm=qr,$$
the Poincar\'{e} map of \eqref{ii.1} at time $nT$  for $k=\ell=m$,
denoted by $\mathcal{P}_{n}^{m,\a,\b}$ (resp.
$\mathcal{P}_{n}^{m,\b,\a}$), equals the Poincar\'{e} map of
\eqref{ii.1} at time $qT$ for $k=\ell=r$; denoted, naturally, by
$\mathcal{P}_{q}^{r,\a,\b}$ (resp. $\mathcal{P}_{q}^{r,\b,\a}$).
\end{lemma}
\begin{proof}
Assume \eqref{ii.23} and $k=\ell=m$. Then, integrating in $[0,t_0^2]$ yields
\[
u(t)=u_0e^{(1-v(t_0^1))\l\int_{0}^{t}\a_1(s)ds},
\qquad v(t)=v_0e^{(u(t_0^2)-1)\l\int_{0}^{t}\b_1(s)ds},
\]
for all $t\in[0,t_0^2]$, because $v(t_0^1)=v_0$ and $u(t_0^2)=u(t_2^1)$. Arguing by induction, assume that
\[
u(t)=u_0e^{\sum_{j=1}^{m-1}(1-v(t_0^j))\l\int_{0}^{t}\a_j(s)ds},\qquad v(t)=v_0e^{\sum_{j=2}^{m}(u(t_0^j)-1)\l\int_{0}^{t}\b_{j-1}(s)ds},
\]
for all $t\in[0,t_0^m]$. Then, integrating in $[t_0^m,T]$, it becomes apparent that
\begin{align*}
u(t)& =u_0e^{\sum_{j=1}^{m} (1-v(t_0^j))\l\int_{0}^{t}\a_j(s)ds},\\ v(t) & =
v_0e^{\sum_{j=2}^{m}(u(t_0^j)-1)\l\int_{0}^{t}\b_{j-1}(s)ds+(u(t_0^1+T)-1)\int_0^t \b_m(s)ds}
\end{align*}
for all $t\in[0,T]$. Thus, iterating $n$ times, we find that, for every $t\in [0,nT]$,
\begin{equation}
\label{ii.26}
\left\{
\begin{array}{ll}
u(t)&=u_0e^{\sum_{i=0}^{n-1} \sum_{j=1}^{m} (1-v(t_0^j+iT))\l\int_{0}^{t}\a_j(s)ds }\\[2ex] v(t)&=v_0e^{\sum_{i=0}^{n-1}[\sum_{j=2}^{m}(u(t_0^j+iT)-1)\l\int_{0}^{t}\b_{j-1}(s)ds+(u(t_0^1+(i+1)T)-1)\int_0^t \b_m(s)ds]}
\end{array}
\right.
\end{equation}
for all $t\in[0,nT]$. Moreover, the interval $[0,nT]$ can be
viewed as an interval consisting of $nm$ pairs of $\a$ and $\b$
intervals, instead of made of $n$ copies of $[0,T]$. Thus, setting
for every $1\leq j\leq m$ and $0\leq i\leq n-1$,
$$
  t_0^j+iT \equiv j+mi,\qquad (u(t_0^j+iT),v(t_0^j+iT)) \equiv (u_{j+mi},v_{j+mi}),
$$
and
$$
  \a_j(t) = \a_j(t+iT) \equiv  \a_{j+mi}(t) ,\qquad  \b_j(t) = \b_j(t+iT) \equiv  \b_{j+mi}(t),
$$
\eqref{ii.26} can be equivalently  expressed as
$$
u(t)=u_0e^{\sum_{h =1}^{nm} (1-v_h)\l\int_{0}^{t}\a_h(s)ds},\qquad v(t)=v_0e^{\sum_{h=2}^{nm+1}(u_h-1)\l\int_{0}^{t}\b_{h-1}(s)ds},
$$
for all $t\in[0,nT]$. Thus, it becomes apparent that
\begin{equation}
\label{ii.27}
\begin{split}
\mathcal{P}_n^{m,\a,\b} (u_0,v_0) & :=(u(nT),v(nT))\\[1ex] & =
\left(u_0e^{(nm-\sum_{h=1}^{nm}v_h)\l A_1},v_0e^{(\sum_{h=2}^{nm+1}u_h-nm)\l B_1}\right).
\end{split}
\end{equation}
As in \eqref{ii.27} $n$ and $m$ are arbitrary integer numbers, it
is apparent that
$$
\mathcal{P}_q^{r,\a,\b}=\mathcal{P}_n^{m,\a,\b}
$$
for all integers $q,r\geq1$ such that $nm=qr$,
\par
Lastly, assume \eqref{ii.24} and $k=\ell=m$. Then, arguing as
above yields
\begin{align*}
\mathcal{P}_n^{m,\b,\a}(u_0,v_0)&:=(u(nT),v(nT))\\&=\left(u_0e^{(\sum_{h=2}^{nm+1}v_h-nm)\l A_1},v_0e^{(nm-\sum_{h=1}^{nm}u_h)\l B_1}\right)
\end{align*}
and, therefore, taking integers $q,r\geq1$ such that $nm=qr$, we find that
$$
\mathcal{P}_q^{r,\b,\a} =\mathcal{P}_n^{m,\b,\a}.
$$
The proof is  complete.
\end{proof}

Since $\mathcal{P}_q^{r,\a,\b}=\mathcal{P}_n^{m,\a,\b}$, their
fixed points are the same. Thus, if $nm=qr$, then, the set of
positive fixed points of the Poincar\'e map of \eqref{ii.1} at
time $nT$ for  $k=\ell=m$ equals the set of positive fixed points
of the Poincar\'e map of \eqref{ii.1}  at time $qT$ for
$k=\ell=r$. The main result of this section invokes this feature
to estimate the number of subharmonics of arbitrary order of
\eqref{ii.1} in any of the cases \eqref{ii.23} and \eqref{ii.24}.

\begin{theorem}
\label{th2.3}
Suppose \eqref{ii.25}, $u_0=v_0>0$, and $k=\ell=m\geq2$, and set
\begin{equation}
\label{ii.28}
\nu(z):=\left\{\begin{array}{ll}z&\quad \hbox{if}\;\;z\in 2\N,\\z-1&\quad \hbox{if}\;\;z\in 2\N+1,
\end{array}
\right.
\qquad
\mu(z):=\left\{\begin{array}{ll}z-2&\quad\hbox{if}\;\;z\in 2\N,\\z-1&\quad\hbox{if}\;\;z\in 2\N+1.
\end{array}
\right.
\end{equation}
Then, for every integer $n\geq 1$ and $\l>2/A_1$, \eqref{ii.1} has $\nu(nm)$ (resp. $\mu(nm)$) $nT$-periodic coexistence states under condition \eqref{ii.23} (resp. \eqref{ii.24}).
\end{theorem}
\begin{proof}
Assume  \eqref{ii.23}. By the semigroup property of the flow,
Lemma \ref{ii.23} implies that
\[
(\mathcal{P}_{n}^{1,\a,\b})^m:=\mathcal{P}_{nm}^{1,\a,\b}=\mathcal{P}_n^{m,\a,\b}.
\]
Thus, the set of positive fixed points of
$\mathcal{P}_{nm}^{1,\a,\b}$, already described by Theorem
\ref{th2.1}, equals the set of positive fixed points of
$\mathcal{P}_n^{m,\a,\b}$. As, due to Theorem \ref{th2.1}, the map
$\mathcal{P}_{nm}^{1,\a,\b}$ has, at least,  $\nu(nm)$ positive
fixed points, the map $\mathcal{P}_n^{m,\a,\b}$ also admits, at
least, $\nu(nm)$ positive fixed points for all $\l > 2/A_1$. When,
instead of \eqref{ii.23}, the condition \eqref{ii.24} holds, then
one should invoke to Theorem \ref{th2.2}, instead of Theorem
\ref{th2.1}.  As the proof follows the same patterns, we will omit
any further technical detail.
\end{proof}

Crucially, since $\mathcal{P}_n^{m,\a,\b} =
\mathcal{P}_q^{r,\a,\b}$, the global bifurcation diagram of the
$nT$-periodic coexistence states for $k=\ell=m$ coincides with the
global bifurcation diagram of the $qT$-periodic coexistence states
for $k=\ell=r$ if $nm=qr$. For instance, if $m=2$, then the set of
components of $nT$-periodic coexistence states provides us with
the set of components of $2nT$-periodic coexistence states for
$m=1$. Thus, for $m=2$ the global bifurcation diagram of
subharmonics can be obtained by removing from Figure \ref{fig-i}
the set of components filled in by odd order subharmonics.
Therefore, the global bifurcation diagram sketched in Figure
\ref{fig-i} provides us with all the global bifurcation diagrams
for every $k=\ell =m\geq 2$ by choosing the appropriate
subharmonic components in that diagram.
\par
Finally, the  next result ascertains the number of coexistence
states with minimal period $nT$ among those given by Theorem
\ref{th2.3}. By minimal period $nT$ it is meant that the
coexistence states are $nT$-periodic but not $mT$-periodic if
$m<n$.  To state that result, we first need to deliver two
well-known facts on number theory. For every integer $n\geq1$, the
Euler \emph{totient function} is defined as
\[
\Phi(n):={\rm card}(\{1\leq m\leq n\;|\; \gcd(n,m)=1\}).
\]
According to Gauss \cite[p. 21]{Gauss}, the Euler totient function
satisfies the next identity
\begin{equation}
\label{ii.29}
n=\sum_{d|n}\Phi(d).
\end{equation}
The next result relates, through \eqref{ii.29}, the Euler totient
function with a very special class of univariate polynomials.

\begin{proposition}
\label{pr2.2} Let $\mathbb{K}[X]$ be the univariate polynomials
ring over a zero characteristic field,  $\mathbb{K}$. Then, for
every sequence $\{h_n\}_{n\geq2}\subset \mathbb{K}[X]$ satisfying:
\begin{enumerate}
\item[{\rm 1.}] ${\rm deg}(h_n)=n-1$,
\item[{\rm 2.}] ${\rm card}\{r\in\R\,:\, h_n(r)=0\}=n-1$, and
\item[{\rm 3.}] $h_{n_1}|h_{n_2}$ if, and only if, $n_1|n_2$,
\end{enumerate}
the following identity holds
$$
\Phi(n)={\rm card}\left\{r\in\R\,:\, h_n(r)=0\;\hbox{and}\; h_d(r)\neq0\;\hbox{if}\;d|n\right\}.
$$
\end{proposition}
\begin{proof}
Setting
\begin{align*}
\mc{C}_h(n) & :={\rm card}\{r\in\R\,:\, h_n(r)=0\},\\
\mc{C}_h^{\rm min}(n) & :={\rm card}\{r\in\R\,:\, h_n(r)=0\;\hbox{and}\; h_d(r)\neq0\;\hbox{if}\;d|n\},
\end{align*}
it is apparent that
\[
\mc{C}_h^{\rm min}(n)=\mc{C}_h(n)-\sum_{\substack{d|n\\d\neq 1,n}}\mc{C}_h^{\rm min}(d).
\]
By definition, $\mc{C}_h^{\rm min}(p)=\Phi(p)$ for every prime
integer $p\geq2$, because $\mc{C}_h^{\rm min}(p)=\mc{C}_h(p)=p-1$.
Moreover, for any given prime integers $p_1, p_2\geq2$,
\begin{equation*}
\mc{C}_h^{\rm min}(p_1p_2)=
\left\{
\begin{array}{ll}
\mc{C}_h(p_1p_2)-\mc{C}_h^{\rm min}(p_1)-\mc{C}_h^{\rm min}(p_2)=p_1p_2-1-\Phi(p_1)-\Phi(p_2)&\quad\hbox{if}\;\, p_1\neq p_2,\\[7pt]
\mc{C}_h(p_1p_2)-\mc{C}_h^{\rm min}(p_1)=p_1p_2-1-\Phi(p_1)
&\quad\hbox{if}\;\, p_1= p_2.
\end{array}
\right.
\end{equation*}
Thus, by \eqref{ii.29},  $\mc{C}_h^{\rm min}(p_1p_2)=\Phi(p_1p_2)$. Now, given $i\geq2$, assume as a
complete induction hypothesis, that, for every
$j\in\{1,2,\ldots,i\}$,
\begin{equation}
\label{ii.30}
\mc{C}_h^{\rm min}(p_1p_2\ldots p_j)=\Phi(p_1p_2\ldots p_j).
\end{equation}
Then, denoting $\g:=p_1p_2\ldots p_{j+1}$, it follows from \eqref{ii.30} that
\[
\mc{C}_h^{\rm min}(\g)=\mc{C}_h(\g)-\sum_{\substack{d|\g\\d\neq 1,\g}}\mc{C}_h^{\rm min}(d)=\g-1-\sum_{\substack{d|\g\\d\neq 1,\g}}\Phi(d)=\g-\sum_{\substack{d|\g\\d\neq \g}}\Phi(d).
\]
Therefore, by \eqref{ii.29}, $\mc{C}_h^{\rm min}(\g)=\Phi(\g)$,
which concludes the induction.
As any integer, $n$, can be factorized  as a (unique)
finite product of prime integers, it becomes apparent
that $\mc{C}_h^{\rm min}(n)=\Phi(n)$. This ends the proof.
\end{proof}

Next theorem is a direct consequence of Proposition \ref{pr2.2}.

\begin{theorem}
\label{th2.4}
Assume \eqref{ii.25}, $k=\ell=m\geq1$, and either \eqref{ii.23}, or \eqref{ii.24}.
Then, \eqref{ii.1} possesses, at least, $\Phi(nm)$ coexistence states with minimal period $nT$ for all $nm>2$.
\end{theorem}
\begin{proof}
First, assume that \eqref{ii.23}. Then, by Lemma 4.3 and Theorem 5.2 of \cite{LM-2020},
the sequence of polynomials $p_{nm}$ whose positive roots are the bifurcation
points to the $nT$-periodic coexistence states, satisfies the hypothesis of Proposition \ref{pr2.2}. Thus,
\[
\mc{C}_p^{\min}(nm)=\Phi(nm).
\]
Subsequently, we denote by
\[
\mc{C}_{p,+}^{\rm min}(nm):={\rm card}\{r\in\R,\,r>0\,:\,
p_{nm}(r)=0\;\hbox{and}\; p_d(r)\neq0\;\hbox{if}\;d|nm\}
\]
the cardinality of the set of bifurcations points of \eqref{ii.1}
to minimal $nT$-periodic coexistence states.
As, thanks to \cite[Cor. 4.7]{LM-2020}, we already know that
$\frac{p_{2nm}(\l)}{2-A\l}$ and $p_{2nm+1}(\l)$ are even polynomials, it is apparent that
\[
\mc{C}_{p,+}^{\rm min}(nm)=\frac{\Phi(nm)}{2}.
\]
Thus, as they emerge at least two solutions from each positive root,
there are, at least, $\Phi(nm)$ coexistence states  with minimal period $nT$ for $nm>2$.
\par
Finally, assume \eqref{ii.24}. Then, owing to \eqref{ii.19},
\[
\mc{C}_{q,+}^{\rm min}(nm)=\mc{C}_{p,+}^{\rm min}(nm)=\frac{\Phi(nm)}{2}.
\]
Therefore, the same conclusion holds. This ends the proof.
\end{proof}

Now we will
ascertain, under the assumptions of Theorem \ref{th2.4},
the classes of periodicity of the subharmonics of \eqref{ii.1}.
Recall that, thanks to Lemma \ref{le2.1}, $\mathcal{P}_n^m=\mathcal{P}_{nm}^1$.
Thus, for $k=\ell=m$, the $nT$-periodic coexistence states of \eqref{ii.1}
are the same as the $nmT$-periodic coexistence states for $k=\ell=1$.
Remember that, in case $k=\ell=1$, we already know from Lemma 3.1 of \cite{LM-2020} that
\begin{equation}
\label{ii.31}
u_h=v_{nm-h}\qquad\hbox{for all}\;\, h\in\{1,2,\ldots,nm-1\}.
\end{equation}

\begin{proposition}
\label{pr2.3}
Assume \eqref{ii.25}, $k=\ell=1$, and either \eqref{ii.23},
or \eqref{ii.24}, and let $(u,v)$ be  a minimal $nmT$-periodic coexistence state
of \eqref{ii.1} with $nm>2$. Then,
\begin{equation}
\label{ii.32}
u_h= v_h \;\; \hbox{if and only if}\;\; nm\in 2\mathbb{N}\;\;\hbox{and}\;\; h=\tfrac{nm}{2}.
\end{equation}
Therefore:
\begin{enumerate}
\item[{\rm 1.}] The $\Phi(nm)$
subharmonics of order $nm$ of \eqref{ii.1} given by Theorem \ref{th2.4}
lie in different periodicity classes if $nm\in 2\mathbb{N}+1$, $nm\geq 3$, and
\item[{\rm 2.}] At least $\Phi(nm)/2$ of these subharmonics
lie in different periodicity classes if $nm\in 2\mathbb{N}$, $nm\geq 2$.
\end{enumerate}
\end{proposition}
\begin{proof}
The proof of \eqref{ii.32} proceeds by contradiction.
Assume that $u_h=v_h$ for some $h\neq nm/2$. Then, by \eqref{ii.31},
\[
v_{nm-h}=u_h=v_h=u_{nm-h}\qquad\hbox{for all}\;\, h\in\{1,2,\ldots,nm-1\},
\]
which, in particular, implies that
\begin{equation}
\label{ii.33}
z_h:=(u_h,v_h)=(u_{nm-h},v_{nm-h})=:z_{nm-h}.
\end{equation}
Note that, by the structure of \eqref{ii.1},
\begin{equation}
\label{ii.34}
  z_h\neq z_{h+1} \quad \hbox{for all}\;\; h\in\{1,2,\ldots,nm-2\}.
\end{equation}
Subsequently, we set
$$
  \o_{\min}:=\min\{h,nm-h\},\quad \o_{\max}:=\max\{h,nm-h\},\quad
  \o^*:=\o_{\max}-\o_{\min}.
$$
Thanks to \eqref{ii.33}, by the $T$-periodicity of \eqref{ii.1}, we find that, for every $k\in\mathbb{N}$,
\[
z_{\o_{\max}}=z_{\o_{\min}+k\o^*\,({\rm mod}\, nm)}.
\]
Suppose $\gcd(\o^*,nm)=1$. Then, by the B\'ezout's Identity,
there exists an integer $k_0\geq1$ such that $(k_0+1)\o^*=\o^*+1\,({\rm mod}\,nm)$, which contradicts \eqref{ii.34}. Thus, $\gcd(\o^*,nm)>1$ and, hence,
there exists $k<nm$ such that $k\o^*=0\,({\rm mod}\,nm)$. Therefore,
$$
   \o_{\min}+k\o^*=\o_{\min}\,({\rm mod}\,nm)
$$
which implies that the solution is $kT$-periodic with $k<nm$.
This contradicts the  minimality of the period and ends the proof.
\end{proof}

We conclude this section with a quick comparison with the previous results of the authors in
\cite{LMZ-2021} though the Poincar\'{e}--Birkhoff theorem. According to \cite{LMZ-2021}, if $nm\geq3$ and $n=3h+i\geq1$ for some $h\geq0$ and $i\in\{0,1,2\}$, then there exists $\l_n>0$ such that, for every $\l>\l_n$, \eqref{ii.1} possesses, at least,
\begin{equation}
\label{ii.35}
\s(n)=2\left(hm+\left[\frac{im}{3}\right]\right)
\end{equation}
$nT$-periodic solutions. Moreover, setting
\begin{equation}
\label{ii.36}
\gamma(n):=\min \left\lbrace \gamma\geq0\;:\; \gcd\left(n,\tfrac{\s(n)}{2}-\g\right)=1 \right\rbrace,
\end{equation}
it turns out that, for every $\l>\l_n$, \eqref{ii.1} has, at least, $\s(n)-2\g(n)$ periodic solutions with minimal period $nT$. Next, we will compare, in some special cases, the lower bounds on the number of independent  subharmonics provided by Proposition \ref{pr2.3} and Theorem 4 of \cite{LMZ-2021}.
Assuming that $n$ and $m$ prime numbers, the Euler totient function $\Phi$ satisfies
\[
\Phi(nm)=\left\{
\begin{array}{lll}
(n-1)(m-1)&\quad\hbox{if}\;\; n\neq m,\\[0.5ex]
n(n-1)&\quad\hbox{if}\;\; n=m,\\[0.5ex]
nm-1&\quad\hbox{if}\;\; nm\;\hbox{is prime}.
\end{array}
\right.
\]
Thus, thanks to \eqref{ii.35}, there exists a constant  $c>0$ such that
\begin{equation*}
\Phi(nm)-(\s(n)-2\g(n))\geq \frac{nm}{3}-c(n+m)+2\g(n),
\end{equation*}
which is positive for sufficiently large $n$ and $m$.
Therefore, within this rank, under the strong assumptions
on the weight functions and on the initial values imposed in this section,
Proposition \ref{pr2.3} is sharper than \cite[Th. 4]{LMZ-2021}.
However, in some other circumstances
the previous difference might be negative, being
in these cases deeper \cite[Th.4]{LMZ-2021}
than Proposition \ref{pr2.3}.
It remains an open problem here to find out the eventual
relationships between $\Phi$ and $\s-2\g$, if any.
\par
In \cite{Ne-1977} (see also \cite[Sec.4.1.2]{BoMH-2022}), where
the Poincar\'{e}--Birkhoff theorem was improved from several
perspectives, another lower bound, also related to $\Phi(n)$,  was
given for the number of subharmonics of order $n$ of \eqref{ii.1}.

\subsection{The case when $k\neq\ell$}\label{sec2.4}

Necessarily, either  $k=\ell+1$, or $\ell=k+1$. Moreover, setting  $m:=\min\{k,\ell\}$,
there exist
$$
0\leq t_0^1<t_1^1\leq t_2^1<t_3^1\leq\cdots\leq t_0^m<t_1^m\leq t_2^m<t_3^m\leq t_1^{m+1}<t_2^{m+1}\leq T,
$$
such that
\begin{equation}
\label{ii.37}
\supp\a_i\subseteq[t_0^i,t_1^i]\quad \hbox{and}\quad \supp\b_j\subseteq[t_2^j,t_3^j]
\end{equation}
if $k=\ell+1$,  whereas
\begin{equation}
\label{ii.38}
\supp\b_j\subseteq[t_0^j,t_1^j]\quad \hbox{and}\quad \supp\a_i\subseteq[t_2^i,t_3^i]
\end{equation}
if $\ell=k+1$.
\par
The next result shows that also in this case \eqref{ii.1} has as many subharmonics
as in the context of Theorem \ref{th2.3} with $k=\ell$.

\begin{theorem}
\label{th2.5}
Assume
\begin{equation}
\label{ii.39}
A_1\!:=\!
\int_{0}^{T}\left(\alpha_1+\alpha_{m+1}\right)\!=\!
\int_{0}^{T}\alpha_2\!=\!\cdots\!=\!\int_{0}^{T}\alpha_{m}\!=\!
\int_{0}^{T}\beta_1\!=\!\int_{0}^{T}\beta_2\!=\!\cdots\!=\!\int_{0}^{T}\beta_{m}
\end{equation}
if \eqref{ii.37} holds, and
\begin{equation}
\label{ii.40}
B_1\!:=\! \int_{0}^{T}\left(\b_1+\b_{m+1}\right)\!=\!
\int_{0}^{T}\b_2\!=\!\cdots\!=\!\int_{0}^{T}\b_{m}\!=\!
\int_{0}^{T}\a_1\!=\!\int_{0}^{T}\a_2\!=\!\cdots\!=
\!\int_{0}^{T}\a_{m}
\end{equation}
under condition \eqref{ii.38}.
Then, much like in Theorem \ref{th2.3},
for every integer $n\geq 1$ and $\l>2/A_1$,
\eqref{ii.1} has $\nu(nm)$ (resp. $\mu(nm)$) $nT$-periodic coexistence
states under condition \eqref{ii.37} (resp. \eqref{ii.38}).
\end{theorem}
\begin{proof}
Assume \eqref{ii.37} and \eqref{ii.39}, and let denote the Poincar\'e map in the time interval $[s_1,s_2]$ by $\mc{P}_{[s_1,s_2]}^{k=\ell}$ if $k=\ell$, and by $\mc{P}_{[s_1,s_2]}^{k=\ell+1}$ if $k=\ell+1$.  Then, the value $\int_0^T\a_1$ in case $k=\ell$ equals $\int_0^T(\a_1+\a_{m+1})$ in case $k=\ell+1$.
\par
Let $(u_0,v_0)$ be a fixed point of $\mc{P}_{[0,nT]}^{k=\ell}=\mc{P}_{[0,t_3^{nm}]}^{k=\ell}$. Then,
the point $(u_0e^{(1-v_0)\int_0^T\a_{m+1}},v_0)$ is a fixed point of $\mc{P}_{[0,nT]}^{k=\ell+1}=\mc{P}_{[0,t_1^{nm+1}]}^{k=\ell+1}$. Indeed, by the structure of \eqref{ii.1},
\begin{equation*}
\begin{split}
\mc{P}^{k=\ell+1}_{[0,t_1^1]}(u_0e^{(1-v_0)\int_0^T\a_{m+1}},v_0)&=
(u_0e^{(1-v_0)\int_0^T(\a_{m+1}+\a_1)},v_0)\\&=(u_0e^{(1-v_0)A_1},v_0)=\mc{P}^{k=\ell}_{[0,t_1^1]}(u_0,v_0).
\end{split}
\end{equation*}
Moreover, thanks to \eqref{ii.39}, we also have that $\mc{P}^{k=\ell+1}_{[t_1^1,t_{3}^{nm}]}=\mc{P}^{k=\ell}_{[t_1^1,t_{3}^{nm}]}$. Thus, since $(u_0,v_0)$ is a fixed point of $\mc{P}_{[0,t_3^{nm}]}^{k=\ell}$, it becomes apparent that
\[
\mc{P}^{k=\ell+1}_{[0,t_{3}^{nm}]}(u_0e^{(1-v_0)\int_0^T\a_{m+1}},v_0)=(u_0,v_0).
\]
Therefore,
\[
\mc{P}^{k=\ell+1}_{[0,t_{1}^{nm+1}]}(u_0e^{(1-v_0)\int_0^T\a_{m+1}},v_0)=(u_0e^{(1-v_0)\int_0^T\a_{m+1}},v_0),
\]
i.e., $(u_0e^{(1-v_0)\int_0^T\a_{m+1}},v_0)$ is an $nT$-periodic coexistence state.
This establishes a bijection between the $nT$-coexistence states of \eqref{ii.1}
in cases $\ell=k$ and $k=\ell+1$, and shows that \eqref{ii.1} has $\nu(nm)$
$nT$-periodic coexistence states under condition \eqref{ii.37}.
\par
As the proof when $\ell=k+1$ can be accomplished similarly, we will omit its technical details here.
\end{proof}

Our next result gives some sufficient conditions for non-existence.

\begin{lemma}
\label{le2.2}
The following non-existence results hold:
\begin{enumerate}
\item[{\rm (i)}] \eqref{ii.1} cannot admit any non-trivial $nT$-periodic coexistence state, $n\in\N$, if $k+\ell=1$.
\item[{\rm (ii)}] \eqref{ii.1} cannot admit any non-trivial $T$-periodic coexistence state if $k+\ell=3$.
\end{enumerate}
\end{lemma}
\begin{proof}
If $k+\ell=1$, then either $k=1$ and $\ell=0$, or $\ell=1$ and $k=0$. Thus, either $u$, or $v$, are constant for all $t\in[0,T]$, which ends the proof.
\par
Now, suppose that $k+\ell=3$. Then, there exist
$$
0\leq t_0^1<t_1^1\leq t_2^1<t_3^1\leq t_0^2<t_1^2\leq T,
$$
such that either
\begin{equation*}
\supp\alpha_1\subseteq[t_0^1,t_1^1],\quad \supp\beta_1\subseteq[t_2^2,t_3^1], \qquad \supp\alpha_2\subseteq[t_0^2,t_1^2],
\end{equation*}
or
\begin{equation*}
\supp\b_1\subseteq[t_0^1,t_1^1],\quad \supp\a_1\subseteq[t_2^2,t_3^1],\quad
\supp\b_2\subseteq[t_0^2,t_1^2],
\end{equation*}
as illustrated in Figure \ref{fig-vi}.
\begin{figure}[h!]
\centering
\includegraphics[scale=0.4]{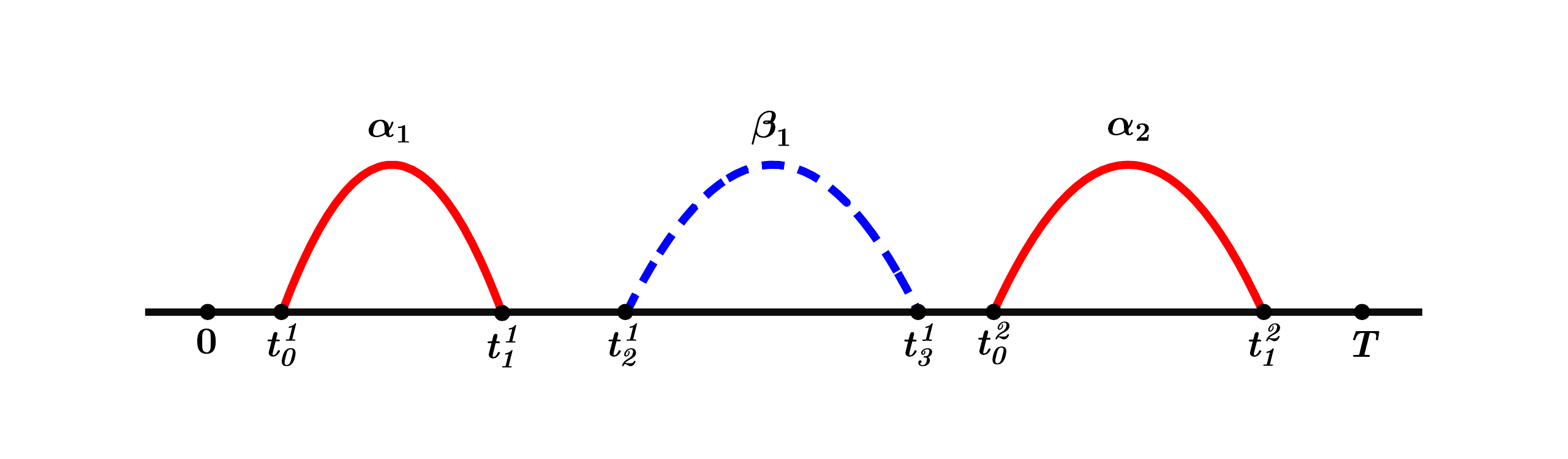}\\
\includegraphics[scale=0.4]{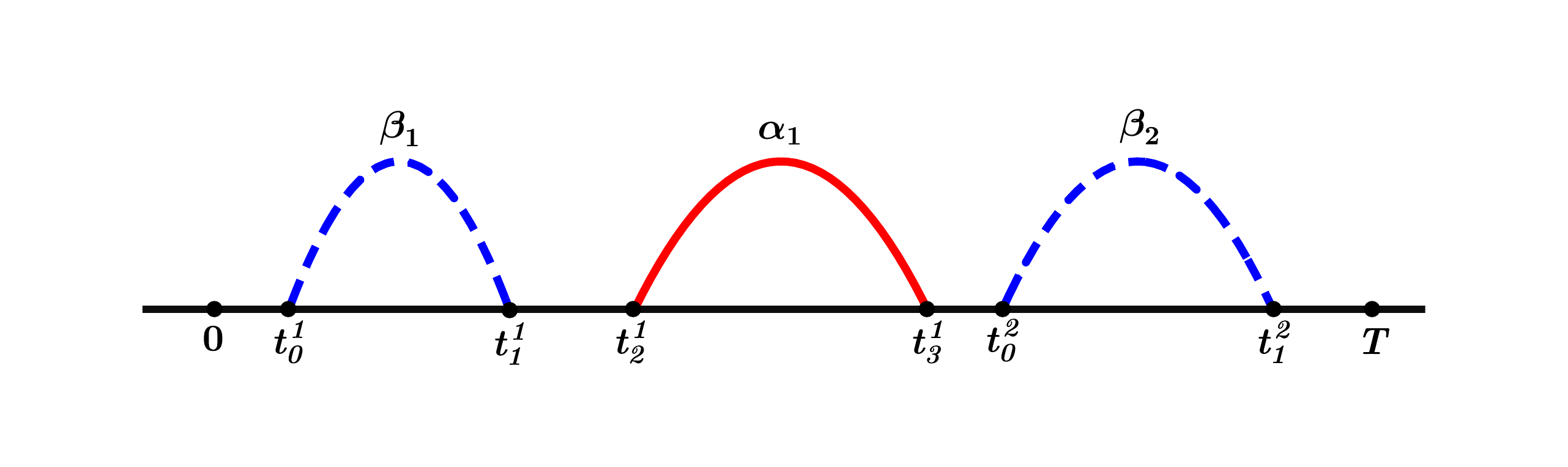}
\caption{Two admissible examples with $k=2$, $\ell=1$, and $k=1$, $\ell=2$. }
\label{fig-vi}
\end{figure}

By the special structure of \eqref{ii.1}, where $\a\b=0$, the orbit of any solution in the interval $[0,T]$ consists of three lines, two of them parallel in the phase-plane to one of the axis, while the third one is parallel to
the other axis. So, these orbits cannot be closed. Therefore, \eqref{ii.1} cannot admit any
$T$-periodic solution.
\end{proof}
\par

The next theorem summarizes the results found in the previous  four sections. It characterizes
the existence of $T$-periodic coexistence states, and subharmonics of all orders of \eqref{ii.1},
in terms of the number of $\a$-intervals and $\b$-intervals in $[0,T]$.

\begin{theorem}
\label{th2.6}
The system \eqref{ii.1} admits, for sufficiently large $\lambda$, some $T$-periodic coexistence state if, and only if,  $k+\ell\geq4$. Moreover, for every $\l>2/A_1$, under the appropriate symmetry properties,  \eqref{ii.1} has subharmonics of all orders, $n\geq 2$, if, and only if, $k+\ell\geq2$.
\end{theorem}

\subsection{The limiting $T$-periodic case $k=\ell=2$}
\label{sec2.5}

According to Theorem \ref{th2.6}, the
condition $k+\ell\geq4$ is necessary and sufficient
so that \eqref{ii.1}
can admit a $T$-periodic coexistence state.
In this section, we deal with the limiting case when $k=\ell=2$ and ascertain
the bifurcation directions to $T$-periodic coexistence states.
Note that when $k=\ell\geq2$, then there exist
$$
   0\leq t_0^1<t_1^1\leq t_2^1<t_3^1\leq t_0^2<t_1^2\leq t_2^2<t_3^2\leq T.
$$
such that either
\begin{equation}
\label{ii.41}
\supp\alpha_1\subseteq[t_0^1,t_1^1],\quad \supp\beta_1\subseteq[t_2^1,t_3^1],\quad \supp\alpha_2\subseteq[t_0^2,t_1^2],\quad \supp\beta_2\subseteq[t_2^2,t_3^2],
\end{equation}
or
\begin{equation}
\label{ii.42}
\supp\beta_1\subseteq[t_0^1,t_1^1],\quad \supp\alpha_1\subseteq[t_2^1,t_3^1],\quad \supp\beta_2\subseteq[t_0^2,t_1^2],\quad \supp\alpha_2\subseteq[t_2^2,t_3^2].
\end{equation}
Figure \ref{fig-vii} shows two admissible configurations.

\begin{figure}[h!]
\centering
\includegraphics[scale=0.4]{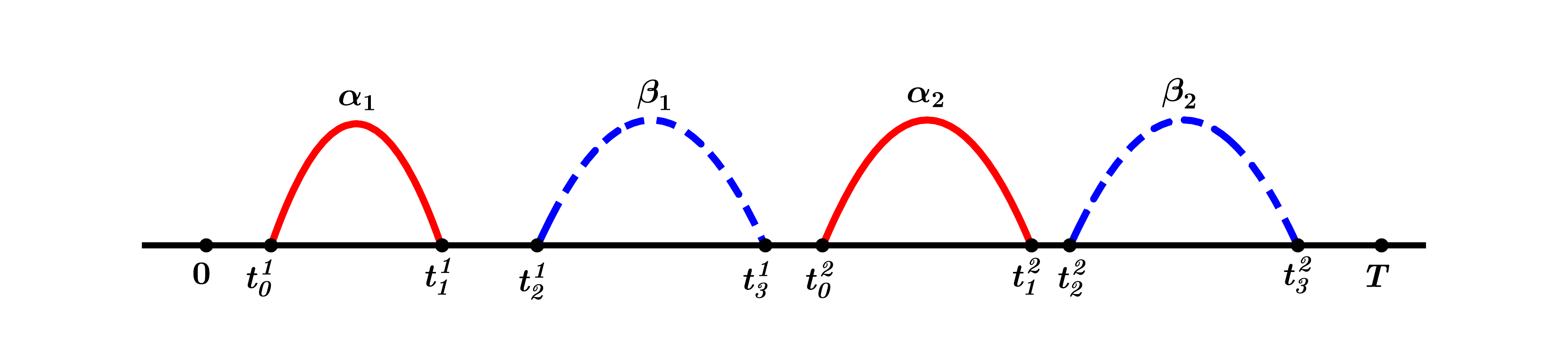}\\
\includegraphics[scale=0.4]{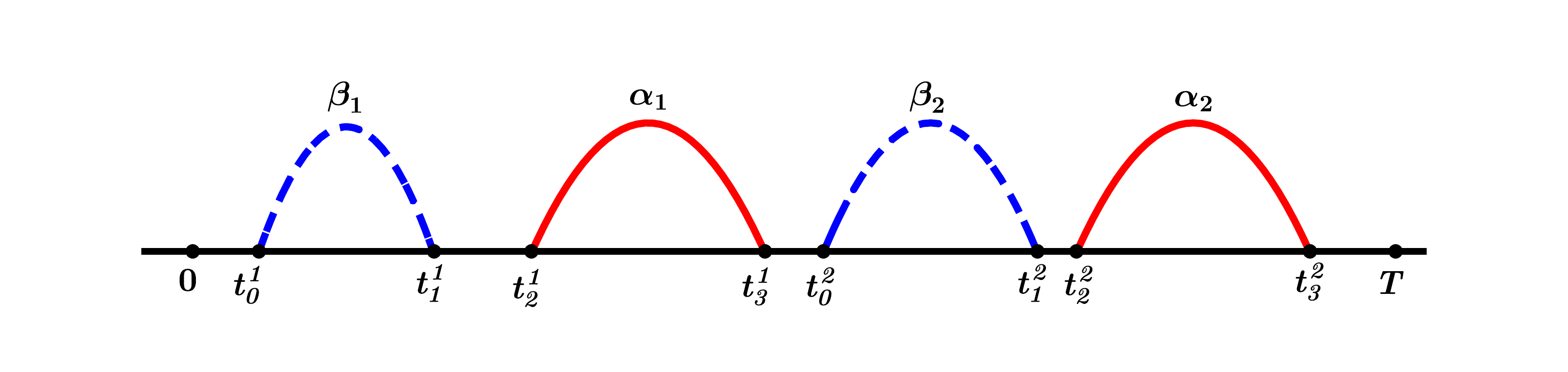}
\caption{Two examples,  with $k=\ell=2$, satisfying \eqref{ii.41}
(above) and \eqref{ii.42} (below).} \label{fig-vii}
\end{figure}

\begin{theorem}
\label{th2.7}
Under the assumption \eqref{ii.41}, or \eqref{ii.42}, \eqref{ii.1} has, at least, two $T$-periodic coexistence states for every $\l>\l_0$, where
\begin{equation}
\label{ii.43}
\l_0:= \sqrt{\frac{(A_1+A_2)(B_1+B_2)}{A_1A_2B_1B_2}}, \qquad
A_i:=\int_{\mathrm{supp}\a_i}\a_i,\quad B_i:= \int_{\mathrm{supp}\b_i}\b_i,\quad i=1,2.
\end{equation}
In general, although $\l>\l_0$ is a sufficient condition, it is far from necessary. Actually, regarding $\l$ as the main bifurcation parameter, $\l_0$ provides with a bifurcation value to $T$-periodic coexistence states of \eqref{ii.1} from the line $(\l,u,v)=(\l,1,1)$, and there are some ranges of values of the parameters, $A_j$, $B_j$, $j\in\{1,2\}$, for which this bifurcation is transcritical.
\end{theorem}
\begin{proof}
Assume \eqref{ii.41}. Then,  integrating \eqref{ii.1} yields
\begin{equation}
\label{ii.44}
u(T)=u_0e^{(1-v_0)\l A_1+(1-v(t_0^2))\l A_2},\qquad
v(T)=v_0e^{(-1+u(t_0^2))\l B_1+(-1+u(T))\l B_2}.
\end{equation}
Thus, a solution $(u(t),v(t))$ of \eqref{ii.1} with initial data $(u_0,v_0)$ is a $T$-periodic coexistence state if, and only, if $u_0,v_0>0$ and $\mc{P}_1(u_0,v_0)=(u_0,v_0)$, i.e, by \eqref{ii.44}, if, and only, if $u_0,v_0>0$ and
\begin{equation*}
(1-v_0)A_1=(v(t_0^2)-1)A_2,\qquad (1-u_0)B_2=(u(t_0^2)-1)B_1,
\end{equation*}
which, again integrating \eqref{ii.1}, is equivalent to
\begin{equation}
\label{ii.45}
\left \{
\begin{array}{ll}
(1-v_0)A_1&=(v_0e^{(-1+u_0e^{(1-v_0)\l A_1})\l B_1}-1)A_2,\\[1ex]
(1-u_0)B_2&=(u_0e^{(1-v_0)\l A_1}-1)B_1.
\end{array}
\right.
\end{equation}
Hence, eliminating $u_0$ from the second equation of \eqref{ii.45},
\begin{equation}
\label{ii.46}
u_0=\frac{B_1+B_2}{B_2+B_1e^{(1-v_0)\l A_1}}.
\end{equation}
So, substituting \eqref{ii.46} into the first equation of  \eqref{ii.45}, it follows that
\begin{equation*}
    (1-v_0)A_1=(v_0e^{(-1+\frac{B_1+B_2}{B_2+B_1e^{(1-v_0)\l A_1}}e^{(1-v_0)\l A_1})\l B_1}-1)A_2.
\end{equation*}
Consequently, naming $x\equiv v_0$ and setting
\begin{equation}
\label{ii.47}
\Phi(x):=x(A_2e^{\frac{B_2(e^{(1-x)\l A_1}-1)}{B_1e^{(1-x)\l A_1}+B_2}\l B_1}+A_1)-(A_1+A_2),
\end{equation}
it is apparent that $\Phi^{-1}(0)$ provides us with the set of $T$-periodic coexistence states of \eqref{ii.1}.   As $\Phi(x)<0$ for all $x\leq 0$, its zeroes are always positive.  Since
\[
    \Phi(0)=-(A_1+A_2)<0,\qquad \Phi(1)=0,
\]
\[
  \Phi(x)>0\quad \hbox{if} \quad x\geq M:=1+A_2/A_1,
\]
and
\[
    \Phi'(1)=A_1+A_2-\lambda^2\frac{A_1A_2B_1B_2}{B_1+B_2},
\]
we find that $\Phi'(1)<0$ if and only if, $\l > \l_0$ (see \eqref{ii.43}). Therefore, \eqref{ii.1}
possesses two $T$-periodic solutions, $(\l,x_\pm)$, with $0<x_-<1$ and $1< x_+< M$. This ends the proof
when \eqref{ii.41} holds. Note that $x_\pm \in (0,M)$.
\par
Assume \eqref{ii.42}. Then, repeating the previous argument it is apparent that the $T$-periodic
coexistence states of \eqref{ii.1} are the zeroes of the map
\begin{equation}
\label{ii.48}
\Psi(x):=x(B_2e^{\frac{A_2(1-e^{(x-1)\l B_1})}{A_1e^{(x-1)\l B_1}+A_2}\l A_1}+B_1)-(B_1+B_2).
\end{equation}
As above, since $\Psi(x)<0$ for all $x\leq 0$, its zeroes are always positive.
Thus, as $\Psi(x)$  satisfies
\[
    \Psi(0)=-(B_1+B_2)<0,\qquad \Psi(1)=0,
\]
\[
  \Psi(x)>0\quad \hbox{if}\quad x\geq N:=1+B_2/B_1,
\]
and
\[
    \Psi'(1)=B_1+B_2-\lambda^2\frac{A_1A_2B_1B_2}{A_1+A_2},
\]
it is apparent that $\Psi'(1)<0$ if $\l >\l_0$. Therefore, in this case, \eqref{ii.1} admits, at least,
two $T$-periodic coexistence states, $(\l,x^{\pm})$, with $0<x^+<1$ and $1<x^-<N$. Note that
$x^\pm \in (0,N)$. This concludes the proof that $\l>\l_0$ is sufficient for the existence
of, at least, two $T$-periodic coexistence states.
\par
It remains to determine the bifurcation directions from $(\l,x)=(\l_0,1)$ in both cases. Now,
it is appropriate to made explicit the dependence of the functions $\Phi$ and $\Psi$ not only on $x$ but also on $\l$, for as $\l$ will be though as a bifurcation parameter.
\par
Assume \eqref{ii.41} and let $\Phi(\l,x)$ denote the function defined by \eqref{ii.47}. Then, the linearization of this function at $(\l,1)$ is given by
\begin{equation}
\label{ii.49}
     \mathfrak{L}(\lambda):=\frac{\partial\Phi }{\partial x}(\lambda,1)=A_1+A_2-\lambda^2\frac{A_1A_2B_1B_2}{B_1+B_2}.
\end{equation}
Thus, using the notations of \cite{LG-2001}, we find that, by the definition of $\l_0$,
\[
  \mf{L}_0:=\mathfrak{L}(\lambda_0) =\frac{\partial\Phi }{\partial x}(\l_0,1)=0,  \qquad \mathfrak{L}_{1}:=\frac{d\mathfrak{L}}{d\lambda}(\lambda_0).
\]
We claim that the next algebraic transversality condition holds
\begin{equation}
\label{ii.50}
\mathfrak{L}_1(N[\mathfrak{L}_0])\oplus R[\mathfrak{L}_0]=\R.
\end{equation}
Indeed, since $\mathfrak{L}_0=0$, it is apparent that $R[\mf{L}_0]=[0]$ and hence,
$N[\mathfrak{L}_0]=\R={\rm span}[1]$. Moreover, differentiating with respect to $\l$ \eqref{ii.49} yields
\[
\mathfrak{L}_1=-\frac{2\lambda_0A_1A_2B_1B_2}{B_1+B_2}\neq0.
\]
Therefore, $\mathfrak{L}_11\notin R[\mathfrak{L}_0]$ and \eqref{ii.50} holds. Consequently, by Theorem 7.1 of
Crandall and Rabinowitz \cite{CrRa-1971}, there exist $\varepsilon>0$ and two analytic functions $\lambda, x:(-\varepsilon,\varepsilon)\to \R$ such that, for some $\l_1\in\R$ to be determined,
$$
   \lambda(s)=\lambda_0+\lambda_1s+\mathcal{O}(s^2), \qquad x(s)=1+s+\mathcal{O}(s^2),\qquad \hbox{as} \;\;
     s\to 0,
$$
and $\Phi(\lambda(s),x(s))=0$ for all $s\in(-\varepsilon,\varepsilon)$. Moreover, besides $(\l,1)$, these
are the unique solutions of \eqref{ii.1} in a neighborhood of $(\l_0,1)$.
\par
Setting
$$
   \v(s):= \Phi(\l(s),x(s)), \qquad |s|<\e,
$$
it is apparent that
\begin{equation*}
  0=\v(s)=\v(0)+\v'(0)s+\frac{1}{2}\v''(0)s^2 + \mc{O}(s^3)\qquad |s|<\e,
\end{equation*}
where $'$ stands for differentiation with respect to $s$. So,
$$
   \v(0)=\v'(0)=\v''(0)=0.
$$
By construction,
$$
  0=\v(0)=\Phi(\lambda_0,1)=0, \qquad \frac{\partial\Phi}{\partial x}(\lambda_0,1).
$$
Moreover, differentiating \eqref{ii.47} with respect to $\l$, it becomes apparent that
$$
  \frac{\partial\Phi}{\partial \lambda}(\lambda_0,1)=0.
$$
Thus, since $\l'(0)=\l_1$,
\[
  \v'(0)= \frac{\partial\Phi}{\partial \lambda}(\lambda_0,1) \lambda_1+\frac{\partial\Phi}{\partial x}(\lambda_0,1)= 0
\]
does not provide any information on the sign of $\lambda_1$. So, we must analyze the second order terms of
$\Phi(\l,x)$ at $(\l_0,1)$. By differentiating $\Phi$, after some straightforward, but tedious,
manipulations, we find that
\[
\frac{\partial^2\Phi}{\partial x^2}(\lambda_0,1)=\l_0^2 \frac{A_1A_2B_1B_2}{(B_1+B_2)^2}\left[(B_1+B_2)(\l_0A_1-2)+\l_0A_1B_1(\l_0B_2-2)\right],
\]
\[
\frac{\partial^2\Phi}{\partial x\partial\lambda}(\lambda_0,1)=-2\lambda_0\frac{A_1A_2B_1B_2}{B_1+B_2}, \qquad \frac{\partial^2\Phi}{\partial\lambda^2}(\lambda_0,1)=0.
\]
Consequently, differentiating and substituting the previous values of the second derivatives, it is apparent that
\begin{align*}
0&=\v''(0) =2\lambda_1\frac{\partial^2\Phi}{\partial x\partial\lambda}(\lambda_0,1)+\frac{\partial^2\Phi}{\partial x^2}(\lambda_0,1)+\frac{\partial^2\Phi}{\partial\lambda^2}(\lambda_0,1)\\&=
\lambda_1\left(-4\lambda_0\frac{A_1A_2B_1B_2}{B_1+B_2}\right)+
\l_0^2\frac{A_1A_2B_1B_2}{(B_1+B_2)^2}\left[(B_1+B_2)(\l_0A_1-2)+\l_0A_1B_1(\l_0B_2-2)\right]
\end{align*}
and therefore,
\begin{equation}
\label{ii.51}
\l_1=\frac{\l_0}{4(B_1+B_2)}\left[(B_1+B_2)(\l_0A_1-2)+\l_0A_1B_1(\l_0B_2-2)\right].
\end{equation}
It is clear that \eqref{ii.51} can reach both positive and negative values depending on the values of the
several parameters  $A_1$, $A_2$, $B_1$ and $B_2$. For instance, if
\begin{equation}
\label{ii.52}
2A_1<A_2,\quad 2B_2<B_1,\quad \frac{2}{3}< \frac{B_2}{A_1}<\frac{3}{2},
\end{equation}
then,
\[
\l_0^2 A_1^2= \frac{(A_1+A_2)(B_1+B_2)}{A_2B_1B_2} A_1 =
\left(\frac{A_1}{A_2}+1\right) \left( 1+ \frac{B_2}{B_1}\right) \frac{A_1}{B_2}
<\left( \frac{3}{2}\right)^3 <4.
\]
Similarly,
\[
 \l_0^2\,B_2^2 =  \frac{(A_1+A_2)(B_1+B_2)}{A_1A_2B_1} B_2= \left(\frac{A_1}{A_2}+1\right) \left( 1+ \frac{B_2}{B_1}\right) \frac{B_2}{A_1} <\left( \frac{3}{2}\right)^3 <4.
\]
Since the estimates \eqref{ii.52} are satisfied, for example, if $B_2=A_1$, $B_1=A_2$ and $2B_2<B_1$,
it becomes apparent that \eqref{ii.52} holds for wide open ranges of values of the several parameters
involved in the setting of \eqref{ii.1}.
\par
Now, assume \eqref{ii.42}, instead of \eqref{ii.41}. Then, the $T$-periodic coexistence states are given by the zeros of the map $\Psi(\l,x)$ defined in \eqref{ii.48}. In this case, the linearization of
$\Psi(\l,x)$ at $(\l,1)$ is given by
\[
   \mathfrak{M}(\l):=\frac{\partial \Psi}{\partial x}(\l,1)=B_1+B_2-\lambda^2\frac{A_1A_2B_1B_2}{A_1+A_2}.
\]
Thus, setting
\[
\mathfrak{M}_0:=\mathfrak{M}(\l_0)=0,\qquad \mathfrak{M}_1:=\frac{d \mathfrak{M}}{d\l}(\l_0),
\]
and adapting the argument given above, it is apparent that the transversality condition
\[
\mathfrak{M}_1(N[\mathfrak{M}_0])\oplus R[\mathfrak{M}_0]=\R
\]
holds true. Moreover, also
$$
        \frac{\partial\Psi}{\partial x}(\lambda_0,1)=0=\frac{\partial\Psi}{\partial \lambda}(\lambda_0,1).
$$
Hence, to find out the bifurcation direction, we must proceed as in the previous case.
A rather straightforward, but tedious, calculation shows that
\[
\frac{\partial^2\Psi}{\partial x^2}(\lambda_0,1)=\l_0^2 \frac{A_1A_2B_1B_2}{(A_1+A_2)^2}\left[(A_1+A_2)(-\l_0B_1-2)+\l_0A_1B_1(\l_0A_2-2)\right],
\]
\[
\frac{\partial^2\Psi}{\partial x\partial\lambda}(\lambda_0,1)=-2\lambda_0\frac{A_1A_2B_1B_2}{A_1+A_2},\qquad \frac{\partial^2\Psi}{\partial\lambda^2}(\lambda_0,1)=0.
\]
Therefore,
\[
\l_1=\frac{\l_0}{4(A_1+A_2)}\left[-(B_1+B_2)(\l_0A_1+2)+\l_0A_1B_1(\l_0B_2-2)\right],
\]
which can reach negative values also in this case. Indeed, if
\[
  B_2 < A_1 <B_1,\quad A_1<A_2,
\]
then
\[
\l_0^2\,B_2^2=\frac{(A_1+A_2)(B_1+B_2)}{A_1A_2B_1}B_2 <4
\]
and consequently $\l_1<0$. This concludes the proof.
\end{proof}

Note that Theorem \ref{th2.7} generalizes Theorems \ref{th2.1} and
\ref{th2.2}. Indeed, if $A_1=A_2$ and $B_1=B_2$, then
$\varphi(x)=\Phi(x)$ and $\psi(x)=\Psi(x)$. Thus,
\[
  \l_0= \frac{2}{\sqrt{A_1B_1}},
\]
which provides us with Theorems \ref{th2.1} and \ref{th2.2}. Moreover, by Lemma \ref{le2.2}, Theorem \ref{th2.7} holds also for the cases
\begin{equation}
    \label{ii.53}
    \supp\a_i\subseteq[t_0^i,t_1^i]\;\;\hbox{and}\;\;\supp\b_j\subseteq[t_2^j,t_3^j],
\end{equation}
and
\begin{equation}
    \label{ii.54}
    \supp\b_i\subseteq[t_0^i,t_1^i]\;\;\hbox{and}\;\;\supp\a_j\subseteq[t_2^j,t_3^j],
\end{equation}
with $i\in\{1,2,3\}$,
$j\in\{1,2\}$ and for some
\[
0\leq t_0^1<t_1^1\leq t_2^1<t_3^1\leq t_0^2<t_1^2\leq t_2^2<t_3^2\leq t_0^3<t_1^3\leq T.
\]
If \eqref{ii.53} holds, then
\[
\l_0=\sqrt{\frac{(B_1+B_2+B_3)(A_1+A_2)}{B_2(B_1+B_3)A_1A_2}},
\]
while
\[
\l_0=\sqrt{\frac{(A_1+A_2+A_3)(B_1+B_2)}{A_2(A_1+A_3)B_1B_2}}
\]
if \eqref{ii.54} holds.

\section{Chaotic dynamics}\label{sec3}

\noindent In this section we consider again the non-autonomous
Volterra predator-prey model
\begin{equation}
\label{iii.1}
\left \{
\begin{array}{ll}
u'=\alpha(t)u(1-v),\\ v'=\beta(t)v(-1+u),
\end{array}
\right.
\end{equation}
where $\alpha\gneq 0$, $\beta\gneq 0$ are $T$-periodic continuous
functions for some $T>0$, with the aim to prove the presence of
chaotic-like dynamics.
\par
As in Section \ref{sec2},  attention is focused on the
\emph{coexistence states}, namely the component-wise positive
solutions of the system \eqref{iii.1}. As already explained in
Section \ref{sec1} (cf. \eqref{i.2}), in this kind of systems it
is natural to perform the change of variables
\[
 x=\log u, \qquad y=\log v,
\]
which moves the equilibrium point $(1,1)$ to the origin $(0,0)$ of the phase-plane.
In the new variables $x$ and $y$,
the model \eqref{iii.1} turns into the next equivalent planar Hamiltonian system
\begin{equation}
\label{iii.2}
\left \{
\begin{array}{ll}
x'=-\alpha(t)(e^y-1),\\ y'=\beta(t)(e^x-1).
\end{array}
\right.
\end{equation}
Thus, in this section, instead of looking for coexistence states of \eqref{iii.1},
we will simply look for solutions of \eqref{iii.2}.
\par
The existence of complex dynamics for prey-predator equations has been studied
since the Eighties, mostly from a numerical point of view
(see Takeuchi and Adachi \cite{TaAd-1984}).
Evidence of chaos has been detected in numerical simulations for
three-dimensional (or higher-dimensional) autonomous systems,
for instance for the interaction of one predator with two preys or the case of a prey,
a predator and a top predator,  as well as introducing as a third variable a nutrient.
\par
For the two-dimensional case results have been obtained for discrete models of Holling type,
as, e.g.,  those of  Agiza et al. \cite{Ag-2009}, which are in line with the classical
works of May \cite{May} and Li and Yorke \cite{LiYo-1975},
when proving chaos for discrete single species logistic-type equations.
Other examples of chaos for two-dimensional systems have been numerically produced by adding some
delay effects in the equations, as in  Nakaoka, Saito and Takeuchi \cite{NST-2006}.
\par
But less results are available in the literature concerning chaotic-like solutions
to planar predator-prey systems with periodic coefficients
(see, for instance, Baek and Do \cite{BD-2009}, Broer et al. \cite{Br-2006},
Kuznetsov, Muratori and Rinaldi \cite{KMR-1992} and Volterra \cite{VSB-2001}).
Typical features of these models is to add logistic and/or Holling
growth effects on the prey population and assume that the intraspecific growth rate takes the form
$r(1+\varepsilon\sin(\omega t))$. The special choice of the periodic coefficient allows to study numerically
the bifurcation diagrams
to  give evidence of complex dynamics for some choices of the parameters.
\par
Up to the best of our knowledge, only very few results provide a complete analytic proof
of the presence of chaotic dynamics, without the need of numeric support, for the classical
Volterra predator-prey system with periodic coefficients.
In Pireddu and Zanolin \cite{PiZa-2008} the Volterra original
system with harvesting effects was considered and it was proved that
chaos may arise under special forms of a periodic harvesting;
these authors also considered  a case of intermittency in the predation
in \cite[pp.221-225]{PiZa-2013}.
Also in Ruiz-Herrera \cite{RH-2012}, a rigorous analysis of chaos for periodically perturbed planar systems,
was performed for a case in which the Volterra system switches periodically
to new system with logistic effects and no predation.
\par
It is the aim of this section to show a simple mechanism producing chaotic dynamics for
system \eqref{iii.2}, just assuming that the coefficient $\alpha(t)$
vanishes on some interval, regardless the length of the vanishing interval,
which might have dramatic consequences from the point
of view of the applications.
Actually, we will prove the following main result which holds for the
more general system
(cf. \eqref{i.2})
\begin{equation}
\label{iii.3}
\begin{cases}
x' = -\lambda\alpha(t)f(y),\\
y' = \lambda\beta(t)g(x),
\end{cases}
\end{equation}
where $f,g:{\mathbb R}\to {\mathbb R}$ are $C^1$-continuous functions with
$f(0)=g(0)=0,$ $f'(0), g'(0) >0$ and $f(s)s>0, g(s)s>0$ for $s\not=0.$
We also assume that at least one of the two functions
is bounded in a neighborhood of $+\infty$ or $-\infty.$ To be more specific
and just to fix a possible case, we will suppose that $f$ is bounded on $(-\infty,0].$
In the application to the predator-prey model we have
$$f(s) = g(s) = e^s-1.$$

The introduction of the parameter $\lambda >0$ in \eqref{iii.3}
is not relevant from the mathematical point of view. For us it is convenient in order to
make a comparison to the result about subharmonic solutions obtained
by the authors in \cite{LMZ-2020}.

\begin{theorem}\label{th3.1}
Assume that there exists  $T_0\in (0,T)$ such that:
\begin{itemize}
\item[$(c_1)\;$]
$\alpha \gneq 0$, $\beta \gneq 0$ on $[0,T_0]$ and there exists
$\hat{t}\in [0,T_0]$ with $\alpha(\hat{t})\beta(\hat{t}) > 0;$
\item[$(c_2)\;$]
$\alpha \equiv 0$ and $\beta \gneq 0$ on $[T_0,T].$
\end{itemize}
Then, for every $\ell\geq 2,$ there exists $\lambda^*=\lambda^*_{\ell}$ such that,
for each $\lambda >\lambda^*$, there exists a constant $K_{\lambda}$ such that,
if
\begin{equation}\label{iii.4}
\int_{T_0}^{T} \beta(t)\,dt > K_{\lambda},
\end{equation}
then, the Poincar\'{e} map associated with \eqref{iii.3} induces chaotic dynamics on
$\ell$-symbols on some compact subset, ${\mathcal Q}$, of  the first quadrant.
\end{theorem}

The proof of Theorem \ref{th3.1} is postponed to a next subsection. From the proof,
it will become apparent  how to determine the constants $\lambda^*$ and $K_{\lambda}.$
The assumption that $f$ and $g$ are $\mc{C}^1$ can be weakened,
by assuming, as in \cite{LMZ-2020}, that $f, g$  are locally Lipschitz with
\begin{equation}\label{iii.5}
0 < \liminf_{s\to 0}\frac{f(s)}{s} \leq \limsup_{s\to 0}\frac{f(s)}{s}  < +\infty,\quad
0 < \liminf_{s\to 0}\frac{g(s)}{s} \leq \limsup_{s\to 0}\frac{g(s)}{s} < +\infty.
\end{equation}
Up to the best of our  knowledge, this is
the most general result available for complex dynamics to the predator-prey equations with
periodic coefficients.  In fact, the previous theorems of Pireddu and Zanolin
\cite{PiZa-2008,PiZa-2013} and Ruiz-Herrera \cite{RH-2012} required very specific
structural assumptions on the coefficients,
which were assumed to be stepwise, in order to transform Volterra equation to a switched system.
Nevertheless,  we will also present a less general version of Theorem \ref{th3.1} where
$\beta=\text{constant} >0,$  $\alpha$ is a piecewise  constant function vanishing
on $[T_0,T]$, according to $(c_2)$,  and $f(s)=g(s)=e^s -1$.
This more elementary and special case is introduced in the following subsection in order to better
explain the geometry and the dynamics associated to our system.
Actually, for expository reasons, this section has been split into three parts.
In the first subsection, we perform a detailed analysis of the equation with a stepwise coefficient.
Then, we show how the the same geometrical ideas
can be adapted to prove Theorem \ref{th3.1}.
Finally, in the last subsection, we will recall some of  the main features
of the Smale's horseshoe, adapted to our situation here,
in order to discuss a possible further improvement of our results from
a numerical point of view.

\subsection{The simplest model of type \eqref{iii.2} with chaotic dynamics}
\label{sec3.1}

\noindent In this section  we consider a special case of Theorem
\ref{th3.1} that already exhibits all the significant geometrical
features of the main result. Precisely, for a given  $T_0\in
(0,T)$, we will consider the $T$-periodic functions $\a(t)$ and
$\b(t)$ defined by
\begin{equation}
\label{iii.6}
\b(t):= \left\{
\begin{array}{lll}
\b_0>0&\;\;\hbox{if}\; t\in[0,T_0),\\[1ex]
\b_1>0&\;\;\hbox{if}\; t\in [T_0,T),
\end{array}
\right. \quad \hbox{and}\quad \a(t):= \left\{
\begin{array}{lll}
\a>0&\;\;\hbox{if}\; t\in[0,T_0),\\[1ex]
0&\;\;\hbox{if}\; t\in [T_0,T),
\end{array}
\right.
\end{equation}
where the positive constants $\a$ and $\b_0,\b_1$ and the exact
values of $T_0$ and $T$ are going to be made precise later.
Also, we will set $T_1:=T-T_0$. Thus, the dynamics of \eqref{iii.2}
on each of the intervals $[0,T_0]$ and $[T_0,T]$ are those of the
associated autonomous Volterra predator-prey systems in the
intervals $[0,T_0]$ and $[0,T_1],$ respectively.
\par
Although the function  $\a(t)$, and
possibly $\b(t)$, has a jump at $t=T_0$, for any given $z:=(x_0,y_0)\in\R^2$,
the Poincar\'{e} map associated to the system \eqref{iii.2} in the
interval $[0,T]$ is well-defined as
\begin{align*}
\Phi:\R^2&\rightarrow\R^2\\ z\;&\mapsto \Phi(z):=(x(T;z),y(T;z)),
\end{align*}
where $(x(t;z),y(t;z))$ stands for the unique solution of \eqref{iii.2} such that
$(u(0;z),y(0;z))=z$, and it is a diffeomorphism. Under the assumption \eqref{iii.6},
the action of system \eqref{iii.2} can be regarded as a composition of the actions of the systems
\begin{equation}
\label{iii.7}
\hbox{(I)}\quad \left \{
\begin{array}{ll}
x'=\a(1-e^y)\\
y'=-\b_0(1-e^x)
\end{array}
\right.\hspace{2cm}
\hbox{(II)} \quad \left \{
\begin{array}{ll}
x'=0\\
y'=-\b_1(1-e^x)
\end{array}
\right.
\end{equation}
on each of the intervals $[0,T_0]$ and $[T_0,T]$, respectively.
In this manner, system \eqref{iii.2} turns out to be a \textit{switched system} with
a $T$-periodic \textit{switching signal} and with \eqref{iii.7}-(I) and
\eqref{iii.7}-(II) as \textit{active subsystems}, according to the terminology adopted
by Liberzon \cite{Li-2003}.  In other words,
the Poincar\'{e} map $\Phi$ is the composition of the two Poincar\'{e} maps associated to each of these systems,
\[
\Phi:=\Phi_{T_1}\circ\Phi_{T_0},
\]
where $\Phi_{T_0}$ and $\Phi_{T_1}$  stand for the Poincar\'{e}
maps associated to the first and second systems of \eqref{iii.7}
on the intervals $[0,T_0]$ and $[0,T_1]$, respectively.
\par
Subsequently, we denote by $\t(t,z)$ the angular polar
coordinate at time $t\geq0$ of the solution $(x(t;z),y(t;z))$ for system \eqref{iii.7}-(I).
Then, for any given $\varrho>0$, the rotation number of the solution
in the interval $[0,\varrho]$ is defined through
\[
{\rm rot}([0,\varrho],z):=\frac{\t(\varrho,z)-\t(0,z)}{2\pi}.
\]
It is an algebraic counter, modulo $2\pi$, of the winding number
of the solution around the origin. It can be equivalently
expressed using \eqref{i.13} on the first subsystem.

\bigskip

\begin{center}
\textbf{Dynamics of \eqref{iii.2} in the interval $[0,T_0]$}
\end{center}

\noindent
We begin by analyzing the dynamics of \eqref{iii.2} on $[0,T_0]$
under the action of $\Phi_{T_0}$, i.e.,  the dynamics of \eqref{iii.7}-(I).
It is folklore that the phase-portrait of the (autonomous) system is a global nonlinear center
around the origin, i.e., every solution different from the equilibrium $(0,0)$
is periodic and determines a closed curve around the origin;
the \emph{first integral}, or \emph{energy function}, of the system being
\begin{equation}
\label{iii.8}
\mc{E}(x,y)=\a(e^y-y)+\b_0(e^x-x).
\end{equation}
Thus, setting
\[
\o:=\mc{E}(0,0)=\a+\b_0=\min_{\R ^2}\mc{E},
\]
for every $\ell>\o$, the corresponding level line of the first integral,
$\Gamma(\ell)=\mc{E}^{-1}(\ell)$, is a closed orbit of a periodic solution.
Moreover, by simply having a glance at the system,
it is easily realized that the solutions run counterclockwise
around the origin.
\par
Subsequently, for every $\ell>\o$, we will denote by $\tau(\ell)$
the  (minimal) period of the orbit $\Gamma(\ell)$. Thanks to a
result of Waldvogel \cite{Wa-1986}, the fundamental period map
$\tau\,:(\o,+\infty)\rightarrow\R$ is increasing and it satisfies
\[
\lim_{\ell\downarrow \o}\tau(\ell)=\frac{2\pi}{\sqrt{\a
\b_0}},\qquad \lim\limits_{\ell\uparrow
+\infty}\tau(\ell)=+\infty.
\]
For the rest of this section, we fix $\ell_1>\o$ and choose
\begin{equation}
\label{iii.9}
T_0:=2\tau(\ell_1).
\end{equation}
If a solution of the system \eqref{iii.7}-(I)  crosses entirely the third
quadrant, then, there exists an interval $[t_0,t_1]\subseteq[0,T_0]$, with  $x(t_0)<0$, $y(t_0)=0$,
$x(t_1)=0$ and $y(t_1)<0$, such that $x(t)<0$ and  $y(t)<0$ for all $t\in (t_0,t_1)$.
Thus, for every  $t\in[t_0,t_1]$, we have that
\begin{align*}
|y(t)| & =|\int_{t_0}^{t}\b_0(e^{x(s)}-1)\,ds|\leq
\int_{t_0}^{t_1}\b_0\,ds\leq\b_0 T_0 =: M,\\
|x(t)| & =|\int_{t}^{t_1}\a(1-e^{y(s)})\,ds|\leq \int_{t_0}^{t_1}\a\, ds\leq\a T_0=:N.
\end{align*}
Hence, if there exists a $\tilde{t}\in[0,T_0]$ such that
\[
x(\tilde t)<0,\quad y(\tilde t)< 0, \quad
x^2(\tilde{t})+y^2(\tilde{t})>M^2+N^2,
\]
then the solution cannot cross entirely the third quadrant in $[0,T_0]$, though, as all the solutions
are periodic around the origin, all must cross the third quadrant in a sufficiently large time.
Therefore, given $(x_2,y_2)$ such that
$$
   x_2<0,\quad y_2<0,\quad x_2^2+y_2^2>M^2+N^2
$$
and setting
$\ell_2:=\mc{E}(x_2,y_2)$ (and, without loss of generality, with $\ell_2 >\ell_1$),
it becomes apparent that, if $z=(0,y_0)\in\Gamma(\ell_2)$ with $y_0>0$, then
\[
\t(t,z)\in[\pi/2,3\pi/2) \quad\hbox{for all} \;\; 0\leq t\leq T_0,
\]
because the orbit through $(x_2,y_2)$, $\G(\ell_2)$, cannot cross the entire third quadrant
in the time interval $[0,T_0]$. Consequently, by \eqref{iii.9}, we find that
\begin{equation}
\label{iii.10}
\left\{
\begin{array}{ll}
{\rm rot}([0,T_0],z)=2&\qquad\hbox{if}\; z\in\Gamma(\ell_1),\\[1ex]
{\rm rot}([0,t],z)<1/2 \;\; \forall\, t\in [0,T_0]&\qquad\hbox{if}\; z=(0,y_0)\in\Gamma(\ell_2),\,y_0>0.
\end{array}
\right.
\end{equation}
\par
Subsequently, in order to analyze the Poincar\'{e} map $\Phi_{T_0}$,
we will focus attention into the annular region,  $\mc{A}$, of the phase-plane
enclosed by the orbits  $\Gamma(\ell_1)$ and $\Gamma(\ell_2)$, where \eqref{iii.10} holds,
which has been represented in Figure \ref{fig-viii},  i.e.,
\[
\mc{A}:=\left\{(x,y)\in\R^2\,:\,\ell_1\leq \mc{E}(x,y)\leq \ell_2\right\}=
\bigcup_{\ell_1\leq\ell\leq\ell_2}\Gamma(\ell).
\]

\begin{figure}[h!]
\centering
\includegraphics[scale=0.45]{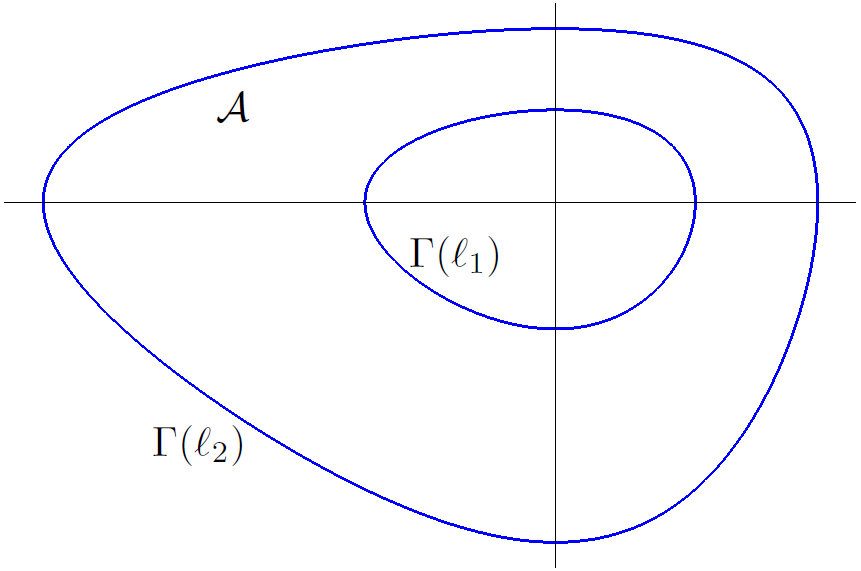}
\caption{The region $\mathcal{A}$ enclosed by curves
$\Gamma(\ell_1)$ and $\Gamma(\ell_2)$ with $\omega<\ell_1<\ell_2$.} \label{fig-viii}
\end{figure}

For graphical purposes, the aspect ratios in Figure \ref{fig-viii} have been slightly modified.
\par
By the dynamical properties  of  \eqref{iii.7}-(I), the Poincar\'{e} map $\Phi_{T_0}$ transforms the portion of $\mc{A}$ on the positive $y$-axis,  i.e., the segment
\begin{equation}
\label{iii.11} \s_0:=\mc{A}\cap
\{(0,y)\in\R^2\,:\;y>0\}=\left\{(0,y) \in\R \times (0,\infty)\;
:\;\ell_1\leq\mc{E}(0,y)\leq \ell_2\right\},
\end{equation}
into the spiraling line  plotted in Figure \ref{fig-ix}(a).
The unique $(0,y_1)$ such that $\mc{E}(0,y_1)=\ell_1$ remains invariant by
$\Phi_{T_0}$ because $T_0=2\tau(\ell_1)$. Thus, $(0,y_1)$ gives
two rounds around the orbit $\Gamma(\ell_1)$ as $t\in [0,T_0]$.
Since the period map $\tau(\ell)$ is increasing with respect to
$\ell$, the points $(0,y)\in\s_0$ with $y>y_1$  close to $y_1$
cannot complete two rounds around the origin, though close to get
it. The bigger is taken $y>y_1$, the bigger is the gap
$\tfrac{\pi}{2}+4\pi-\theta(T_0,(0,y))$, until $y$ approximates
$y_2$, the unique value of $y$ such that $\mc{E}(0,y_2)=\ell_2$,
where, according to the choice of $\Gamma(\ell_2)$, we already
know that $\theta(T_0,(0,y))<3\pi/2$. Similarly, defining
$(x_{+}(\ell_1),0)$ as the intersection of $\Gamma(\ell_1)$ with
the positive $x$-axis, Figure \ref{fig-ix}(b) shows a plot of the
parallel (vertical) segment
\begin{equation}
\label{iii.12} \s_{1}:=\mc{A}\cap (\{x_{+}(\ell_1)\}\times
[0,\infty)) =\{(x_{+}(\ell_1),y)\in\R\times
[0,\infty)\,:\;\ell_1\leq\mc{E}(x_{+}(\ell_1),y)\leq \ell_2\}.
\end{equation}
\par

\begin{figure}[h!]
\centering
\includegraphics[scale=0.44]{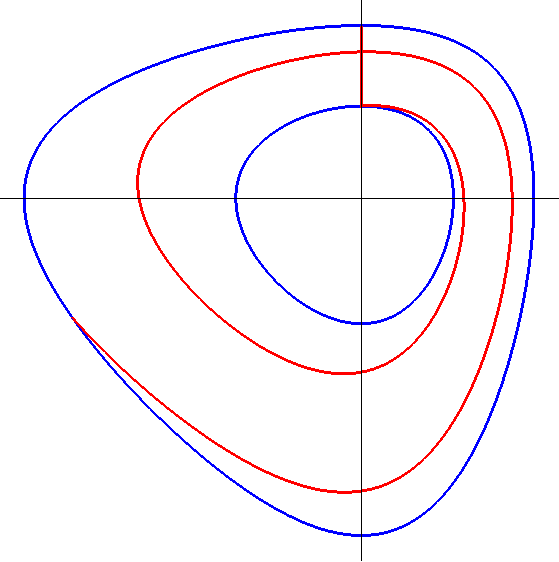}
\qquad
\includegraphics[scale=0.44]{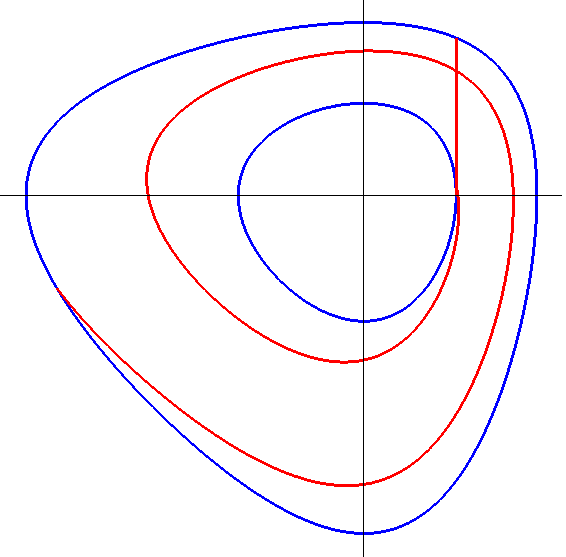}
\caption{The segments $\s_{i}$ and the curves $\Phi_{T_0}(\s_{i})$
for $i=0$  (left) and $i=1$ (right).} \label{fig-ix}
\end{figure}

The curves $\Phi_{T_0}(\s_i)$ look like sort of logarithmic
spirals with the angular polar coordinate increasing along their
trajectories. As illustrated by Figure \ref{fig-ix}(b),
$\Phi_{T_0}(\s_1)$ is a curve looking like $\Phi_{T_0}(\s_0)$.

Throughout the rest of this section, we consider the topological
square, $\mc{Q}$,  enclosed by the segments $\s_0$ and $\s_{1}$ in
$\mc{A}$, i.e.,
$$
\mc{Q}:= \{(x,y)\in\mc{A}\,:\,0\leq x\leq x_+(\ell_1),\, y>0\}.
$$
The plot of Figure \ref{fig-x} shows $\Phi_{T_0}(\mc{Q})$, which
is the spiral-like region enclosed in the annuls ${\mathcal A}$
and bounded by the curves $\Phi_{T_0}(\sigma_0)$ and
$\Phi_{T_0}(\sigma_1).$
Subsequently, we will also consider the topological square
$$
\mc{R}:=\{(x,y)\in\mc{A}\,:\,0\leq x\leq x_1,\, y<0\},
$$
which has been also represented in  Figure \ref{fig-x},  where $\Phi_{T_0}(\mc{Q})\cap\mc{R}$
consists of two smaller rectangular regions which have been named
as $\mc{R}_0$ and $\mc{R}_1$.

\begin{figure}[h!]
\centering
\includegraphics[scale=0.4]{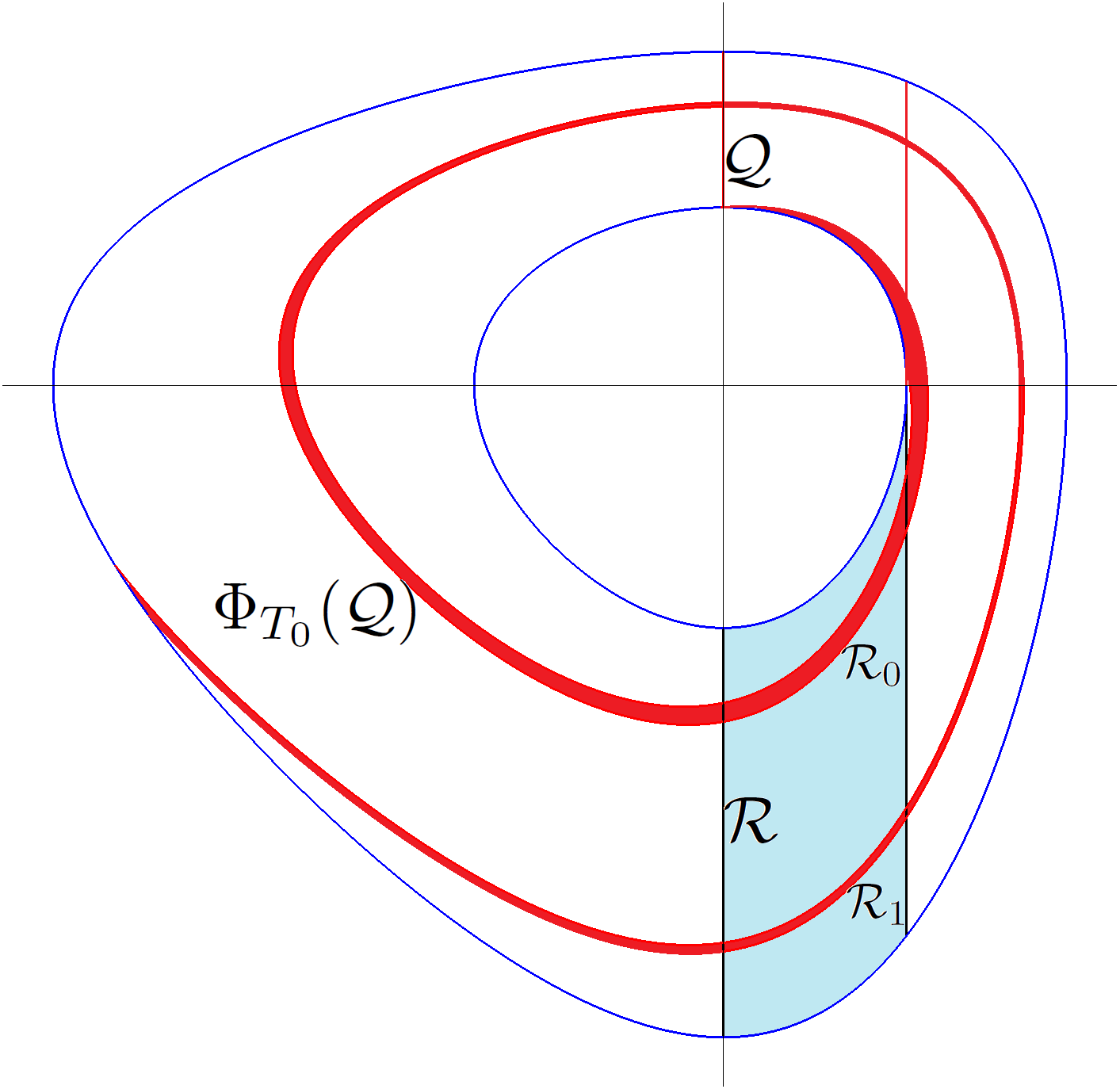}
\caption{The regions $\mc{Q}$, $\mc{R}$, $\Phi_{T_0}(\mc{Q})$, $\mc{R}_0$
and $\mc{R}_1$.} \label{fig-x}
\end{figure}

\begin{remark}\label{re3.1}
{\rm For convenience in the exposition, we have chosen $\ell_1$
and $T_0$ to satisfy \eqref{iii.10}. But one can
adjust the parameters to have $T_0\geq j\tau{\ell_1}$ for some integer $j\geq 2$, of course.
In this case, we should  modify the choice of $\ell_2$ in order to get
the second condition of \eqref{iii.10}. Now, we will have
${\rm rot}([0,T_0],z)\geq j$ if $z\in\Gamma(\ell_1)$ as first condition.
Indeed, the intersection of $\Phi_{T_0}(\mc{Q})$ with $\mc{R}$
consists of $j\geq 2$ rectangular regions.}
\end{remark}

\bigskip

\begin{center}
\textbf{Dynamics of \eqref{iii.2} in the interval $[T_0,T]\equiv
[0,T_1].$}
\end{center}

\noindent Throughout this paragraph we recall that $T_1:=T-T_0$ and
consider the Poincar\'{e} map $\Phi_{T_1}$. Choosing  $\a=0$ in
the interval $[T_0,T]$ our main goal in this section is to show
that, for sufficiently large $\b_1
T_1=\int_{T_0}^{T}\beta(t)\,dt>0$, the region $\mc{R}$ is mapped
across $\mc{Q}$ by the Poincar\'{e} map $\Phi_{T_1}.$ Actually we
have that $\Phi_{T_1}(\mc{R})$ intersects transversally $\mc{Q}$.
Such transversality entails a complex behavior reminiscent of
Smale's horseshoe.
\par
Since $\alpha(t)=0$ and $\beta(t)=\b_1$ for all $t\in[0,T_1]$ in
\eqref{iii.7}-(II), we have that
\begin{equation}
\label{iii.13}
x'(t)=0\;\;\hbox{and}\;\; x(t)=x_0 \;\;\hbox{for all}\;\; t\in[0,T_1].
\end{equation}
 Thus,
\begin{equation}
\label{iii.14} y(T_1)-y(0)=\int_{0}^{T_1}
\b_1(e^{x_0}-1)dt=\b_1(e^{x_0}-1)T_1.
\end{equation}
By \eqref{iii.13} and \eqref{iii.14}, we find that $\Phi_{T_1}(z)=z$ if $z=(0,y_0)\in\mc{R}$, and, hence,
\[
\Phi_{T_1}(\{(x,y)\in\mc{R}\,:\,x=0\})=\{(x,y)\in\mc{R}\,:\,x=0\}.
\]
In other words, the left side of $\mc{R}$ consists of fixed points of $\Phi_{T_1}$.

\begin{figure}[h!]
\centering
\includegraphics[scale=0.2]{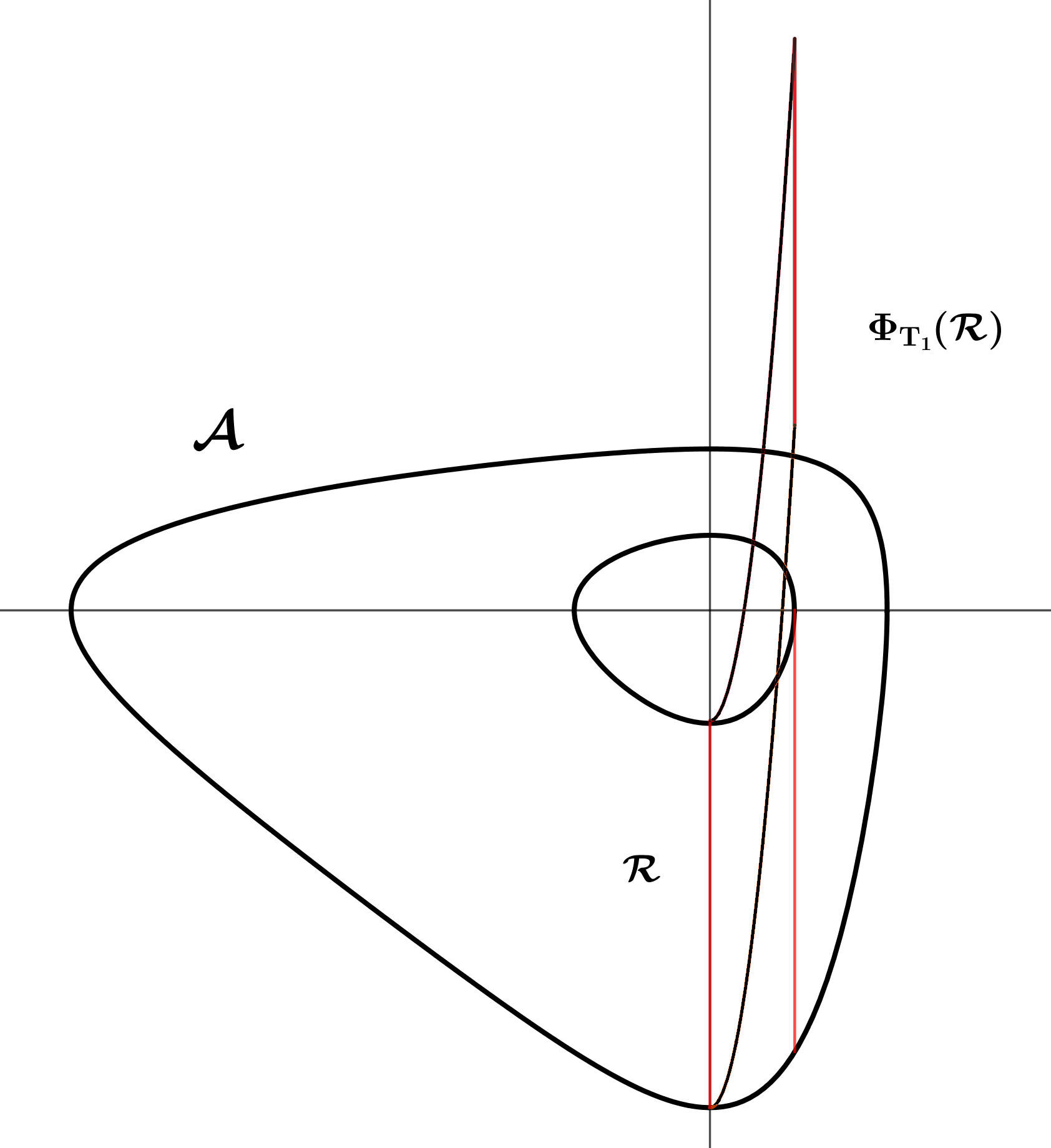}\qquad
\includegraphics[scale=0.35]{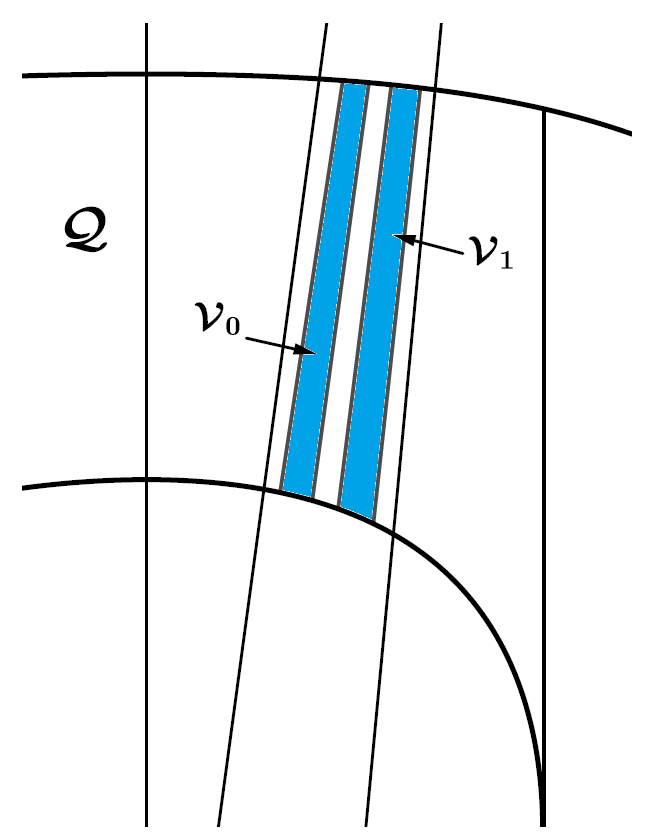}
\caption{The topological squares $\mc{R}$ and $\Phi_{T_1}(\mc{R})$ (left panel), as well as the squares
$\mc{V}_0:=\Phi_{T_1}(\mc{R}_0)\cap\mc{Q}$ and $\mc{V}_1:=\Phi_{T_1}(\mc{R}_1)\cap\mc{Q}$ (right panel). }
\label{fig-xi}
\end{figure}

On the other hand, by \eqref{iii.13}, for every
$z=(x_1,y)\in\mc{R}$, with $0 <x_1 \leq x_+(\ell_1),$ we have that
$x(t)=x_1$ for all $t\in [0,T_1]$, and hence
\[
y(T_1)-y(0)=\int_{0}^{T_1} \b_1(e^{x_1}-1)dt=\b_1(e^{x_1}-1)T_1.
\]
Consequently, if we denote by $y_2^-<0$ and $y_2^+>0$ the unique
values of $y$ such that $\mc{E}(0,y_2^\pm)=\ell_2$, then, setting
$M^*=M^*(\ell_2):=\max\{|y_2^-|,|y_2^+|\}$ and choosing $T_1$
satisfying
\begin{equation} \label{iii.15}
    T_1>\frac{2M^*}{\b_1(e^{x_+}-1)}, \quad \text{for }
    x_+:=x_+(\ell_1),
\end{equation}
it becomes apparent that $y(T_1)-y(0)>2M^*$. So,
\[
 \Phi_{T_1}(\{(x,y)\in\mc{R}\,:\,x=x_+\})\subsetneq\{x_+\}\times (y_2^+,\infty).
\]
Therefore,  for $i=1,2$, the topological rectangle
$\mc{V}_i:=\Phi_{T_1}(\mc{R}_i)\cap\mc{Q}$ crosses  $\mc{Q}$,
transversally, as represented in the second picture of Figure
\ref{fig-xi}. As \eqref{iii.15} holds for sufficiently large
$\b_1>0$, regardless the size of $T_1$, as a rather direct
consequence of the abstract theory of Papini and Zanolin
\cite{PaZa-2004b,PaZa-2004a}, it becomes apparent that, for every
$T_0\in (0,T)$, the problem \eqref{iii.1}, with the special choice
\eqref{iii.6}, exhibits complex dynamics for sufficiently large
$\b_1>0$. Actually, the geometry of the problem is very similar to
the one considered by  Pascoletti, Pireddu and Zanolin \cite[Figs.
6, 7, 8]{PPZ-2008} and Labouriau and Sovrano \cite[Def.
3.3]{LaSo-2020}, consisting of a twist map acting in an annular
region, composed with a shift map on a strip. The type of chaotic
dynamics which occurs is that stated in Definition \ref{def1.1}
with the semi-conjugation to the Bernoulli shift on two symbols.
Actually, we can produce a semi-conjugation with respect to a
larger set of $\ell$ symbols by suitably adapting the parameters
$(\alpha,\beta_0,\beta_1)$ as well as $T_0$ and $T_1$, as already
sketched in Remark \ref{re3.1}. This concludes the proof of
Theorem \ref{th3.1} for  stepwise constant coefficients.
\par
We omit the  technical details regarding the application of the results
from Papini and Zanolin \cite{PaZa-2004b,PaZa-2004a} for the special choice \eqref{iii.6}, since
we will  present this approach in a more thoroughly manner along the proof
of Theorem \ref{th3.1}. On the other hand, the case of stepwise constant coefficients
in \eqref{iii.6} and \eqref{iii.7}-(I)\&(II) suggests that, at least from a numerical
point of view, we are in a situation  where  a stronger result about
chaotic dynamics can be obtained, namely the \textit{conjugation} to the Bernoulli
automorphism, due to the presence of a Smale horseshoe.
Indeed, the geometry that we have described above  suggests a rather elementary mechanism
to generate complex dynamics, by
adopting the original methodology of Smale \cite{Sm-1965}, as
illustrated in the  Conley--Moser approach in  \cite[Ch. III]{Mo-1973}.
In the next subsection, we will prove Theorem \ref{th3.1} in its more general form, using the theory of
topological horseshoes. Then, we will end this paper  by discussing some
possible sharper results in the frame of the original horseshoe geometry,
by assuming \eqref{iii.1} with the special
choice of coefficients in \eqref{iii.6}.

\subsection{Proof of Theorem \ref{th3.1}}\label{sec3.2}

\noindent As already discussed recalled in Section \ref{sec1}, in
order to prove the presence of chaotic dynamics according to
Definition \ref{def1.1} there are various different approaches of
topological nature, though these results provide a weaker form of
chaos with respect to classical Smale's horseshoe, because they
guarantee the semi-conjugation to the Bernoulli shift
automorphism, instead of the conjugation.  However, the approaches
based on the the so-called theory of topological horseshoes,
guarantee a broader range of applications. Here we briefly recall
some basic facts from the ``stretching along the path method'', by
specializing our presentation to planar homeomorphisms.
\par
Let $\Phi: {\mathbb R}^2\to {\mathbb R}^2$ be a  planar
homeomorphism  and let ${\mathcal M}$ be a compact set of the
plane which is homeomorphic to the unit square. We select two
disjoint compact arcs on the boundary of ${\mathcal M}$ that we
conventionally denote ${\mathcal M}^-_{\rm left}$ and ${\mathcal M}^-_{\rm right}$
and call the left and right sides of ${\mathcal M}$.
Then, setting ${\mathcal M}^-:= {\mathcal M}^-_{\rm left}\cup
{\mathcal M}^-_{\rm right}$, the pair $\widehat{\mathcal M}:=({\mathcal M},{\mathcal M}^-)$
is called an oriented
rectangle. Given two oriented rectangles $\widehat{\mathcal M}$
and $\widehat{\mathcal N}$ and a compact set ${\mathcal H}\subset
{\mathcal M},$ we write
$$({\mathcal H},\Phi):
\widehat{\mathcal M}
\sap
\widehat{\mathcal N}$$
if the following property holds:
\begin{itemize}
\item[] for every path $\gamma_0: [0,1]\to {\mathcal M}$ with
$\gamma(0)$ and $\gamma(1)$ belonging to different components of ${\mathcal M}^-$,
there exists a sub-path $\gamma_1:=\gamma_0|_{[s_0,s_1]}$ such that
$\gamma_1(t)\in {\mathcal H}$ for all $t\in [s_0,s_1]$ and $\Phi(\gamma_1(t))\in {\mathcal N}$
with $\gamma(s_0)$ and $\gamma(s_1)$ belonging to different components of ${\mathcal N}^-$.
\end{itemize}
If ${\mathcal H}={\mathcal M}$, we just write $\Phi:\widehat{\mathcal M}
\sap \widehat{\mathcal N}$. Moreover, for any integer $\ell\geq 2$, we will use the notation
$$
\Phi:\widehat{\mathcal M} \sap^{\ell} \widehat{\mathcal N},
$$
if there are $\ell$  pairwise disjoint (nonempty) compact sets
${\mathcal H}_0,\dots, {\mathcal H}_{\ell-1}$ in ${\mathcal M}$, such that
$$
({\mathcal H_i},\Phi): \widehat{\mathcal M} \sap \widehat{\mathcal N} \quad\hbox{for all}\;\;
i=0,\dots, \ell-1.
$$

In the proof of Theorem \ref{th3.1} the following result,
adapted from Pascoletti, Pireddu and Zanolin \cite{PPZ-2008}) will be used.

\begin{lemma}\label{le3.1}
Assume that $\Phi=\Phi_2\circ\Phi_1$ and let $\widehat{Q}$ and $\widehat{R}$ be two oriented rectangles
such that
\begin{itemize}
\item[{\rm i)}]
$\Phi_1:\widehat{\mathcal Q} \sap^{\ell} \widehat{\mathcal R}$ for some $\ell\geq 2$,
\item[{\rm ii)}]
$\Phi_2:\widehat{\mathcal R} \sap \widehat{\mathcal Q}$.
\end{itemize}
Then, $\Phi$ induces chaotic dynamics on $\ell$ symbols in the set ${\mathcal Q}$.
\end{lemma}

We are  in position now to prove Theorem \ref{th3.1}, with
$\Phi_1$ and $\Phi_2$ the Poincar\'{e} maps associated to the system \eqref{iii.3}
in the intervals $[0,T_0]$ and $[T_0,T]$, respectively. Clearly,
$\Phi=\Phi_2\circ\Phi_1$ is the Poincar\'{e} map on the interval $[0,T]$
and its $n$-periodic points corresponds to the $nT$-periodic solutions to
the differential system.
\par
For  the sake of simplicity in the exposition, we restrict ourselves to the case of $\ell=2$
symbols, conventionally $\{0,1\}.$ The case of an arbitrary $\ell\geq 2$ can be easily
proved via a simple modification of our argument.
\par
As a first step, we focus our attention in the time-interval $[0,T_0]$
where the supports of $\alpha$ and $\beta$ intersect nontrivially on a set containing
a non-degenerate interval, $J$, where, without loss of generality, we can suppose that
$$\min_{t\in J}\{\a(t),\b(t)\}\geq \varsigma_0 >0.$$
In $[0,T_0]$ we are precisely in the same situation as the authors in
\cite[\S 2]{LMZ-2020}. Accordingly, we just recall some main facts from \cite{LMZ-2020} which are
needed for our proof and send the reader for the technical details to
the original article, if necessary.
For a fixed $\eta >0$ with $\min\{f'(0),g'(0)\} > \eta,$ we can find a (small) radius
$r_0>0$ such that
$$\theta(T_0,z_0) - \theta(0,z_0)\geq \lambda \eta \varsigma_0|J|
\quad \hbox{for all}\;\; z_0\;\; \hbox{with}\;\; \|z_0\|=r_0$$
(see \cite[Formula (13)]{LMZ-2020}). Here, as in  Section \ref{sec3.1}, we are denoting by  $\theta(t,z_0)$ the
angular coordinate associated with the solution of \eqref{iii.3} with $z_0\not=0$ as initial point.
Then, for
\begin{equation}\label{iii.16}
\lambda > \lambda^*:= \frac{7\pi}{2\eta \varsigma_0|J|},
\end{equation}
we have that
\begin{equation}\label{iii.17}
\theta(T_0,z_0)>4\pi \quad \text{if } \;\; 0\leq \theta(0,z_0)\leq \pi/2\;\; \text{with}\;\; \|z_0\|=r_0.
\end{equation}
On the other hand, following the same argument as in the previous steps \eqref{iii.9}-\eqref{iii.10}
(see also \cite[p. 2401]{LMZ-2020}), we can find a (large) radius $R_0>r_0$ such that the solutions
departing from the first quadrant outside the disc of radius $R_0$ cannot cross the third quadrant,  that
is
\begin{equation}\label{iii.18}
\theta(T_0,z_0)<\frac{3\pi}{2} \quad \text{if } \;\; 0\leq \theta(0,z_0)\leq \pi/2\;\; \text{with}\;\;  \|z_0\|=R_0.
\end{equation}
Suppose now that $\lambda > \lambda^*$ is fixed.
By the continuous dependence of the solutions from initial data (see also \cite[Pr. 1]{LMZ-2020})
there are two radii $r_{\lambda}$ and $R_{\lambda}$ with
$$0 < r_{\lambda} \leq r_{0} < R_{0} \leq R_{\lambda}$$
such that any solution $\zeta(t;z_0)$ of \eqref{iii.3}  with $z_0$ in the
first quadrant and $r_0\leq \|z_0\|\leq R_0$ satisfies $r_{\lambda}\leq \|\zeta(t;z_0)\|\leq R_{\lambda}$ for all $t\in[0,T_0]$.
\par
Subsequently, we introduce  the  sets
$${\mathcal Q}:= \{z=(x,y)\in{\mathbb R}^2: 0\leq x\leq r_0, \, y\geq 0, \, r_0\leq \|z\|\leq R_0\},$$
and
$${\mathcal Q}^-_{\rm left}:= {\mathcal Q}\cap C_{r_0}, \qquad
{\mathcal Q}^-_{\rm right}:= {\mathcal Q}\cap C_{R_0},$$
where $C_{\rho}$ denotes the circumference of center at the origin and radius $\rho>0$, and
$${\mathcal R}:= \{z=(x,y)\in{\mathbb R}^2: 0\leq x\leq r_{\lambda}, \, y\leq 0, \,
r_{\lambda}\leq \|z\|\leq R_{\lambda}\},$$
$${\mathcal R}^-_{\rm left}:= {\mathcal R}\cap \{(x,y):x=0\}, \qquad
{\mathcal R}^-_{\rm right}:= {\mathcal R}\cap \{(x,y):x=r_{\lambda}\}.$$
So that the oriented rectangles $\widehat{Q}$ and $\widehat{R}$ are defined, too.
It is obvious that both ${\mathcal Q}$ and ${\mathcal R}$ are homeomorphic to the unit square.
For instance, the map
$$(u,v)\mapsto \bigl(u r_0, (\varrho(v)^2 - u^2 r_0^2)^{1/2} \bigr),$$
with $\varrho(v)=r_0+v(R_0-r_0),$ provides a homeomorphism from the unit square $[0,1]^2$
onto ${\mathcal Q},$
mapping $v=0$ to ${\mathcal Q}^-_{\rm left}$ and $v=1$ to ${\mathcal Q}^-_{\rm right}$ and,
similarly,
the map $$(u,v)\mapsto \bigl(u r_{\lambda}, -(\varrho_{\lambda}(v)^2 - u^2 r_{\lambda}^2)^{1/2} \bigr),$$
with $\varrho_{\lambda}(v)=r_{\lambda}+v(R_{\lambda}-r_{\lambda}),$
provides a homeomorphism from the unit square onto ${\mathcal R},$
mapping $u=0$ to ${\mathcal R}^-_{\rm left}$ and $u=1$ to ${\mathcal R}^-_{\rm right}$.
Thus, the definition of oriented rectangles is well posed.
\par
We also introduce the pairwise disjoint compact (nonempty) subsets of ${\mathcal Q}$
$${\mathcal H}_{1-j}:=\{z\in {\mathcal Q}: \Phi_1(z)\in {\mathcal R},
\, 3\pi/2 + 2j\pi\leq \theta(T_0,z)\leq 2\pi+  2j\pi\} \quad \text{for }\, j=0,1.$$
By definition,
$${\mathcal H}_{0}\sqcup {\mathcal H}_{1}\subset {\mathcal Q}\cap \Phi_1^{-1}({\mathcal R}).$$
Figure \ref{fig-xii} shows a possible hierarchy of the sets $\mc{H}_0$ and $\mc{H}_1$ within
$\mc{Q}.$

\begin{figure}[h!]
\centering
\includegraphics[scale=0.25]{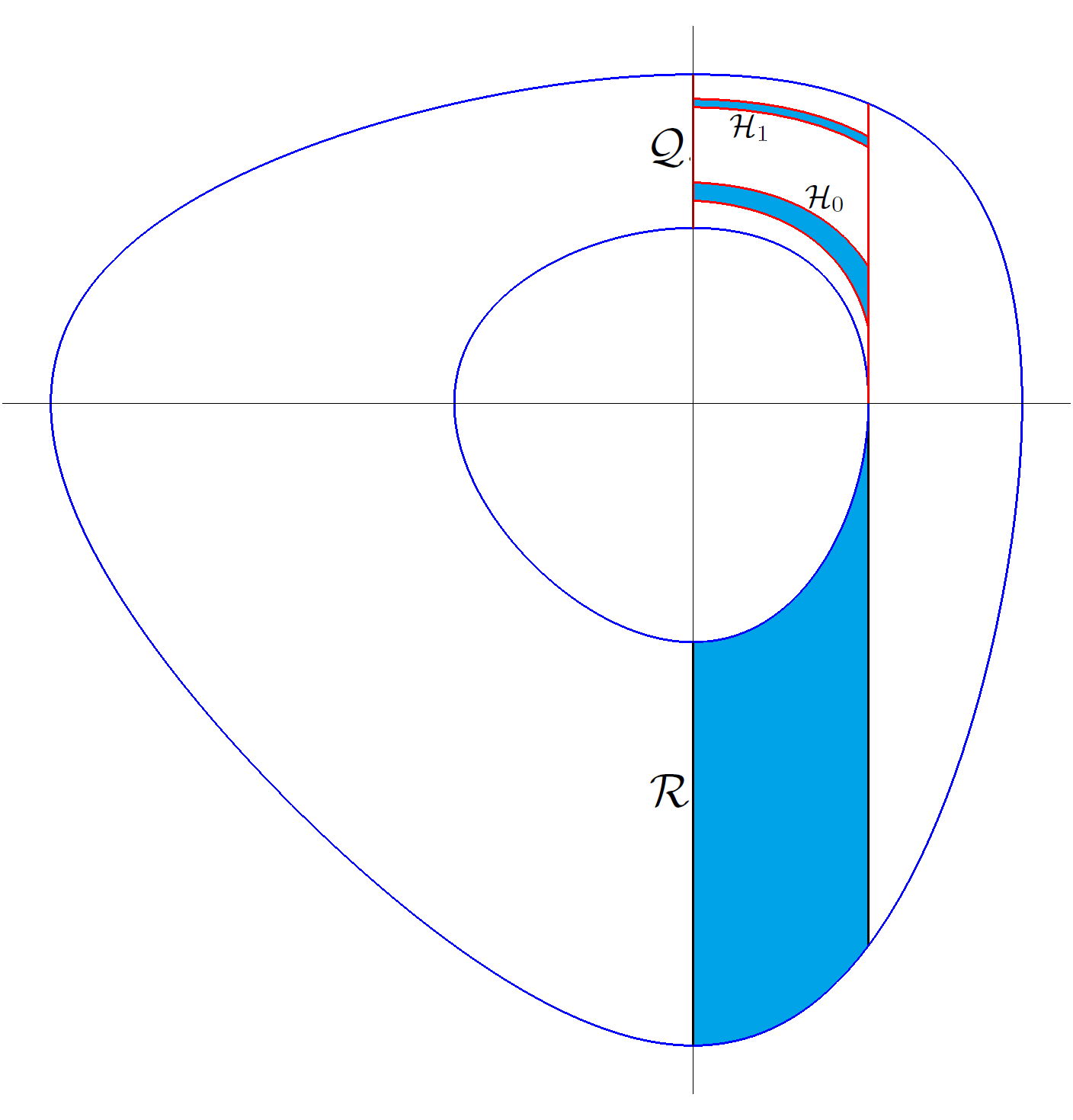}
\caption{The  sets $\mc{H}_0$ and $\mc{H}_1$, as well as
$\mc{Q}$ and $\mc{R}$, for the stepwise constant coefficients
studied in Section 3.1. Although in this special situation
${\mathcal H}_{0}\sqcup {\mathcal H}_{1}={\mathcal Q}\cap \Phi_1^{-1}({\mathcal R})$,
by the choice of $T_0$,
in general, ${\mathcal Q}\cap \Phi_1^{-1}({\mathcal R})$ might contain more components.}
\label{fig-xii}
\end{figure}

Let $\gamma:[0,1]\to {\mathcal Q}$ be a continuous map
such that $\gamma(0)\in {\mathcal Q}^-_{\rm left}$
and $\gamma(1)\in {\mathcal Q}^-_{\rm right}$ and let us consider the evolution of
$\gamma(s)$ through the first Poincar\'{e} map $\Phi_1.$
The angular coordinate
$\theta(T_0,\gamma(s))$ as a function of the parameter $s\in [0,1]$ is a continuous map
which, according to \eqref{iii.17} and \eqref{iii.18} satisfies
$$\theta(T_0,\gamma(0))> 4\pi, \quad \theta(T_0,\gamma(1))\leq 3\pi/2.$$
Hence, the path $s\mapsto \theta(T_0,\gamma(s))$ crosses at least twice the
portion of the fourth quadrant between $C_{r_{\lambda}}$ and $C_{R_{\lambda}}$ and
therefore it also crosses (at least twice) the region ${\mathcal R}$ from  $x=0$
to $x=r{\lambda}.$
\par
By an elementary continuity argument, there are two subintervals, $[s_0,s_1]$ and
$[s_3,s_4]$, with
$0 < s_0 <s_1 < s_3 < s_4 <1$,  such that
$$\theta(T_0,\gamma(s_0))=4\pi, \quad \theta(T_0,\gamma(s_1)) = \frac{3\pi}{2} + 2\pi,
\quad \frac{3\pi}{2} + 2\pi\leq \theta(T_0,\gamma(s))\leq 4\pi\; \forall\, s\in[s_0,s_1],$$
and
$$\theta(T_0,\gamma(s_3))=2\pi, \quad \theta(T_0,\gamma(s_4)) = \frac{3\pi}{2},
\quad \frac{3\pi}{2}\leq \theta(T_0,\gamma(s))\leq 2\pi\; \forall\, s\in[s_3,s_4].$$
Thus, the path $[s_0,s_1]\ni s\mapsto \theta(T_0,\gamma(s))$ crosses the fourth quadrant and
hence  there is a subinterval $[s'_0,s'_1]\subset [s_0,s_1]$ such that
$\Phi_1(\gamma(s))\in {\mathcal R}$ for all $s\in [s'_0,s'_1]$ with
$\Phi_1(\gamma(s'_0))\in {\mathcal R}^-_{\rm right}$ and
$\Phi_1(\gamma(s'_1))\in {\mathcal R}^-_{\rm left}.$ By the definition of ${\mathcal H}_0$,
we have that $\gamma(s)\in {\mathcal H}_0$ for all $s\in [s'_0,s'_1]$.
Therefore, we have proved that
$({\mathcal H}_0,\Phi_1): \widehat{Q} \sap \widehat{R}.$
In the same manner, we can find a subinterval $[s'_3,s'_4]\subset [s_3,s_4]$ such that
$\Phi_1(\gamma(s))\in {\mathcal R}$ for all $s\in [s'_3,s'_4],$ with
$\Phi_1(\gamma(s'_3))\in {\mathcal R}^-_{\rm right}$ and
$\Phi_1(\gamma(s'_4))\in {\mathcal R}^-_{\rm left}.$ Therefore, by the
definition of ${\mathcal H}_1$
we have that $\gamma(s)\in {\mathcal H}_1$ for all $s\in [s'_3,s'_4],$ thus proving that
$({\mathcal H}_1,\Phi_1): \widehat{Q} \sap \widehat{R}.$
This ends the proof of
Lemma \ref{le3.1}(i) for $\ell=2.$
\par
To have the result for an arbitrary $\ell\geq 2,$ we have just to modify the choice of $\lambda^*$ in
\eqref{iii.16} to $\lambda > \lambda^*:= {(4\ell -1)\pi}/{(2\eta \varsigma_0|J|)}$
and introduce corresponding subsets $\mc{H}_0,\dots, \mc{H}_{\ell-1}$ of
${\mathcal Q}\cap \Phi_1^{-1}({\mathcal R}).$
\par
Now, we consider the map $\Phi_2$ by studying the equation \eqref{iii.3}
in the interval $[T_0,T]$, where  $\alpha\equiv 0$ and hence
$$\Phi_2(x_0,y_0)= \bigl(x_0,y_0+ \lambda g(x_0)\int_{T_0}^{T}\beta(t)\,dt\bigr).$$
Thus, the dynamics is the same as that of \eqref{iii.2} in the interval $[T_0,T]$
considered in Section 3.1, modulo a minor change in the parameters involved.
It is clear that the points on ${\mathcal R}^-_{\rm left}$ remain stationary, while those
of ${\mathcal R}^-_{\rm right}$ move upward at the new position
$$y_0 + \lambda g(x_{\lambda})\int_{T_0}^{T}\beta(t)\,dt\geq -R_{\lambda}
+ \lambda g(x_{\lambda})\int_{T_0}^{T}\beta(t)\,dt.$$
Therefore, if
\begin{equation}\label{iii.19}
\int_{T_0}^{T}\beta(t)\,dt > K_{\lambda}:= \frac{2 R_{\lambda}}{\lambda g(x_{\lambda})},
\end{equation}
then Lemma \ref{le3.1}(ii) holds. Indeed, any path in ${\mathcal R}$
linking the two sides of ${\mathcal R}^-$ is stretched to a path crossing entirely the
set ${\mathcal Q}$ (from ${\mathcal Q}^-_{\rm left}$ to ${\mathcal Q}^-_{\rm right}$)
and remaining inside the strip $[0,r_{\lambda}]\times {\mathbb R}.$  This concludes the proof of Theorem \ref{th3.1}.
\qed

\begin{remark}\label{re3.2}{\em
The chaotic dynamics associated with the Poincar\'{e} map $\Phi=\Phi_2\circ\Phi_1$
comes from the composition of a
twist rotation (due to $\Phi_1$) with a shearing parallel to the $y$-axis (due to $\Phi_2$).
Clearly, if we consider the Poincar\'{e} map of initial time $T_0$, we will obtain
$\Phi_1\circ\Phi_2$ and, by a similar argument, through a symmetric counterpart
of Lemma \ref{le3.1} where the order of the two maps is commuted, we may prove the existence of a
horseshoe type structure inside the set $\mc{R}.$
Moreover, using the fact that the shear map moves downward the points with $x<0$
(in fact, $g(x)<0$ for $x<0$),  we can also start from a set
$${\mathcal Q}:= \{z=(x,y)\in{\mathbb R}^2: 0\leq -r_0\leq x\leq 0, \, y\leq 0, \, r_0\leq \|z\|\leq R_0\}$$
and a target set
$${\mathcal R}:= \{z=(x,y)\in{\mathbb R}^2:  r_{\lambda}\leq x \leq 0, \, y\geq 0, \,
r_{\lambda}\leq \|z\|\leq R_{\lambda}\},$$
again reversing  the roles of $\mc{Q}$ and $\mc{R}$ by commuting $\Phi_1$ with $\Phi_2.$
\par
Furthermore, we can obtain a variant of Theorem \ref{th3.1} by keeping condition $(c_1)$
and modifying condition $(c_2)$ to
\begin{itemize}
\item[$(c'_2)\;$]
$\alpha \gneq 0$ and $\beta \equiv 0$ on $[T_0,T].$
\end{itemize}
In this situation, the assumption \eqref{iii.4} should be replaced by
a similar hypothesis involving $\int_{T_0}^T \alpha.$  Under conditions $(c_1)-(c'_2)$,
the  Poincar\'{e} maps produce a dynamics where a twist rotation is composed
with a shearing parallel to the $x$-axis.}
\end{remark}

\begin{remark}\label{re3.3}{\em
H\'{e}non proposed in \cite{He-1969}, as a model problem, the mapping given by the quadratic equations
$$
  x_1= x\cos \alpha - (y-x^2)\sin \alpha,
    \qquad y_1= x\sin \alpha + (y-x^2)\cos \alpha,
$$
as a simple example of an area-preserving mapping which exhibits chaotic dynamics.
The mapping in H\'{e}non's model splits into a
product of a shearing parallel to the $y$-axis and a rotation.
It is interesting that the typical numerical features observed in the experiments
in \cite{He-1969} appear also in Volterra's equations for the setting of Theorem \ref{th3.1},
as shown in Figure \ref{fig-xiii}.

\begin{figure}[h!]
\centering
\includegraphics[scale=0.4]{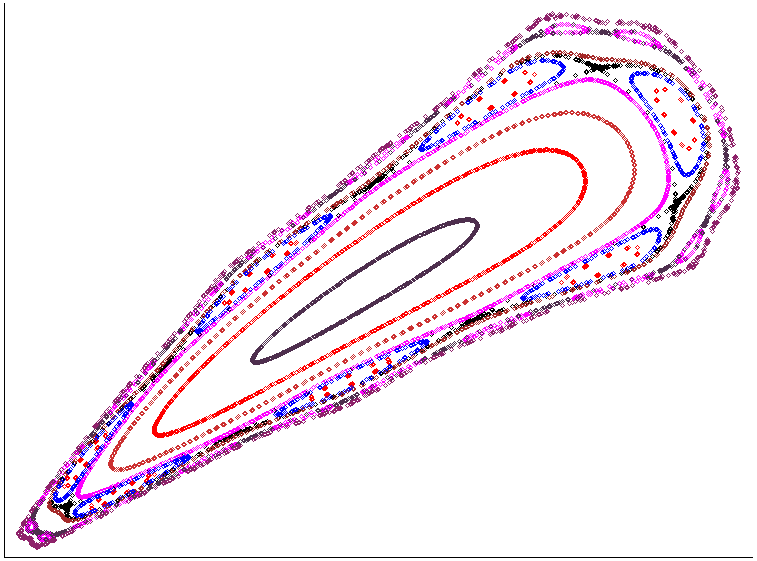}\qquad
\includegraphics[scale=0.4]{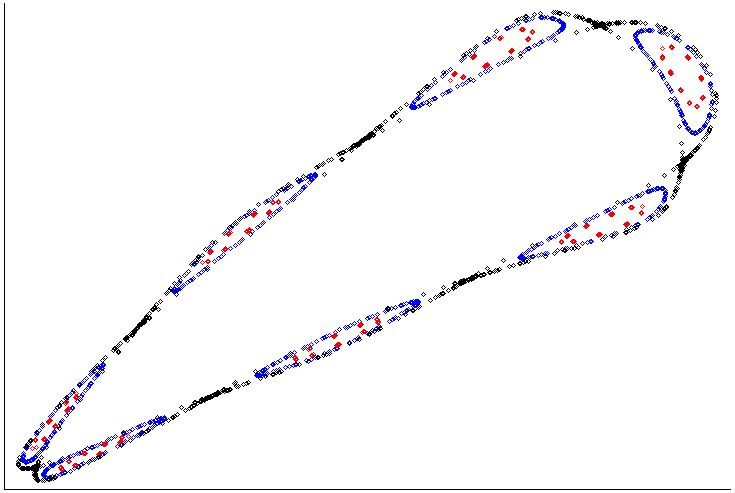}
\caption{Some numerical experiments for a periodic Volterra system.}
\label{fig-xiii}
\end{figure}

The left picture of Figure \ref{fig-xiii} shows the first 800 iterations of
the Poincar\'{e} map for a Volterra system
$$
    \left\{ \begin{array}{l} u' = \lambda\alpha(t)u(1-v),\\ v' = -\lambda\beta(t)v(1-u),\end{array}
    \right.
$$
under the general assumptions of Theorem \ref{th3.1},
starting from different initial points. The numerical experiments reveal
the presence of stability regions (invariant curves
around the constant coexistence state $(1,1)$, according to Liu \cite{Liu-1995}),
as well as some more complicated discrete orbits of ``chaotic type''.
The right figure highlights some  special orbits,
in particular, seven-fold island chains (according to  Arrowsmith and Place \cite[p.262]{AP-1992})
separated away by heteroclinic connections
and including higher order subharmonics. See also H\'enon \cite[Figs. 4, 5]{He-1969}. }
\end{remark}

\bigskip

\begin{center}
\textbf{Chaotic dynamics in the degenerate case for system \eqref{iii.3}}
\end{center}

\noindent To conclude our analysis, we show  now how to adapt the
proof of Theorem \ref{th3.1} to deal with degenerate weights in
system \eqref{iii.3}. This will be achieved by applying the
estimates previously obtained by the authors in \cite{LMZ-2021},
where some results about the existence of periodic solutions for
\eqref{iii.3} were found when the functions $\a$ and $\b$ have a
common support of zero measure. This case provides also a
connection with the theorems in Section \ref{sec2}.
\par
To fix ideas, we suppose that in the interval $[0,T_0]$ there are
$\ell$ positive humps of  $\a$ separated away by $k=\ell$ positive
humps of $\b$, as in \eqref{ii.23}, with the corresponding support
intervals intersecting on sets of zero measure. The symmetric case
when, instead of \eqref{ii.23}, the condition \eqref{ii.34} holds,
can be treated similarly by interchanging the roles of $\alpha$
and $\beta$ and modifying the choice of the initial set $\mc{Q}$
and the target set $\mc{R}$, as it will explained in Remark
\ref{re3.4} below.
\par
Then the following result holds, where, as before, we suppose that
$f,g:{\mathbb R}\to {\mathbb R}$ are $C^1$-functions such that
$f(0)=g(0)=0$, $f'(0)>0$, $g'(0) >0$, and $f(s)s>0$, $g(s)s>0$ for
$s\not=0$. Moreover, at least one of the two functions, e.g., $f$,
is bounded on $(-\infty,0].$

\begin{theorem}\label{th3.2}
Assume that there exists $T_0\in (0,T)$ such that:
\begin{itemize}
\item[$(c'_1)\;$] $\ell=k \geq 4$, with $\ell,k$ even integers;
\item[$(c_2)\;$]
$\alpha \equiv 0$ and $\beta \gneq 0$ on $[T_0,T].$
\end{itemize}
Then, there exists $\lambda^*=\lambda^*_{\ell}$ such that, for
every $\lambda >\lambda^*,$ there exists a constant $K_{\lambda}$
for which,  whenever
\begin{equation}\label{iii.20}
\int_{T_0}^{T} \beta(t)\,dt > K_{\lambda},
\end{equation}
the Poincar\'{e} map associated with \eqref{iii.3} induces chaotic
dynamics on $\ell/2$ symbols on some compact set ${\mathcal Q}$
contained in the first quadrant.
\end{theorem}
\begin{proof}
For simplicity in the exposition in the proof, we will focus
attention in the case $\ell=k=4$, where, starting on the region
$\mc{Q}$ we obtain two crossings of the region $\mc{R}$ and, as a
consequence, a complex dynamics on two symbols.
\par
We will follow the same argument as in the  proof of Theorem
\ref{th3.1}, just  emphasizing the necessary modifications. As a
first step, we will focus our attention in the time-interval
$[0,T_0]$.

As the supports of $\alpha$ and $\beta$ do not overlap, the
associated dynamics is a composition of shear maps of the
following form. If $\a\gneq 0$, $\b\equiv 0$ and $y>0$ (resp.
$y<0$), the points are moved parallel to the $x$-axis from right
to left (resp. from left to right). Similarly, when $\a\equiv 0$,
$\b\gneq 0$ and $x<0$ (resp. $x>0$), then the points are moved
parallel to the $y$-axis in a decreasing (resp. increasing) sense.
Once passed two positive humps  of $\a$ and an intermediate
positive hump of $\b$, the points in the first quadrant with $y_0
\geq \delta_0$ end in the fourth quadrant for sufficiently large
$\lambda >0$. Thus, after another interval where $\b\gneq 0,$ we
come back to the first quadrant and can repeat the process, as
described in detail by the authors in \cite{LMZ-2021}. Since the
points on the $x$-axis (resp. the $y$-axis) do not move when
$\b\equiv 0$  (resp. when $\a\equiv 0$), in the proof of this
theorem it is convenient to slightly modify the choice of $\mc{Q}$
and $\mc{R}$ as follows
$${\mathcal Q}:= \{z=(x,y)\in{\mathbb R}^2: 0\leq x\leq r_1, \, y\geq 0, \, r_0\leq \|z\|\leq R_0\},
$$
for $0 <r_1<r_0,$ and
$${\mathcal Q}^-_{\rm left}:= {\mathcal Q}\cap C_{r_0}, \qquad
{\mathcal Q}^-_{\rm right}:= {\mathcal Q}\cap C_{R_0}.$$ In this
manner, there is $\hat{y}_0:= (r_0^2-r_1^2)^{1/2}>0$ such that
$y\geq \hat{y}_0$ for all $(x,y)\in \mc{Q}.$ According to Lemmas
1, 2 and Theorem 3 of \cite{LMZ-2021}, for that choice, there
exists $\lambda^*$ such that the condition \eqref{iii.17} is
reestablished for every (fixed) $\lambda > \lambda^*$, i.e.,
\begin{equation}\label{iii.21}
\theta(T_0,z_0)>4\pi \quad \text{if }\;\;   z_0 \in \mc{Q}\cup C_{r_0}.
\end{equation}
On the other hand, exactly as explained above, for sufficiently
large $R_0>r_0$, the solutions departing from the first quadrant
outside the disc of radius $R_0$ cannot cross the third quadrant,
that is
\begin{equation}\label{iii.22}
\theta(T_0,z_0)<\frac{3\pi}{2} \quad \text{if } \;
\;  z_0 \in \mc{Q}\cup C_{R_0}.
\end{equation}
By continuous dependence, once fixed a $\lambda > \lambda^*$, one
can find two radii $r_{\lambda}$ and $R_{\lambda}$, with
$$0 < r_{\lambda} \leq r_{1} < R_{0} \leq R_{\lambda},$$
such that any solution $\zeta(t;z_0)$ of \eqref{iii.3} with
$z_0\in \mc{Q}$ lies in the set
$${\mathcal R}:= \{z=(x,y)\in{\mathbb R}^2: 0\leq x\leq r_{\lambda}, \, y\leq 0, \,
r_{\lambda}\leq \|z\|\leq R_{\lambda}\},$$ for all $t\in [0,T_0].$
\\
Thus, much like in the proof of Theorem \ref{th3.1}, we can also
define
$${\mathcal R}^-_{\rm left}:= {\mathcal R}\cap \{(x,y):x=0\}, \qquad
{\mathcal R}^-_{\rm right}:= {\mathcal R}\cap
\{(x,y):x=r_{\lambda}\}.$$

From now on, we have just to repeat the proof of Theorem
\ref{th3.1} without any significant change in order to show that
$\Phi_1: \widehat{\mc{Q}}\sap^2 \widehat{\mc{R}}$. The
verification that $\Phi_2: \widehat{\mc{R}}\sap \widehat{\mc{Q}}$
proceeds exactly as before. Therefore, according to Lemma
\ref{le3.1}, we get the chaotic dynamics for
$\Phi=\Phi_2\circ\Phi_1$ on two symbols.
\par
Note that, in order to produce a semi-conjugation on $m$-symbols,
we need to make at least $m$-turns around origin, starting at
$\mc{Q}$, in the time-interval $[0,T_0].$ This can be achieved,
for  sufficiently large $\lambda$, if both $\a$ and $\b$ are
assumed to have, at least, $2m$ positive humps. The proof is
complete.
\end{proof}

\begin{remark}\label{re3.4}{\em
If, instead of \eqref{ii.23}, the condition \eqref{ii.24} holds,
then we can assume $(c'_2)$, instead of $(c_2)$, and take $\mc{Q}$
and $\mc{R}$ to be} adjacent to the $x$-axis and opposite with
respect to the $y$-axis.
\end{remark}

\begin{center}
\textbf{Chaotic dynamics  when $\a(t)\b(t)>0$ for all $t\in
[0,T]$.}
\end{center}

\noindent So far, in this section we have studied the system
\eqref{iii.3}  by assuming that either $\a\equiv 0$ and $\b\gneq
0$,  or $\b\equiv 0$ and $\a\gneq 0$,  on some time-interval.  In
both these cases, the dynamics is spanned by the superposition of
a twist rotation with a shear map. In this section, we would like
to stress the fact that a rich dynamics can be also produced,
through a different mechanism, when $\a$ and $\b$ are throughout
positive and appropriately separated away from each other, in a
sense to be specified below. To analyze the simplest geometry, we
restrict ourselves to consider the system \eqref{iii.2} with
stepwise constant function coefficients, $\alpha(t)$ and
$\beta(t)$, as in Section \ref{sec3.1}. More precisely, we assume
that $T=T_0+T_1$ and
\begin{equation}\label{iii.23}
\b(t):= \left\{
\begin{array}{lll}
\b_0>0&\;\;\hbox{if}\; t\in[0,T_0),\\[1ex]
\b_1>0&\;\;\hbox{if}\; t\in [T_0,T),
\end{array}
\right. \qquad
\a(t):= \left\{
\begin{array}{lll}
\a_0>0&\;\;\hbox{if}\; t\in[0,T_0),\\[1ex]
\a_1>0&\;\;\hbox{if}\; t\in [T_0,T),
\end{array}
\right.
\end{equation}
where the positive constants $\a_0$, $\a_1$, $\b_0$, $\b_1$ and
the exact values of $T_0$ and $T_1$ will be made precise later. In
this case, the dynamical behaviors of \eqref{iii.2} on each of the
intervals $[0,T_0]$ and $[T_0,T]\equiv [0,T_1]$ are those of the
associated autonomous Volterra-type systems
\begin{equation}
\label{iii.24}
\hbox{(I)}\quad \left \{
\begin{array}{ll}
x'=\a_0(1-e^y)\\
y'=-\b_0(1-e^x)
\end{array}
\right.\hspace{3cm}
\hbox{(II)} \quad \left \{
\begin{array}{ll}
x'=\a_1(1-e^y)\\
y'=-\b_1(1-e^x)
\end{array}
\right.
\end{equation}
respectively. Note that the systems \eqref{iii.24}-(I) and
\eqref{iii.24}-(II) have the same equilibrium point (the origin),
and both describe a global center. However, for different choices
of the pairs $(\a_0,\b_0)$ and $(\a_1,\b_1)$ the shape of the
level lines of the Hamiltonian may change. Geometrically, this
situation is reminiscent to that already studied by Takeuchi et
al. in  \cite{TDHS-2006}, where a predator-prey model with
randomly varying coefficients was considered.
\par
Now, making a choice so that
\begin{equation}\label{iii.25}
\frac{\b_1}{\a_1} > \frac{\b_0}{\a_0},
\end{equation}
or the  converse inequality, it is possible to find level lines of
the two systems crossing to each other.
\par
From this, we can construct two annular domains ${\mathcal A}_0$
and ${\mathcal A}_1$, filled by periodic orbits of
\eqref{iii.24}-(I) and \eqref{iii.24}-(II), respectively, in such
a manner that the annuli intersect into four rectangular regions,
as illustrated in Figure \ref{fig-xiv},  where the annulus
${\mathcal A}_0$ is obtained for $\a_0=\b_0=1$ and the level lines
passing through the initial points $(0,0.8)$ and $(0,1.5)$,
whereas ${\mathcal A}_2$ is obtained for $\a_1=0.1, \b_1=5$ and
the level lines passing through the  initial points $(0,1.7)$ and
$(0,2.5)$.  From a numerical point of view, the larger is the gap
in condition \eqref{iii.25}, the wider are the linked annuli which
can be constructed.

\begin{figure}[h!]
\centering
\includegraphics[scale=0.5]{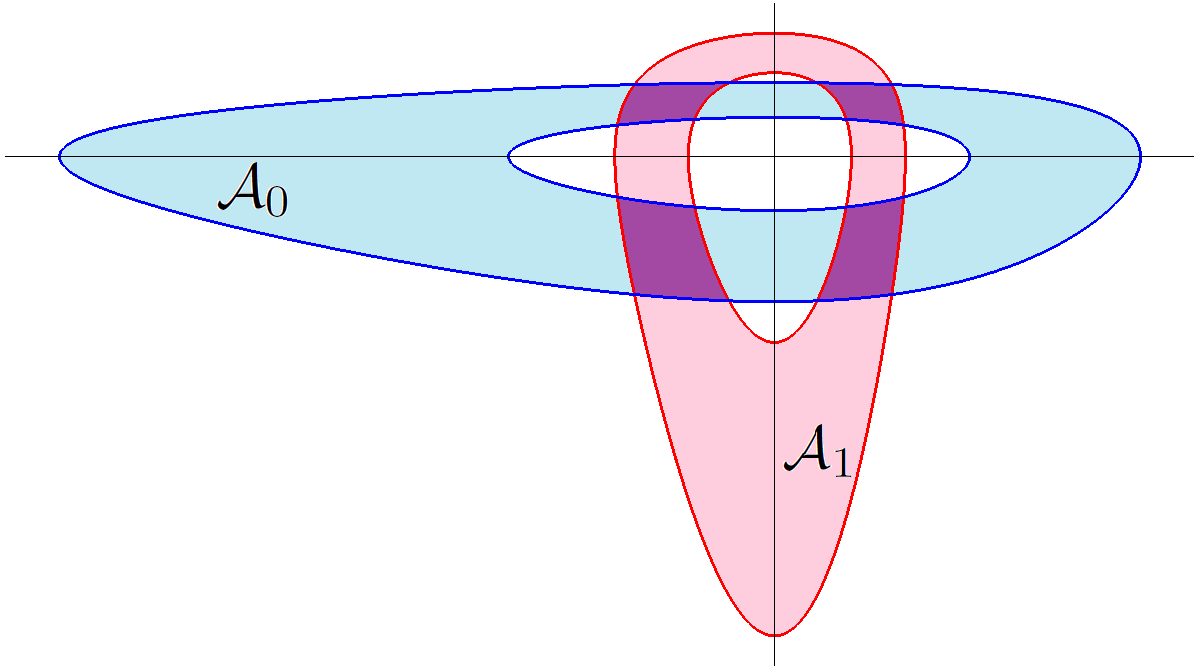}
\caption{ Two linked annuli ${\mathcal A}_0$ and ${\mathcal A}_1$
filled in by the periodic orbits of \eqref{iii.24}-(I) and
\eqref{iii.24}-(II), respectively, and crossing to each other into
four rectangular regions. } \label{fig-xiv}
\end{figure}

In general, to describe in a precise manner the construction of
the  two linked annuli, we denote by $\mc{E}_i$ the energy
functions for the pair $(\a_i,\b_i)$ ($i=0,1$), so that
$$\mc{A}_i:= \{(x,y)\in {\mathbb R}^2: c_i \leq \mc{E}_i(x,y) \leq d_i\},$$
for
$$
   d_i> c_i > \min \mc{E}_i = \mc{E}_i(0,0) =\a_i+\b_i.
$$
We also denote by $x^-_i(\ell) < 0 < x^+_i(\ell)$ the abscissas of the intersection points
of the level line $\mc{E}_i=\ell\in[c_i,d_i]$ with the $x$-axis and, symmetrically, by
$y^-_i(\ell) < 0 < y^+_i(\ell)$ the ordinates of the intersection points
of the level line $\mc{E}_i=\ell\in[c_i,d_i]$ with the $y$-axis. Then $\mc{A}_0$ and $\mc{A}_1$
are \textit{linked} provided that
$$x^-_0(c_0) < x^-_1(d_1); \quad   x^+_1(d_1)< x^+_0(c_0), \qquad
y^-_1(c_1) < y^-_0(d_0); \quad  y^+_0(d_0)< y^+_1(c_1).$$ This
definition corresponds to that considered by  Margheri, Rebelo and
Zanolin \cite{MRZ-2010} (see also Papini, Villari and Zanolin
\cite[Def. 3.2]{PVZ-2019}). The regions obtained as intersections
of the two annuli, are the four components of
$$\{(x,y)\in {\mathbb R}^2: c_0\leq \mc{E}_0(x,y) \leq d_0\, \land\, c_1\leq \mc{E}_1(x,y) \leq d_1\},$$
each one lying in a different quadrant. They are all homeomorphic to the unit square.
Now, if we denote by $\mc{Q}$ and $\mc{R}$
any pair chosen among these four intersections, we define the following orientations
$$
\mc{Q}^-:= \{(x,y)\in \mc{Q}: \mc{E}_0(x,y)=c_0\}\cup  \{(x,y)\in \mc{Q}: \mc{E}_0(x,y)=d_0\},
$$
$$\mc{R}^-:= \{(x,y)\in \mc{R}: \mc{E}_1(x,y)=c_1\}\cup  \{(x,y)\in \mc{R}: \mc{E}_1(x,y)=d_1\},$$
not being relevant the order in which  the names ``left'' and
``right'' in  these components of $[\cdot]^-$ are assigned.
Finally, if we denote by $\Phi_1$ and $\Phi_2$ the Poincar\'{e}
maps associated with \eqref{iii.24}-(I) and \eqref{iii.24}-(II),
respectively, by the monotonicity of the period map (already
exploited as, for instance, in Remark \ref{re3.1}), we can prove
that $\Phi_1(\mc{Q})$ crosses $\ell$-times $\mc{R}$ or, more
precisely, $\Phi_1:\widehat{\mc{Q}}\sap^{\ell} \widehat{\mc{R}},$
provided that $T_0$ is sufficiently large. Similarly, one  can
prove that $\Phi_2:\widehat{\mc{R}}\sap^m \widehat{\mc{Q}},$ for
sufficiently large  $T_1$. Thus, a chaotic dynamics on $\ell\times
m$ symbols is produced for the map $\Phi$ in the set $\mc{Q}.$
Equivalently, instead of taking sufficiently large $T_0$ and
$T_1$, one can put a parameter $\lambda$ in front of $\a$ and
$\b$, with $T_0$ and $T_1$ fixed, to obtain complex dynamics on
$\ell\times m$ symbols for all $\lambda > \lambda^*$, where
$\lambda^*$ depends on $(\ell,m)$, as well as on the several
coefficients of the equation.
\par
Here the geometrical configuration is the same as that considered
by Burra and Zanolin \cite{BuZa-2009}, and later generalized by
Margheri, Rebelo and Zanolin \cite{MRZ-2010} and Papini, Villari
and Zanolin \cite{PVZ-2019}.  The reader is sent to these
references for any further technical details.

\subsection{Ideal horseshoe dynamics for weights \eqref{iii.6}}\label{appendix}

\noindent In this section, we will perform a schematic
construction of the Smale's horseshoe inspired on the numerical
simulations of model \eqref{iii.2} with an stepwise configuration
of $\a$ and $\b$ as in \eqref{iii.6}.  These simulations show that
there is \emph{transversal intersection} between the sets
$\mc{H}_0, \mc{H}_1$ and $\Phi(\mc{H}_0), \Phi(\mc{H}_1)$ in the
sense that, for $i,j\in\{1,2\}$, the intersection
\[
\mc{H}_i\cap\Phi(\mc{H}_j)
\]
is a unique connected set. Assuming that this is the behavior for
all forward and backward iterates of the Poincar\'e map $\Phi$,
implies that $\Phi$ is not only semi-conjugated to the Bernouilli
shift, as proved in the previous subsections (also for the general
case of Theorem \ref{th3.1}), but also conjugated. Figures
\ref{fig-xix}-\ref{fig-xx}-\ref{fig-xxi} give evidence of these facts.
\par
We will assume that there are
two disjoint proper \emph{horizontal} topological squares,
$\mc{H}_0$, $\mc{H}_1\subsetneq\mc{Q}$, such that
\begin{equation}
\label{iii.26}
   \Phi(\mc{H}_i):=\mc{V}_i \quad \hbox{for}\;\;i\in\{0,1\}.
\end{equation}
By \emph{horizontal}, we mean that the lateral sides of $\mc{H}_0$
and $\mc{H}_1$ lie on the lateral sides of the square $\mc{Q}$, as
sketched in  Figure \ref{fig-xv}. In this figure, as in the
remaining figures of this section, we will represent any
topological square as an homeomorphic square quadrangle. In order
to have \eqref{iii.26}, it is assumed that
$$
\Phi_{T_0}(\mc{Q})\cap \mc{R}= \mc{R}_0\cup\mc{R}_1
$$
and, setting
$$
\mc{V}_i:= \Phi_{T_1}(\mc{R}_i)\cap \mc{Q},\qquad i\in\{0,1\},
$$
it follows that
$$
\Phi_{T_1}(\Phi_{T_0}(\mc{Q})\cap \mc{R}) \cap \mc{Q} =
\Phi_{T_1}(\mc{R}_0\cup\mc{R}_1)\cap \mc{Q} =\mc{V}_0\cup\mc{V}_1.
$$
\begin{figure}[h!]
\centering
\includegraphics[scale=0.57]{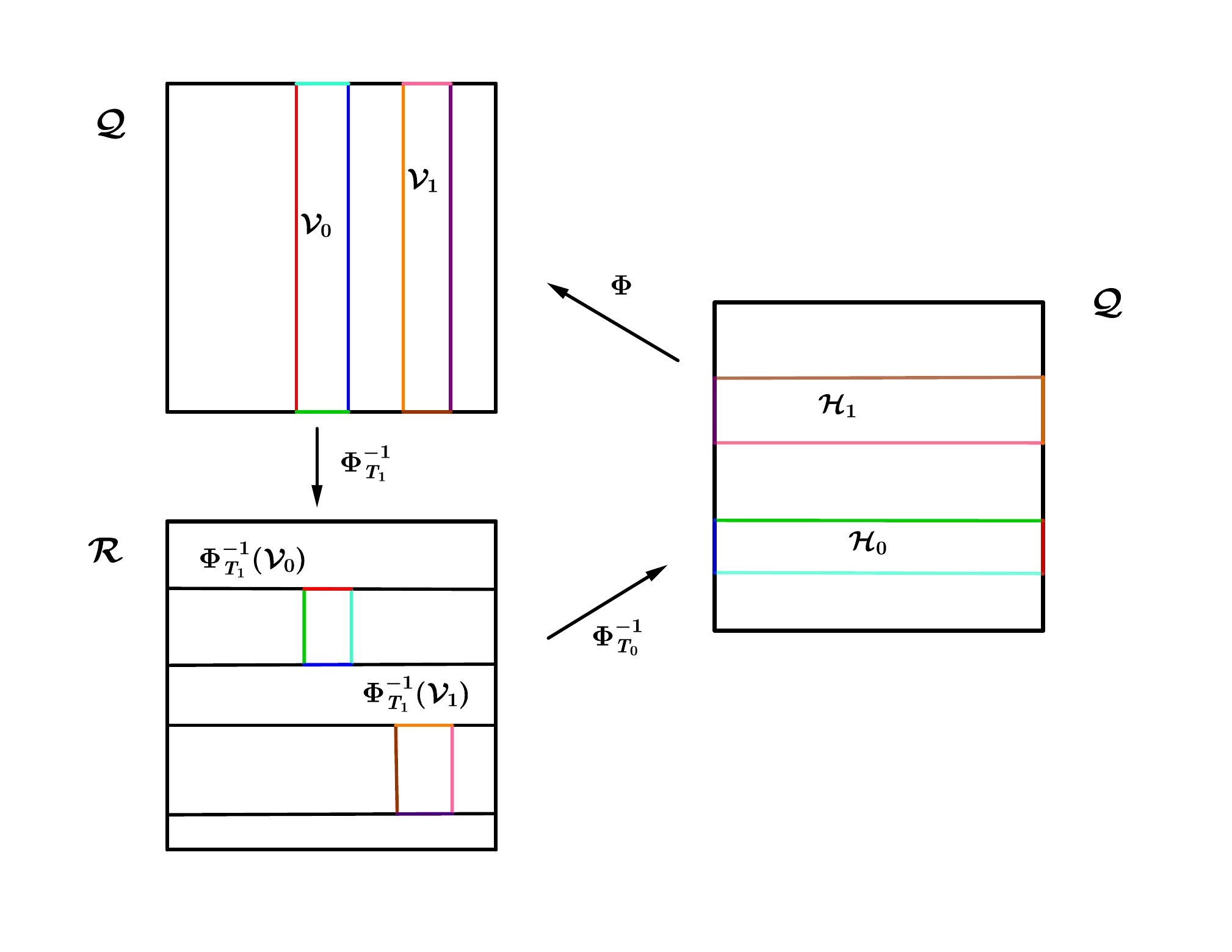}
\caption{$\Phi$ establishes an homeomorphism between $\mc{H}_i$ and $\mc{V}_i$, $i\in \{0,1\}$.}
\label{fig-xv}
\end{figure}
\par

By simply having a look at Figure \ref{fig-xv}, it is easily realized that, for every $i\in\{0,1\}$,
$\mc{S}_i:= \Phi_{T_1}^{-1}(\mc{V}_i)$ is a \emph{vertical} sub-square of $\mc{R}_i$; vertical in the sense that it is a portion of $\mc{R}_i$ linking the upper and lower sides of $\mc{R}_i$. It turns out that $\Phi_{T_0}$ sends $\mc{R}_i\setminus \mc{S}_i$ outside $\mc{Q}$. Since
$$
  \mc{H}_i:=\Phi_{T_0}^{-1}(\mc{S}_i)=\Phi_{T_0}^{-1}(\Phi_{T_1}^{-1}(\mc{V}_i))
  =\Phi^{-1}(\mc{V}_i),\qquad i\in \{0,1\},
$$
it is apparent that, for every $i\in \{0,1\}$,  $\Phi$ establishes an homeomorphism between
$\mc{H}_i$ and $\mc{V}_i$. In particular,
$$
   \Phi(\mc{H}_i)=\mc{V}_i\;\;\hbox{and}\;\;\Phi^{-1}(\mc{V}_i)=\mc{H}_i\;\; \hbox{for each}\;\;
   i\in\{0,1\}.
$$
This feature is pivotal in the next construction. The intersection (in $\mc{Q}$) of the \emph{vertical squares} $\mc{V}_i$ with the \emph{horizontal squares} $\mc{H}_j$, $i, j\in\{0,1\}$, generates $2^2=4$ topological squares in $\mc{Q}$, namely
\[
\mc{Q}_{(i,j)}:=\mc{V}_i\cap\mc{H}_j,\qquad i,j\in\{0,1\},
\]
which have been represented in Figure \ref{fig-xvi}.
\begin{figure}[h!]
\centering
\includegraphics[scale=0.9]{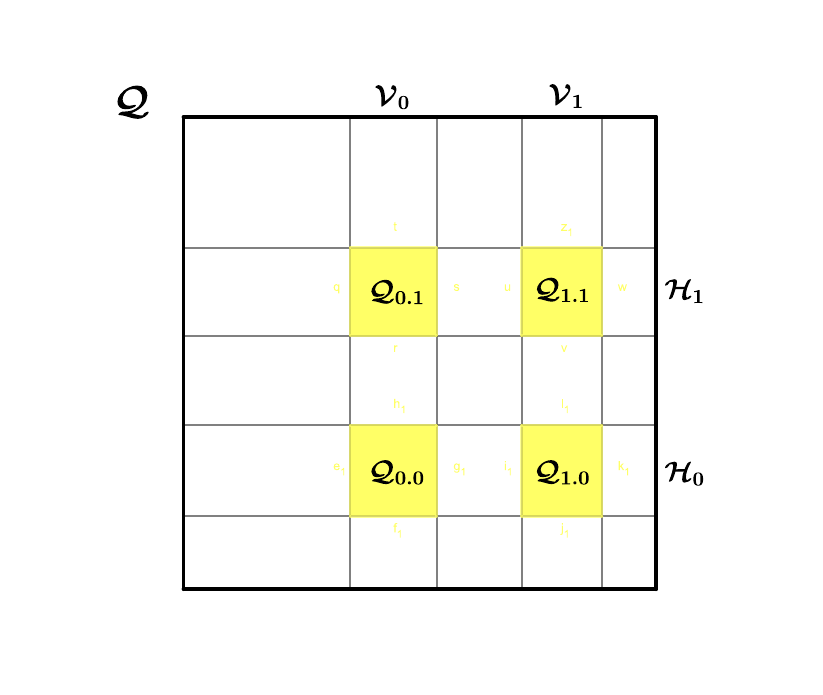}
\caption{The invariant squares $\mc{Q}_{(i,j)}$, $i, j\in\{0,1\}$.}
\label{fig-xvi}
\end{figure}
\par

Subsequently, we will denote $\Phi^0$ the identity map.
By construction, for every $s_{-1}, s_0\in\{0,1\}$ and
$$
  z\in\mc{Q}_{(s_{{-\!1}},s_0)}=\mc{V}_{s_{{-\!1}}}\cap\mc{H}_{s_0},
$$
we have that
\begin{equation}
\label{iii.27}
\Phi^0(z)=z\in\mc{H}_{s_0}\;\;\hbox{and}\;\;\Phi^{-1}(z)\in\Phi^{-1}(\mc{V}_{s_{{-\!1}}})=
\mc{H}_{s_{-\!1}}.
\end{equation}
In other words,
\begin{equation}
\label{iii.28}
  \Phi^0(\mc{Q}_{(s_{{-\!1}},s_0)})=\mc{Q}_{(s_{{-\!1}},s_0)}\subset \mc{H}_{s_0},\quad
  \Phi^{-1}(\mc{Q}_{(s_{{-\!1}},s_0)})\subset \mc{H}_{s_{-\!1}}.
\end{equation}
Naturally, the previous \emph{duplication process} can be repeated for each of the topological
squares $\mc{V}_i$, $i\in\{0,1\}$. Much like $\mc{Q}$, for each $i\in\{0,1\}$, the upper and lower sides of  $\mc{V}_i$ consist of two arcs of trajectory of $\G(\ell_2)$ and $\G(\ell_1)$, respectively. Thus,
replacing $\mc{Q}$ by $\mc{V}_i$ we can generate the four vertical topological squares
\begin{equation}
\label{iii.29}
\mc{V}_{(i,j)} := \Phi(\mc{V}_{i})\cap\mc{V}_j,\qquad i, j\in\{0,1\}.
\end{equation}
Then, as in $\mc{Q}$, $\mc{V}_0$ and $\mc{V}_1$, for every $i, j\in\{0,1\}$, $\mathcal{V}_{(i,j)}$ provides us with a vertical topological square linking  $\Gamma(\ell_1)$ to $\Gamma(\ell_2)$; vertical in the sense that their upper and lower sides consist of certain arcs of trajectory of $\G(\ell_2)$ and
$\G(\ell_1)$, respectively. According to \eqref{iii.29}, it becomes apparent that, for every $i, j\in\{0,1\}$,
\begin{equation}
\label{iii.30}
\begin{split}
\Phi^{-1}(\mc{V}_{(i,j)}) & =\Phi^{-1}(\Phi(\mc{V}_{i})\cap\mc{V}_j)\subsetneq \Phi^{-1}(\mc{V}_{j})=\mc{H}_j,\\
\Phi^{-2}(\mc{V}_{(i,j)}) & = \Phi^{-2}(\Phi(\mc{V}_{i})\cap\mc{V}_j)\subsetneq\Phi^{-2}(\Phi(\mc{V}_{i}))= \Phi^{-1}(\mc{V}_{i})=\mc{H}_i.
\end{split}
\end{equation}
Similarly, we can duplicate the horizontal squares by setting
\begin{equation}
\label{iii.31}
\mc{H}_{(i,j)} :=\Phi^{-1}(\mc{H}_{j})\cap\mc{H}_i,\qquad i, j\in\{0,1\}.
\end{equation}
By construction, for every $i, j\in\{0,1\}$, $\mc{H}_{(i,j)}$ is an horizontal topological square connecting the lateral sides of $\mc{Q}$, i.e.,
linking $\{(x,y)\in\mc{Q}:x=0\}$ to  $\{(x,y)\in\mc{Q}:x=x_1\}$. By \eqref{iii.31}, we have that, for every $i, j\in\{0,1\}$,
\begin{equation}
\label{iii.32}
\begin{split}
 \Phi^0(\mc{H}_{(i,j)}) & = \mc{H}_{(i,j)} =\Phi^{-1}(\mc{H}_{j})\cap\mc{H}_i \subsetneq\mc{H}_{i}, \\
 \Phi(\mc{H}_{(i,j)}) & = \Phi(\Phi^{-1}(\mc{H}_{j})\cap\mc{H}_i)  \subsetneq
 \Phi(\Phi^{-1}(\mc{H}_{j})) = \mc{H}_{j}.
\end{split}
\end{equation}
Therefore, we can consider the $2^4=16$  topological squares in $\mc{Q}$ defined as
\begin{equation}
\label{iii.33}
\mc{Q}_{(s_{-2},s_{-1},s_0,s_1)}:=\mc{V}_{(s_{-2},s_{-1})}\cap\mc{H}_{(s_0,s_1)},\qquad s_{-2},s_{-1},s_0,s_1\in\{0,1\}.
\end{equation}
We claim that, for every $s_{-2}, s_{-1}, s_0, s_1\in\{0,1\}$,
\begin{equation}
\label{iii.34}
\Phi^{\k}(\mc{Q}_{(s_{-2},s_{-1},s_0,s_1)})\subsetneq \mc{H}_{s_{\kappa}}, \qquad \kappa\in\{-2,-1,0,1\}.
\end{equation}
Indeed, by \eqref{iii.33} and \eqref{iii.32}, we have that
$$
  \Phi( \mc{Q}_{(s_{-2},s_{-1},s_0,s_1)})\subsetneq \Phi(\mc{H}_{(s_0,s_1)}) \subsetneq \mc{H}_{s_1}.
$$
Thus, \eqref{iii.34} holds for $\kappa=1$. Moreover, by \eqref{iii.33} and \eqref{iii.31},
$$
  \Phi^0( \mc{Q}_{(s_{-2},s_{-1},s_0,s_1)})=\mc{Q}_{(s_{-2},s_{-1},s_0,s_1)}\subset
  \mc{H}_{(s_0,s_{-1})} \subset \mc{H}_{s_0},
$$
which establishes \eqref{iii.34} for $\kappa=0$.  Similarly, according to \eqref{iii.33} and \eqref{iii.30}
\begin{align*}
    \Phi^{-1}( \mc{Q}_{(s_{-2},s_{-1},s_0,s_1)})  \subsetneq \Phi^{-1}(\mc{V}_{(s_{-2},s_{-1})}) \subsetneq\mc{H}_{s_{-1}},
\end{align*}
and
\begin{equation*}
\Phi^{-2}( \mc{Q}_{(s_{-2},s_{-1},s_0,s_1)}) \subsetneq \Phi^{-2}(\mc{V}_{(s_{-2},s_{-1})}) \subsetneq
\mc{H}_{s_{-2}},
\end{equation*}
which shows \eqref{iii.34} for $\kappa=-1,-2$, which ends the proof of \eqref{iii.34}.
\par
The first picture of Figure \ref{fig-xvii} represents the sixteen topological squares
$\mc{Q}_{(s_{-2},s_{-1},s_0,s_1)}$, $s_i\in \{0,1\}$, $i\in \{-2,-1,0,1\}$, as defined by \eqref{iii.33},
while the second one  provides a magnification of $\mc{Q}_{(0,0)}$  where  the dyadic fractal behavior of this process can be appreciated.
\begin{figure}[h!]
\centering
\includegraphics[scale=0.7]{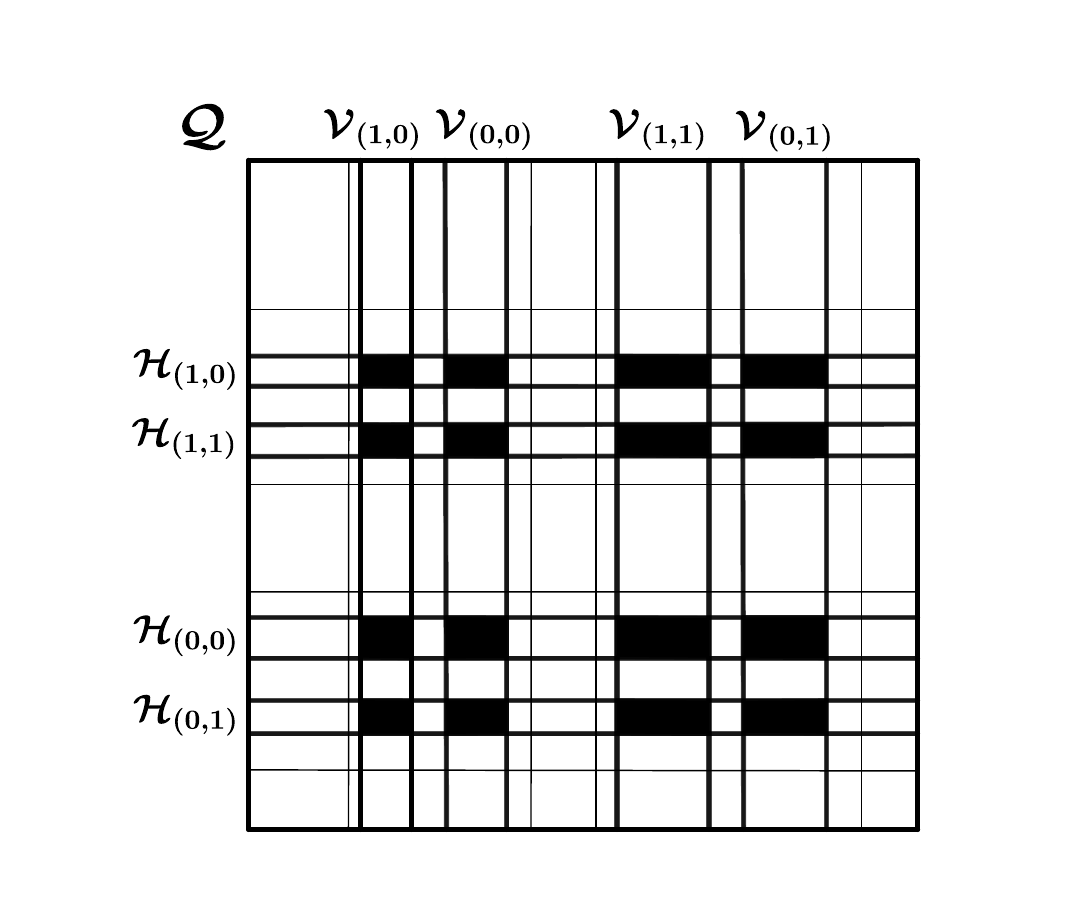}\qquad
\includegraphics[scale=0.2]{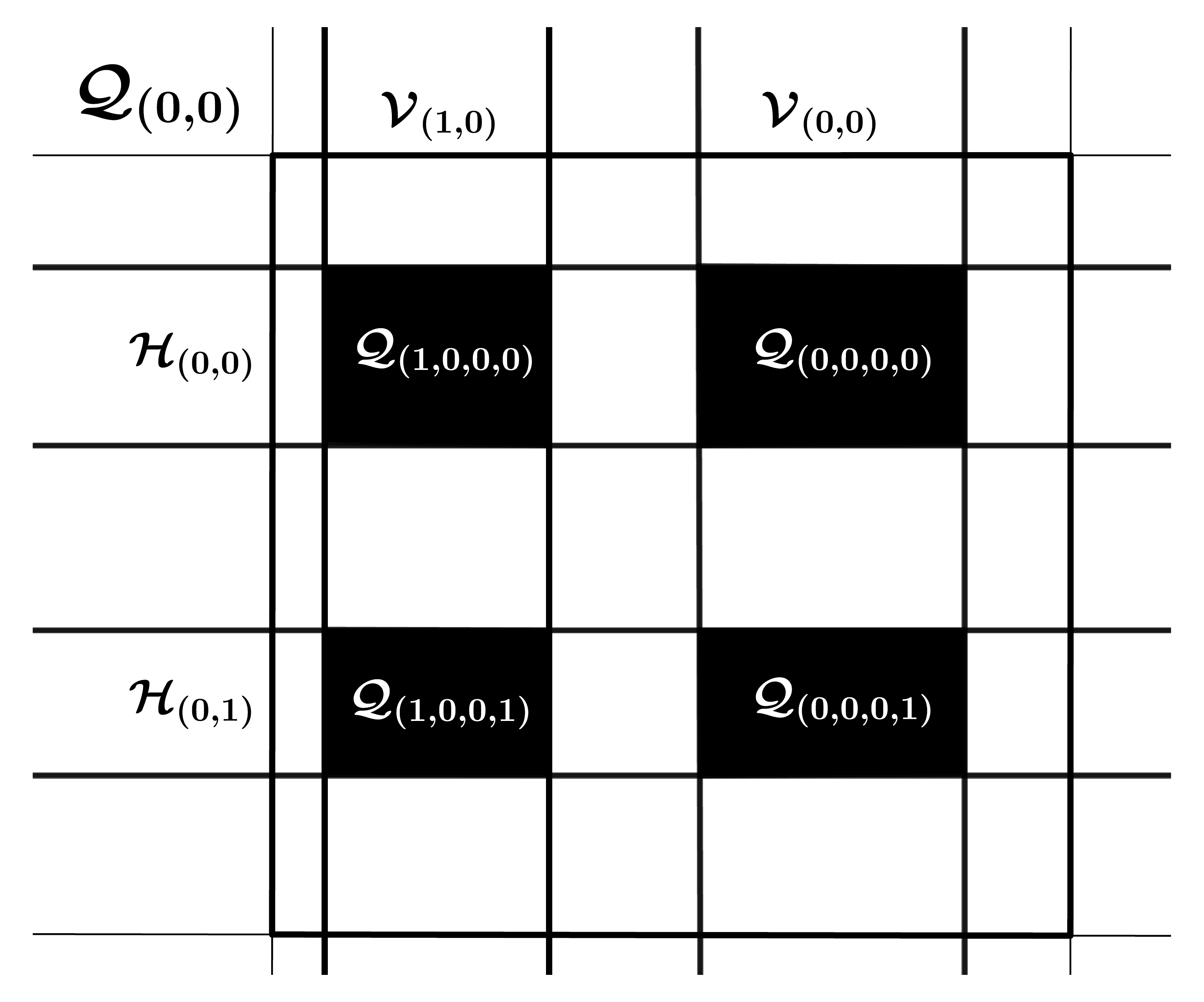}
\caption{The squares $\mc{Q}_{(s_{-2},s_{-1},s_0,s_1)}$ (left), and a zoom of $\mc{Q}_{(0,0)}$  (right).}
\label{fig-xvii}
\end{figure}
\par
Iterating $m$ times the previous process, it turns out that we can generate $2^m$ vertical rectangles. Namely, for $v_m:=(s_{-m},\ldots, s_{-2},s_{-1})\in\{0,1\}^m$,
\[
\mc{V}_{v_m}:= \Phi(\mc{V}_{(s_{-m},\ldots, s_{-2})})\cap\mc{V}_{s_{-1}}=\Phi(\stackrel{m}{\cdots}(\Phi(\mc{V}_{s_{-m}})\cap\mc{V}_{s_{-m+1}})\cdots)\cap \mc{V}_{-2})\cap\mc{V}_{-1}.
\]
Thus, by definition, for $\k\in\{-1,-2,\ldots,-m\}$,
\begin{equation}
\label{iii.35}
\Phi^{\k}(\mc{V}_{v_m})\subsetneq\Phi^{-1}(\mc{V}_{s_{\k}})=\mc{H}_{s_{\k}}.
\end{equation}
Similarly, there exists $2^m$ horizontal rectangles such that for $h_m=(s_{0},s_{1},\ldots, s_{m-1})\in\{0,1\}^m$,
\[
\mc{H}_{h_m}:= \Phi(\mc{H}_{(s_{1},\ldots, s_{m-1})})\cap\mc{H}_{s_{0}}=\Phi^{-1}(\stackrel{m-1}{\cdots}(\Phi^{-1}(\mc{H}_{s_{m-1}})\cap\mc{H}_{s_{m-2}})\cdots)\cap \mc{H}_{s_1})\cap\mc{H}_{s_0}.
\]
and, hence, for $\kappa\in\{0,1,\ldots,m-1\}$,
\begin{equation}
\label{iii.36}
\Phi^{\k}(\mc{H}_{h_m})\subsetneq\mc{H}_{s_{\k}}.
\end{equation}
Figure \ref{fig-xviii} shows the steps $m=2$ and $m=3$ of this
process for both the vertical and horizontal squares.

\begin{figure}[h!]
\centerline{\includegraphics[scale=0.7]{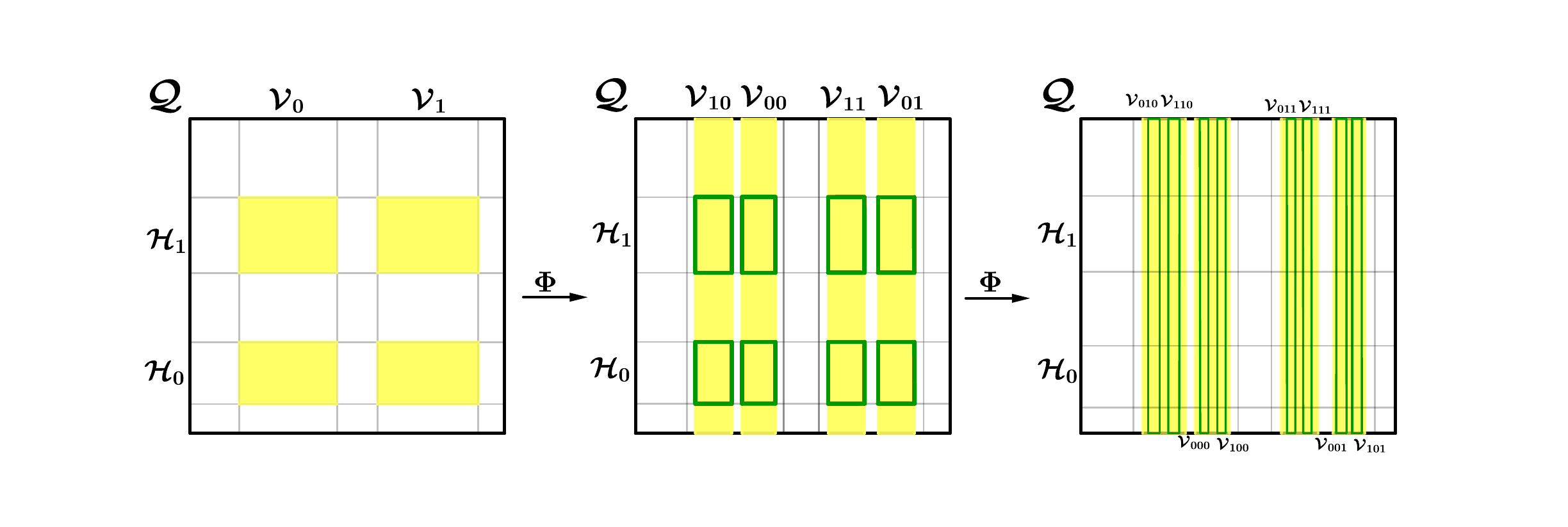}}
\centerline{\includegraphics[scale=0.7]{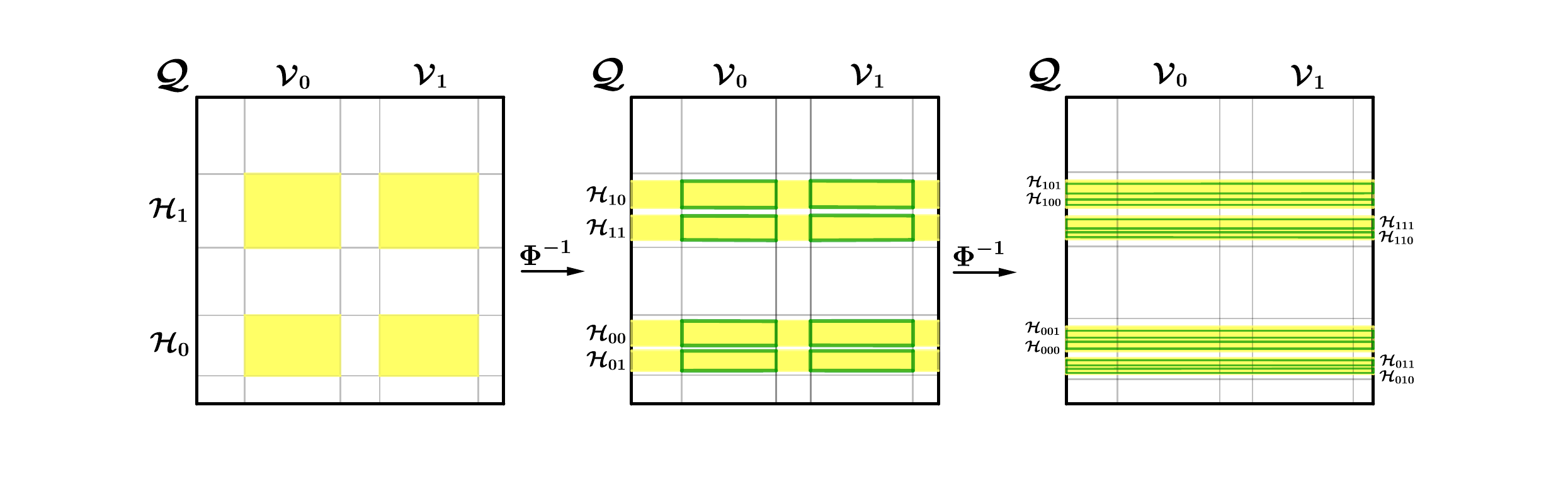}}
\caption{Dyadic fragmentation in horizontal and vertical lines.}
\label{fig-xviii}
\end{figure}
\par
Then, setting
$$
   \mc{Q}_{\Sigma_m}:=\mc{V}_{v_m}\cap\mc{H}_{h_m},
$$
where
$$
   \Sigma_m:= (v_m,h_m)=(s_{-m},\ldots,s_{-1},s_0,\ldots,s_{m-1})\in\{0,1\}^{2m},
$$
it becomes apparent that, by \eqref{iii.35} and \eqref{iii.36}, for every  integer $\kappa\in[-m,m-1]$,
\[
\Phi^{\k}(\mc{Q}_{\Sigma_m})\subset \mc{H}_{s_{\kappa}}.
\]
As we have constructed two sequences of nonempty compact, connected nested topological squares, $\mc{V}_{v_m}$, $\mc{H}_{h_m}$, by the Cantor principle,
\[
\mc{V}_{s_v}:=\lim_{m\to+\infty}\mc{V}_{v_m}\quad\hbox{and}\quad\mc{H}_{s_h}:=\lim_{m\to+\infty}\mc{H}_{h_m}
\]
are two nonempty continua with
$$
  s_v=(\ldots,s_{-2},s_{-1})\in\{0,1\}^\N,\quad   s_h=(s_0,s_1,\ldots)\in\{0,1\}^\N.
$$
Moreover, by the transversality assumptions,
the intersection $\mc{V}_{s_v}\cap\mc{H}_{s_h}$ is a point and $(s_v,s_h)\in\{0,1\}^\Z$. Furthermore,
\[
\Lambda:=\bigcup_{s_v,s_h\in\{0,1\}^\N}\mc{V}_{s_v}\cap\mc{H}_{s_h}
\]
is the invariant set of $\Phi$. Therefore,  thanks to the choice
of the labels $s_v$ and $s_h$ that we have done, it becomes
apparent that $\Phi$ is conjugated to the Bernoulli full shift in
two symbols.

We conclude this section with some geometrical and numerical
schemes illustrating the actual occurrence of the theoretical
horseshoe framework described above. To simplify the exposition,
we consider the case of an annular region under the effect of the
composition of a twist map with a vertical shear map. Even if the
geometry of the level lines associated with the predator-prey
system does not consist of circumferences, nonetheless this can be
assumed as a reasonable approximation if we consider small
orbits around the equilibrium point (as in \cite{Liu-1995}) or we
suppose to have performed an action-angle transformation leading to
an equivalent planar system where the radial component is
constant.

The following pictures describe the geometric effect of the composition
of a twist map $\Phi_{T_0}$ and a vertical shear map $\Phi_{T_1}$
on a domain in the first quadrant
which is defined like the set $\mc{Q}$ in Section \ref{sec3.1}, and denoted again by $\mc{Q}$. Actually,
we describe the effect of the maps and their composition $\Phi=\Phi_{T_1}\circ\Phi_{T_0}$,
with respect to the homeomorphic unit square
$[0,1]^2$, considered as a reference domain, where we can transfer the geometry.
Figure \ref{fig-xix} shows $\mc{Q}$, its image through a twist map
(with a sufficiently large twist, corresponding to a sufficiently large time $T_0$)
and the effect on the unit square. More precisely,
if we denote by $\eta=(\eta_1,\eta_2)$ the homeomorphism from the unit square to $\mc{Q}$,
the two bands in the third panel of Figure \ref{fig-xix} represent the
sets $[0,1]^2\cap (\hat{\eta}^{-1}\circ\Phi_{T_0}\circ \eta)^{-1}([0,1]^2)$ where $\hat{\eta}=(\eta_1,-\eta_2)$  is the homeomorphism from the unit square to
$\mc{R}$, which is the target set in the fourth quadrant symmetric to $\mc{Q}$. As a next step, Figure \ref{fig-xx} shows the effect on the vertical shear mapping on the
set $\mc{R}$ and we also show the set of points in the unit square which are mapped
into $\mc{Q}$, that is
$[0,1]^2\cap (\eta^{-1}\circ\Phi_{T_1}\circ \hat{\eta})^{-1}([0,1]^2)$.

\begin{figure}[h!]
\centering {\includegraphics[scale=0.3]{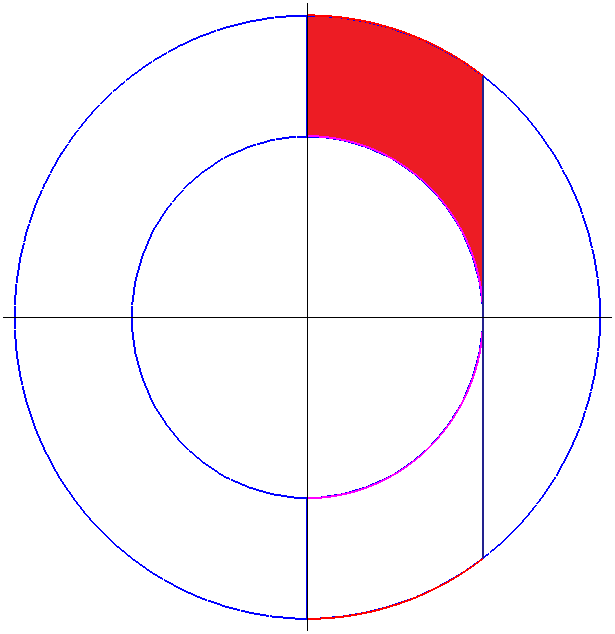}}\quad
{\includegraphics[scale=0.3]{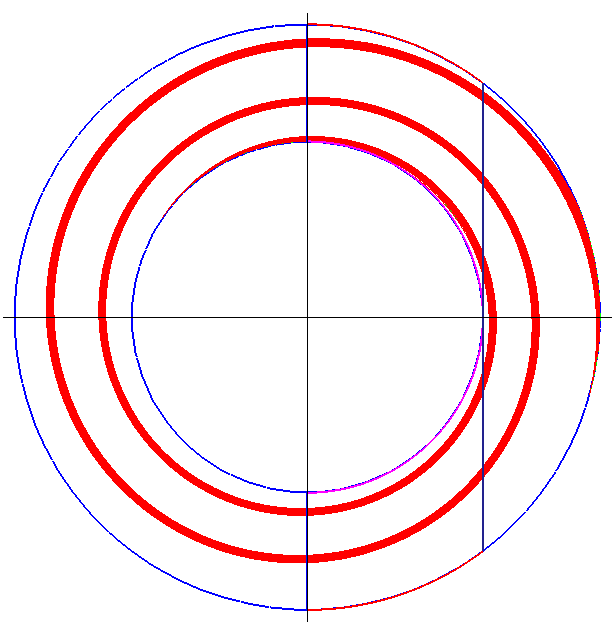}}\quad
{\includegraphics[scale=0.35]{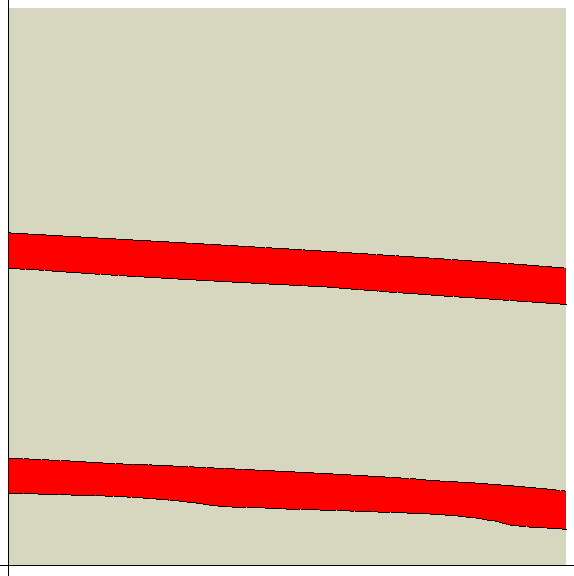}}
\caption{Example of an annular region under the action of a twist map. From left to right:
a rectangular domain $\mc{Q}$ in the first quadrant, its evolution by a twist map,
the part of the unit square (homeomorphic to $\mc{Q}$) which is mapped to the
region in the fourth quadrant which is symmetric to $\mc{Q}$.
} \label{fig-xix}
\end{figure}

\begin{figure}[h!]
\centering {\includegraphics[scale=0.5]{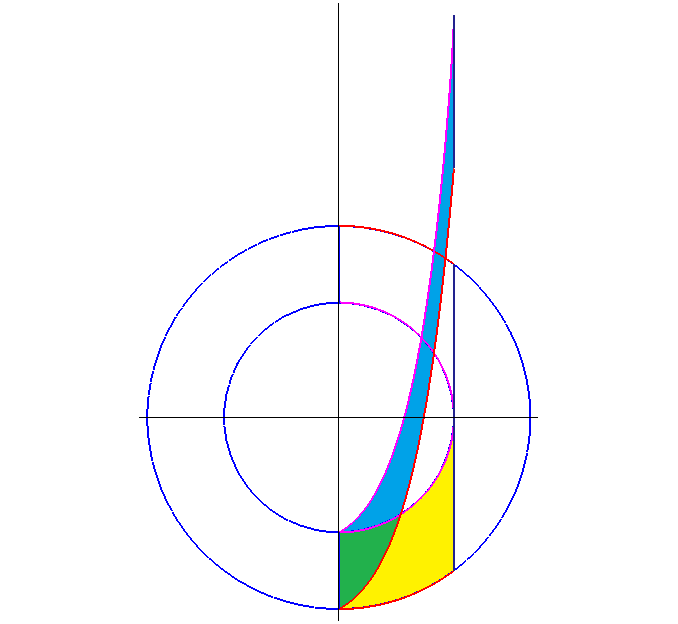}}\quad
{\includegraphics[scale=0.35]{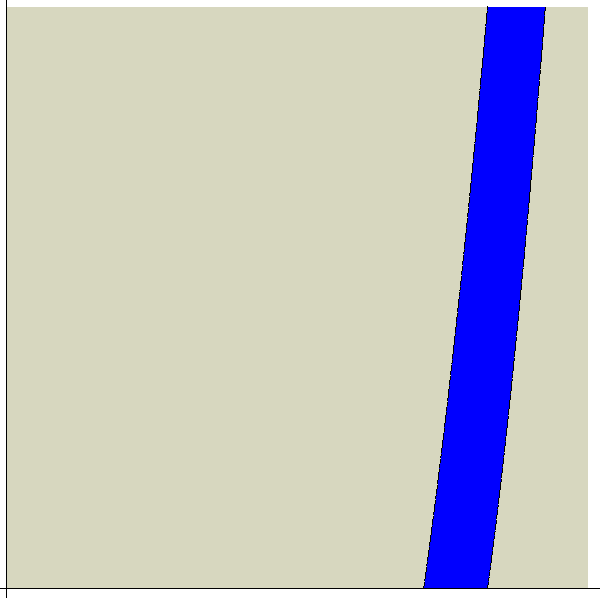}}
\caption{The region in the fourth quadrant (symmetric to $\mc{Q}$) and its transformation
by a vertical shear map as in Section \ref{sec3.1} (left panel). The right panel put
in evidence, with respect to the homeomorphic unit square, the part of the region
which arrives to $\mc{Q}$.} \label{fig-xx}
\end{figure}

Finally, Figure \ref{fig-xxi} puts in evidence the set of points of $\mc{Q}$,
represented in the unit square, which come back to $\mc{Q}$ after $\Phi=\Phi_{T_1}\circ\Phi_{T_0}$, that is $[0,1]^2\cap(\eta^{-1}\circ \Phi\circ \eta)^{-1}([0,1]^2).$
Obviously, this is a subset of the two bands domain appearing in the
third panel of Figure \ref{fig-xix}.

\begin{figure}[h!]
\centering {\includegraphics[scale=0.5]{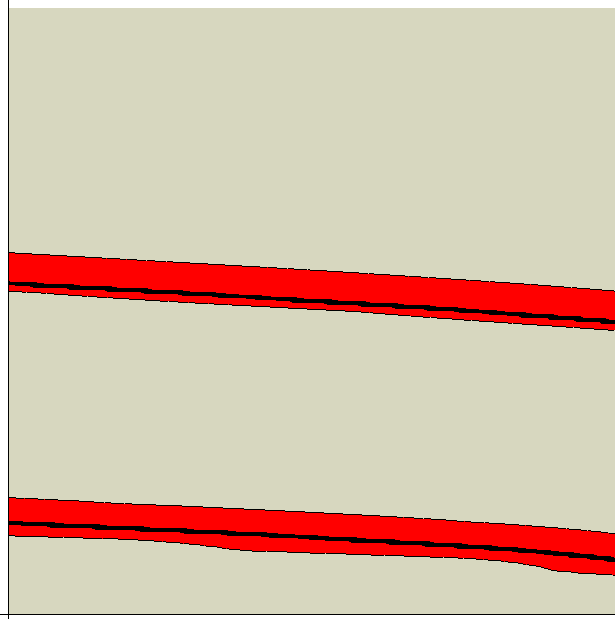}}
\caption{The final effect of the twist map and the vertical shear map, as viewed from the
unit square. The larger horizontal bands represent the set of points in $\mc{Q}$ which are moved to $\mc{R}$ under the action of the twist map.
The narrow darker bands represent the set of points which come back to $\mc{Q}$
after the application of the vertical shear map.} \label{fig-xxi}
\end{figure}

If we consider the inverse homeomorphism $\Phi^{-1}= \Phi_{T_0}^{-1}\circ \Phi_{T_1}^{-1}$
and look for the sets of points of $\mc{Q}$ which remain in the domain after
the first iteration, we can repeat, symmetrically the same argument as before.
As a first step, passing to the unit square, we will consider the set of the points
$w\in[0,1]^2$ such that $(\hat{\eta}^{-1}\circ\Phi_{T_1}^{-1}\circ \eta)(w)\in [0,1]^2$, which clearly equals $[0,1]^2\cap (\eta^{-1}\circ\Phi_{T_1}\circ \hat{\eta})^{-1}([0,1]^2)$, that is, the blue set appearing in the second panel of Figure \ref{fig-xx}. Next
we will consider the set of points of $\mc{R}$ which belong to $\mc{Q}$
after the action of the inverse twist $\Phi_{T_0}^{-1}$, which transferred to the unit square is $[0,1]^2\cap (\eta^{-1}\circ\Phi_{T_0}^{-1}\circ \hat{\eta})^{-1}([0,1]^2)$. This is exactly the set described in the third panel of Figure \ref{fig-xix}.
At the end, the set
$[0,1]^2\cap(\eta^{-1}\circ \Phi^{-1}\circ \eta)^{-1}([0,1]^2)$
will be made by two narrow bands inside the rectangular region in the second panel of
Figure \ref{fig-xx}.

From these numerical outcomes it is apparent that the ``real dynamics'' follows precisely the
abstract scheme of the Smale's horseshoe map.

\bibliographystyle{plain}
{\small{
\bibliography{LG-MH-Z4_biblio}
}}

\end{document}